%% file: wo-main.tex
\documentclass{fundam}
\usepackage{graphicx}
\usepackage[utf8]{inputenc}
\usepackage[T1]{fontenc}
\usepackage{longtable}
\usepackage{rotating}
\usepackage{textcomp}
\usepackage{version}
\usepackage{hyperref}
\usepackage{amsmath}
\usepackage{amsfonts}
\usepackage{amssymb,dsfont}
\usepackage{xcolor}
\usepackage[inline, shortlabels]{enumitem}
\usepackage{appendix}
\usepackage[english,noabbrev,capitalise,nameinlink]{cleveref}
\usepackage{tikz}
\usetikzlibrary{automata,positioning}
\usetikzlibrary{arrows,automata}
\usetikzlibrary{shapes,snakes}
\usetikzlibrary{fit}
\usepackage{caption}
\usepackage{subcaption}
\usepackage{MnSymbol,wasysym}
\usepackage{ stmaryrd }
\usepackage{environ}
\usepackage{float}
\usepackage{tabularx}
\usepackage{environ}
\usepackage{lipsum}
\usepackage{amssymb}
\usepackage{circuitikz}
\usetikzlibrary{calc}

\usetikzlibrary{shapes.geometric,arrows}

\crefname{lemma}{Lemma}{Lemmas}
\Crefname{lemma}{Lemma}{Lemmas}

\usepackage[textsize=small]{todonotes}

\include{macros}

\include{macros-th}

\include{macros-b}

\usepackage{soul}
\usepackage{comment}
\usepackage{framed}
\usepackage{mathtools}
\usepackage{booktabs} % for \toprule, \midrule, \cmidrule & \bottomrule macros
\usepackage{cite}
   \usepackage{thmtools, thm-restate}

\newif\ifappendix

\appendixtrue
\begin{document}
 
\title{Safety Verification of Wait-Only Non-Blocking Broadcast Protocols}
\author{Lucie Guillou\\
	MPI-SWS\\
	Kaiserslautern, Germany
	\and Arnaud Sangnier\\
	DIBRIS, Universit\`a di Genova \\
	Genova, Italy
\and Nathalie Sznajder \\
LIP6, CNRS, Sorbonne Universit\'e\\
Paris, France}

\runninghead{L. Guillou, A. Sangnier, N. Sznajder}{Safety Verification of Wait-Only Networks}

\setcounter{page}{1}
\publyear{}
\papernumber{}
\volume{}
\issue{}

\maketitle  

\begin{abstract}
Broadcast protocols are programs designed to be executed by networks of processes. Each process runs the same protocol, and communication between them occurs synchronously in two
ways: \emph{broadcast}, where one process sends a message to all others, and \emph{rendez-vous}, where one process sends a message to at most one other process. In both cases,
communication is \emph{non-blocking}, meaning the message is sent even if no process is able to receive it. We consider two coverability problems: the state coverability problem asks whether there exists a number of processes that allows reaching a given state of the protocol, and the configuration coverability problem asks whether there exists a number of processes that allows covering a given
configuration. These two problems are known to be decidable and Ackermann-hard. We show that when the protocol is \emph{Wait-Only} (i.e., it has no state from which a process can both send and 
receive messages), these problems become P-complete and \pspace-complete, respectively.
%	We study networks of processes that all execute the same finite protocol and communicate synchronously in two different ways: a process can broadcast one message to all other processes or send it to at most one other process. In both cases, if no process can receive the message, it will still be sent.
%	We establish a precise complexity class for two coverability problems with a parameterised number of processes: the state coverability problem and the configuration coverability problem. It is already known that these problems are Ackermann-hard (but decidable) in the general case. We show that when the protocol is \emph{Wait-Only}, i.e., it has no state from which a process can send and receive messages, the complexity drops to P and \pspace, respectively.
%		
\end{abstract}

\input{intro}
\input{model}

\input{prelim}
\input{copypaste}

\input{scover}

\input{ccover}

\input{waitonlyb}

\input{conclusion}

\bibliography{wo-biblio}

\clearpage
\appendix
\input{annex}

\end{document}

%% file: macros.tex
 \newtheorem{observation}{Observation}
\newenvironment{sketch-proof} {\textit{Sketch of proof.}}{\hfill$\square$\\ }
\newenvironment{proofof}[1]{\textit{Proof of #1.}}{\hfill$\square$\\ }
 \newcommand{\PP}{\ensuremath{P}}   
\newcommand{\CC}{\ensuremath{\mathcal{C}}}
\newcommand{\Cinit}{\ensuremath{\mathcal{I}}}
\newcommand{\CCprot}[1]{\ensuremath{\CC_{#1}}}
\newcommand{\Cinitprot}[1]{\ensuremath{\Cinit_{#1}}}
\renewcommand{\SS}{\ensuremath{\mathcal{S}}}
\newcommand{\activeset}[1]{\ensuremath{#1_A}}
\newcommand{\waitingset}[1]{\ensuremath{#1_W}}

\newcommand{\nat}{\mathbb{N}}
\newcommand{\set}[1]{\{#1\}}
\newcommand{\mset}[1]{\Lbag #1 \Rbag}

\newcounter{sarrow}

\newcommand{\trans}{\rightarrow}
\newcommand{\transup}[1]{\xrightarrow{#1}}
\newcommand{\abtransup}[1]{\xRightarrow{#1}}
\newcommand{\abtrans}{\Rightarrow}
\newcommand{\abtransStep}{\Rightarrow_{\textsf{step}}}
\newcommand{\abtransExt}{\Rightarrow_{\textsf{ext}}}
\newcommand{\abtransSwitch}{\Rightarrow_{\textsf{switch}}}
\newcommand{\abtransStepUp}[1]{\xRightarrow{#1}_{\textsf{step}}}
\newcommand{\abtransExtUp}[1]{\xRightarrow{#1}_{\textsf{ext}}}
\newcommand{\abtransSwitchUp}[1]{\xRightarrow{#1}_{\textsf{switch}}}
\newcommand{\abtransWUp}[1]{\overset{#1}{\Rightarrow}_{\kappa}}
\newcommand{\abtransExtUpSmall}[1]{\overset{#1}{\Rightarrow}_{\textsf{ext}}}
\newcommand{\abtransSwitchUpSmall}[1]{\overset{#1}{\Rightarrow}_{\textsf{switch}}}
\newcommand{\abtransStepUpSmall}[1]{\overset{#1}{\Rightarrow}_{\textsf{step}}}

\newcommand{\SCover}{\textsc{StateCover}}
\newcommand{\CCover}{\textsc{ConfCover}}

\newcommand{\pspace}{\textsc{PSpace}}

\newcommand{\qinit}{\ensuremath{q_{\textit{in}}}}
\newcommand{\gammainit}{\ensuremath{\gamma_{\textit{in}}}}
\renewcommand{\set}[1]{\ensuremath{\{#1\}}}

\newcommand{\Act}{\ensuremath{\mathsf{Act}}}
\newcommand{\Op}{\ensuremath{\mathsf{Op}}}

\makeatletter
 \newcommand{\problemtitle}[1]{\gdef\@problemtitle{#1}}% Store problem title
\newcommand{\probleminput}[1]{\gdef\@probleminput{#1}}% Store problem input
\newcommand{\problemquestion}[1]{\gdef\@problemquestion{#1}}% Store
\NewEnviron{decproblem}{
  \problemtitle{}\probleminput{}\problemquestion{}% Default input is empty
  \BODY% Parse input
  \par\addvspace{.2\baselineskip}
  \noindent
  \fbox{
    
    \centering
  \begin{tabular}{rl}
    \multicolumn{2}{l}{ \@problemtitle} \\% Title
    \hline
    \rule{0pt}{3ex}
    \normalsize \textbf{Input:} &  \@probleminput \\% Input
    \normalsize \textbf{Question:} &  \@problemquestion% Question
  \end{tabular}
    }
  \par\addvspace{.2\baselineskip}
   
}

\newcommand{\interp}[1]{\llbracket #1 \rrbracket}

\newcommand{\mtrans}{\rightarrow}

\newcommand{\rdvprot}{RDV protocol}
\newcommand{\mtransUp}[1]{\xrightarrow{#1}}
\newcommand{\mtransup}[1]{\xrightarrow{#1}}
\newcommand{\mtransse}[1]{\xrightarrow{\textsf{se}, #1}}
\newcommand{\mtransrdv}[1]{\xrightarrow{\textsf{rdv}, #1}}
\newcommand{\Rec}[1]{R(#1)}
\newcommand{\recfrom}[1]{R(#1)}
\newcommand{\mconf}{\ensuremath{{C}}}
\newcommand{\mconfs}{\mathcal{C}}
\newcommand{\mconfsInit}{\mathcal{C}_{init}}
\newcommand{\mconfsinit}{\mathcal{C}_{init}}
\newcommand{\mconfInit}{{C}_{init}}
\newcommand{\Interp}[1]{\llbracket #1 \rrbracket}
\newcommand{\staterec}[1]{Q_{?}(#1)}
\newcommand{\ie}{i.e.}
\newcommand{\statecovernb}{\SCover}
\newcommand{\confcovernb}{\CCover}

\newcommand{\constantM}{\ensuremath{X}}

%% file: macros-th.tex
\NewEnviron{decproblem2}{
  \problemtitle{}\problemquestion{}% Default input is empty
  \BODY% Parse input
  \par\addvspace{.2\baselineskip}
  \noindent
  \fbox{
    
   \centering
  \begin{tabular}{rl}
    \multicolumn{2}{l}{\normalsize \@problemtitle} \\% Title
    \hline
    \rule{0pt}{3ex}
    \normalsize \textbf{Question:} & \normalsize \@problemquestion% Question
  \end{tabular}
    }
  \par\addvspace{.2\baselineskip}
   
}

\newcommand{\vect}[1]{\mathtt{\textbf{#1}}}

\newrobustcmd{\Afffuncs}[1]{\kl[affine function]{\mathsf{Aff}_{#1}}}
\newrobustcmd{\afffunc}{\kl[affine function]{(\vect{A},\vect{b})}}
    
\newrobustcmd{\Trfuncs}[1]{\kl[affine function]{\mathsf{Trsl}_{#1}}}

\newrobustcmd{\size}[1]{\mathtt{size}(#1)}

\newrobustcmd{\guards}[1]{\kl[guard]{\mathtt{G}(#1)}}
\newrobustcmd{\pguards}[1]{\kl[guard]{\mathtt{PG}(#1)}}
\newrobustcmd{\aguard}{\kl[guard]{\mathtt g}}
\newrobustcmd{\asys}{\kl[affine counter system]{S}}

\newrobustcmd{\broadprot}{\mathit{BP}}
    
\newrobustcmd{\source}[1]{\kl[\source]{\mathit{source}(#1)}}
\newrobustcmd{\target}[1]{\kl[\target]{\mathit{target}(#1)}}
\newrobustcmd{\guard}[1]{\kl[\guard]{\mathit{guard}(#1)}}
\newrobustcmd{\update}[1]{\kl[\update]{\mathit{update}(#1)}}
\newrobustcmd{\ACS}{\mathcal{ACS}}
\newrobustcmd{\KS}{\mathcal{KS}}
\newrobustcmd{\TCS}{\mathcal{TCS}}
\newrobustcmd{\OCS}{\mathcal{OCS}}
\newrobustcmd{\VASS}{\mathcal{VASS}}
\newrobustcmd{\FMACS}{\mathcal{FMACS}}

\newrobustcmd{\transsysof}[1]{\kl[\transsysof]{\mathfrak{T}(#1)}}
\newrobustcmd{\rbtranssysof}[2]{\kl[\transsysof]{\mathfrak{T}_{#2}(#1)}}
\newrobustcmd{\Net}{N}


\newrobustcmd{\arun}{\rho}

\newrobustcmd{\Reach}{\textsc{CS-Reach}\xspace}
\newrobustcmd{\ControlReach}{\textsc{CS-ControlReach}\xspace}
\newrobustcmd{\ControlReapReach}{\textsc{CS-RepControlReach}\xspace}
\newrobustcmd{\ParamControlReach}{\textsc{Param-CS-ControlReach}\xspace}
\newrobustcmd{\ParamControlReapReach}{\textsc{Param-CS-RepControlReach}\xspace}
\newrobustcmd{\PCover}{\textsc{PN-Coverability}\xspace}
\newrobustcmd{\PReach}{\textsc{PN-Reach}\xspace}
\newrobustcmd{\CSMC}[1]{\textsc{CS-Model-Checking}(#1)\xspace}
\newrobustcmd{\Consistency}{\textsc{Consistency}\xspace}
\newrobustcmd{\InterNonEmpty}[1]{\textsc{Intersection-NonEmptiness}(#1)\xspace}
\newrobustcmd{\Membership}[1]{\textsc{Membership}(#1)\xspace}
\newrobustcmd{\AtwoAnonempty}{\textsc{A2A-NonEmptiness}\xspace}
\newrobustcmd{\ParityVass}{\textsc{Parity-VASS-Game}\xspace}
\newrobustcmd{\SReach}{\textsc{Sure-VASS-MDP-Reach}\xspace}
\newrobustcmd{\ASReach}{\textsc{AlmostSure-VASS-MDP-Reach}\xspace}
\newrobustcmd{\LSReach}{\textsc{LimitSure-VASS-MDP-Reach}\xspace}
\newrobustcmd{\SReapReach}{\textsc{Sure-VASS-MPD-RepReach}\xspace}
\newrobustcmd{\ASReapReach}{\textsc{AlmostSure-VASS-MPD-RepReach}\xspace}
\newrobustcmd{\LSReapReach}{\textsc{LimitSure-VASS-MPD-RepReach}\xspace}
\newrobustcmd{\AHNControlReach}{\textsc{AHN-ControlReach}\xspace}
\newrobustcmd{\AHNTarget}{\textsc{AHN-Target}\xspace}
\newrobustcmd{\TAHNControlReach}{\textsc{TAHN-ControlReach}\xspace}
\newrobustcmd{\RBNCountingReach}{\textsc{RBN-CountingReach}\xspace}
\newrobustcmd{\RBNLocalControlReach}{\textsc{RBN-Local-ControlReach}\xspace}
\newrobustcmd{\RBNLocalTarget}{\textsc{RBN-Local-Target}\xspace}
\newrobustcmd{\PRBNControlReach}[2]{\textsc{PRBN-ControlReach}$^{#1}_{#2}$\xspace}
\newrobustcmd{\ParityRBNGame}{\textsc{Parity-RBN-Game}\xspace}
\newrobustcmd{\SafeParityRBNGame}{\textsc{SafeParity-RBN-Game}\xspace}
\newrobustcmd{\DRBNCover}{\textsc{DRBN-Cover}\xspace}

\usetikzlibrary{arrows,shapes,decorations,automata,backgrounds,fit,matrix}
\usetikzlibrary{petri}
\tikzstyle{init-left}=[pin={[pin distance=7pt,pin
  edge={<-,black,thick},pin position=left]}]

\tikzstyle{pstate}=[rectangle, draw=black, rounded corners]
\tikzstyle{petri-p}=[circle,thick,inner sep=0pt,minimum size=5mm]
\tikzstyle{petri-t}=[rectangle,thick,inner sep=0pt,minimum width=7mm,minimum height=1mm]
\tikzstyle{petri-t2}=[rectangle,thick,inner sep=0pt,minimum width=1mm,minimum height=7mm]
\tikzstyle{petri-tok}=[circle,inner sep=0pt,minimum size=4pt, color=black,fill=black]
\tikzstyle{gate}=[rectangle,thick,inner sep=0pt,minimum size=6mm,draw=black]
\tikzstyle{nstate}=[rectangle, line width=1mm,draw=gray, rounded
corners,inner sep=0pt,minimum size=1cm]
\tikzstyle{nstatep}=[rectangle, line width=1mm,draw=gray, rounded corners,inner sep=4pt,minimum size=1cm]
\tikzstyle{nstate2}=[circle, line width=1mm,draw=gray,inner sep=0pt,minimum size=1cm]
\tikzstyle{nstated}=[rectangle, line width=1mm,draw=gray, rounded
corners,inner sep=4pt,minimum size=1cm]

%% file: macros-b.tex
\usetikzlibrary{calc,decorations.pathmorphing,shapes}

\renewcommand{\SS}{\ensuremath{\mathcal{S}}}

\makeatother

\newcommand{\arrowa}[1]{\ensuremath{\Rightarrow}}
\newcommand{\arrowP}[1]{\ensuremath{\xrightarrow{}_{#1}}}

\newcommand{\arrowPa}[1]{\ensuremath{\Rightarrow_\PP}}

\newcommand{\iflong}[1]{}

\newcommand{\Toks}{\ensuremath{\textit{Toks}}}

\newcommand{\mst}{\mathsf{st}}

\newcommand{\starg}[1]{\mst(\mathit{{\kern-1pt}#1})}

%% file: intro.tex
\section{Introduction}

Verification of distributed systems present specific challenges compared to centralized systems. Indeed, the concurrent behavior of the different entities of the 
system induces a lot of interleavings, making the whole system very difficult to analyze. This is specifically the case when the number of components of the system is known in
advance. Then, verification, and in particular model-checking of such systems amount to verifying a centralized system: one models each component separately
and builds the global system by defining synchronisations between the components. This leads to a single, potentially huge, system to check against a specification. 
Developping techniques that can tackle this state space explosion have been an active subject of research in the last decades~\cite{25yearsGrumbergVeith}.

In some applications, however, the number of processes is not known in advance; we refer to these as \emph{parameterized systems} (where the parameter
is the number of processes). This is typically true for protocols like cache-coherency, or bus protocols. A lot
of distributed algorithms are also designed to work for any number of participants. It is then impossible to verify the model for every possible
number of processes and new techniques have to be used. Apt and Kozen~\cite{apt86limits} have shown that verification of parameterized systems is undecidable in general, even when the processes follow programs modelled by finite-state machines. However, two ingredients are crucial in their proof: \emph{identities} of the agents, and the \emph{structure} of communication. Hence, one can recover
decidability in models where the agents are symmetric (they all execute the same code) and where the communication is restricted. For instance,
in token-passing systems on a ring, in which a single valueless token is passed along a ring between the participants that all implement the same program, Emerson and
Najmoshi~\cite{EmersonN03} have proved that the verification of stuttering invariant specifications is decidable. Their proof relies on the notion of \emph{cut-off}:  it is sufficient 
to check the property for a small number of processes to prove it for any number of processes. 

Hague~\cite{Hague11} studied parameterized model-checking in the case of asynchronous communication: via a shared register accessible without a lock. 
This weak means of communication allows to obtain decidability of the control state reachability problem, even when one process is identified as a master
process following a specific program, and with programs modeled by pushdown systems. Interestingly, the problem becomes undecidable when the number of 
processes is fixed. %add esparza cav 13?
%asynchronous communication like shared registers

In guarded protocols~\cite{EmersonK03}, the system under consideration consists of one distinguished control process, and arbitrarily many identical user 
processes, that evolve concurrently. Communication in these systems is done through boolean guards labelling the transitions of the protocol followed by the
individual processes. Then, guards consist of boolean combinations of constraints over local states of other processes, that allow a process to take a transition or
not. Verifying stuttering invariants properties is decidable when one restricts the guards to be only conjunctive or only disjunctive guards~\cite{EmersonK00}, again relying
 on a cut-off argument. 
%The model of Emerson and Kahlon is very close to other models of synchronous communications.
 A seminal study by German and Sistla~\cite{german92}
explored models where an arbitrary number of identical processes evolve concurrently, and communicate by rendez-vous, a mechanism
that synchronizes two processes: the one that sends a message, and any other process in a state where it can receive the message. They synchronously
evolve by taking their respective actions. If no process is ready to receive the message, the sending action cannot be taken. 
They show that if the systems studied consist in
one single control process and arbitrarily many user processes, it is decidable to check whether all the executions of a process satisfy an LTL specification. Moreover,  when the system
consists only in identical user processes, they show that model-checking becomes polynomial in the size of the protocol and exponential in the size of the formula. This drop in complexity
is due to a property of such systems that we call ``copycat property'': if a state is coverable by one process, it can be covered by an arbitrarily large number of processes. This property allows efficient saturation algorithms. In the same model, several works~\cite{horn20deciding,bala21finding} studied the decidability for the existence of a cut-off for the problem of 
checking if all the processes can end in the same state of the procotol. In a similar way, \cite{esparza-verif-lics99} studied the problem introduced in~\cite{EmersonN98} where many 
identical processes can also communicate by broadcast: when a message is sent, it is received by all the processes (able to receive it). They show that model-checking safety properties is decidable, relying on well quasi order theory, while
model-checking liveness properties is undecidable. In between rendez-vous and broadcast lies another means of communication that we call non-blocking rendez-vous. This communication mechanism, motivated by Java Threads programming, involves at most two processes: when a process sends a message, it is received by
at most one process ready to receive the message, and both processes jointly change their local state. However, when no process is ready to receive the message,
the message is sent anyway and lost, and only the sender changes its local state. A model in which both broadcast and non-blocking rendez-vous are allowed is useful to analyze behaviours
of Java Threads: when a Thread is suspended in a waiting state, it can be woken up upon the reception of a \texttt{notify} message sent by another
Thread, but the sender is not blocked if no Thread is suspended; it simply continues its execution and the \texttt{notify} message is lost, like in a non-blocking rendez-vous. When a Thread sends a \texttt{notifyAll} message that will be received by all the suspended Threads waiting for that message, it is modelled by the broadcast mechanism. Observe that broadcast is also a non-blocking means of communication (the absence of processes to receive the message does not prevent the message to be sent). This model has been studied first in~\cite{delzanno-towards-tacas02}, where safety properties are proven to be decidable. In their paper on Guarded Protocols~\cite{EmersonK03}, Emerson and Kahlon showed that disjunctive guards were equivalent to the rendez-vous defined in~\cite{german92}, and that non-blocking rendez-vous were strictly more expressive than
rendez-vous. Moreover, they show that broadcast is strictly more expressive than non-blocking rendez-vous. 

In this paper, we focus on the model where both broadcast and non-blocking rendez-vous are allowed, model that we call non-blocking broadcast protocols. We investigate the complexity
of the coverability problem, which is a safety property: Given a non-blocking broadcast protocol, and a configuration of this protocol, is there a number of processes that allow to cover
this configuration? If the configuration given is the basis of an upward-closed set of \emph{bad} configurations, this problem amounts to checking whether this set of configurations will never
be reached, hence it is a safety property. This problem is known to be decidable~\cite{delzanno-towards-tacas02,EmersonK03} and 
Ackermann-hard~\cite{schmitz-power-concur13,esparza-verif-lics99,aminof-expressive-lpar15}. When considering protocols allowing only non-blocking
rendez-vous, the complexity of this problem is then \textsc{ExpSpace}-complete~\cite{guillou-safety-concur23} (note the complexity gap with the rendez-vous of German and Sistla~\cite{german92} where this problem is in P).
In~\cite{guillou-safety-concur23}, we have introduced a syntactic restriction of the protocols, namely Wait-Only protocols, 
in which there is no state from which a process can both send and receive a message. The state space of such protocols is then partitioned between \emph{action states}, from which
a process can (only) send messages, and \emph{waiting states} in which a process can (only) receive messages. This is a natural restriction when one models Java Threads as mentioned above: indeed,
when a Thread is waiting for a \texttt{notify} or \texttt{notifyAll} message, it is suspended, and will not perform any action. In such protocols, it will then never happen than from
a given state, a process can both receive a message (typically to be waken up) and send a message, i.e., perform an action on its own.  

\paragraph{Our contributions.}
In this paper, we present a complete analysis of the coverability problem for Wait-Only non-blocking broadcast protocols, that will for now on be simply called Wait-Only protocols for simplicity. While the copycat property holds in rendez-vous protocols, it does not hold when considering non-blocking rendez-vous. This is why we show that Wait-Only protocols exhibit another handy property, that
we call \emph{copypaste} property: it states that whenever an \emph{action} state is coverable, it can be reached by a number of processes arbitrarily high (reminiscent of the said
copycat property). Moreover, it states that whenever a set of action states and a \emph{waiting state} are coverable (separately), we can guarantee that the set of action states and the waiting
state can be covered \emph{together}, and the action states can be populated by a number of processes arbitrarily high. This property allows us to design an algorithm
to solve the configuration coverability problem in polynomial space. As a byproduct, it also gives us a cut-off for this problem. We also show that when one is interested in covering a specific \emph{state} of the protocol (and not a complete configuration), the copypaste property gives a polynomial time algorithm
to solve the problem. In both cases, we give matching lower bounds.  We then turn to Wait-Only non-blocking rendez-vous protocols, that allow only non-blocking rendez-vous, and no
broadcast. We show that in that case, we are able to compute, for each coverable waiting state, the maximum number of processes that can simultanesouly populate it. This gives rise
to an improved version of the copypaste property: when two states (that can be an action state and a waiting state) are coverable by an \emph{unbounded number of processes}, then
they are coverable by an unbounded number of processes \emph{together}. Hence, we can design a polynomial time algorithm for solving the configuration coverability problem, to be compared with \textsc{ExpSpace}-completeness of the same problem on general non-blocking rendez-vous protocols. We also provide a matching lower bound. This is another illustration of the fact that weakening the communication capabilities makes parameterized model-checking easier.

\paragraph{Organisation of the paper.}
We start by giving formal definitions of our model and the problems we consider in~\Cref{sec:model}. We explain and prove the copypaste property in~\Cref{sec:prelim}, before
proving P-completeness of the state coverability problem for Wait-Only non blocking broadcast protocols in~\Cref{sec:Scover:in:P}. We then show in~\Cref{sec:CCover} that the
configuration coverability problem for these protocols is \textsc{PSpace}-complete. Finally, \Cref{sec:RDV} gives the proof of $P$-completeness of configuration coverability (and state coverability)
for protocols restricted to non-blocking rendez-vous as a communication mechanism. We finish with concluding remarks. 

This paper is an extended version of~\cite{GuillouSS24}, and for the sake of completeness of results on safety properties over non-blocking rendez-vous protocols, contain also
some ommitted proofs of results published in~\cite{guillou-safety-concur23}.

%% file: model.tex
\section{Model and verification problems}\label{sec:model}
We denote by $\nat$  the set of natural numbers.
%and $[i,j]$\nas{je crois que ce n'est jamais utilisé} the set $\set{k\in \nat \mid i\leq k \mbox{ and } k \leq j}$ for $i,j \in \nat$. 
  For a finite set $E$, the set $\nat^E$ represents
the multisets over $E$. For two elements $s,s' \in  \nat^E$, we let
$s+s'$ be the multiset defined by $(s+s')(e) = s(e) +s'(e)$ for
all $e \in E$. We say that $s'$ is bigger than $s$, and write $s \preceq s'$ if and only if $s(e) \leq
s'(e)$ for all $e \in E$. If $s \preceq s'$, then $s'-s$ is the multiset
such that  $(s'-s)(e) = s'(e)-s(e)$ for
all $e \in E$. Given a subset $E' \subseteq E$ and $s \in \nat^E$, we denote by $|s|_{E'}$ the sum $\Sigma_{e\in E'}s(e)$ of elements of $E'$ present in $s$. The size of a multiset $s$ is given by
$|s| =|s|_E$. For $e \in E$, we use sometimes the
notation $e$ for the multiset $s$ such that $s(e)=1$ and
$s(e')=0$ for all $e' \in E\setminus\set{e}$ . When more convenient, we will use an alternative representation of multisets: for instance for the multiset
with four elements $a, b,b$ and $c$, we will also use 
the notations $\mset{a, b, b, c}$ or $\mset{a, 2\cdot b, c}$.

\subsection{Networks of Processes using Rendez-Vous and Broadcast}

We now present the model under study in this work. We consider networks of processes where each entity executes the same protocol described by a finite state automaton. Given a finite alphabet $\Sigma$ of messages, the transitions of a protocol are labelled with four types of actions that can be executed by the processes of the network. For $m \in \Sigma$ a process can  (1) send a (non-blocking) rendez-vous over the message $m$ with $!m$, (2) send a broadcast over $m$ with $!!m$, (3) receive a rendez-vous or a broadcast over $m$ with $?m$ and (4) perform an internal action with $\tau$ (assuming $\tau \not\in \Sigma$). In order to refer to these different actions, we denote by $!\Sigma$ the set $\set{!m \mid m \in \Sigma}$, by $!!\Sigma$ the set $\set{!!m \mid m \in \Sigma}$ and by $?\Sigma$ the set $\set{?m \mid m \in \Sigma}$. Finally, we use the notation $\Op_\Sigma$ to represent the set of labels $ !!\Sigma \cup !\Sigma \cup ?\Sigma \cup \set{\tau}$ and $\Act_\Sigma$ to represent the set of actions $ !!\Sigma \cup !\Sigma \cup \set{\tau}$. 
%Finally we will assume that for all $m\in\Sigma$ we cannot have both 
%$!m\in\Act_\Sigma$ and $!!m\in\Act_\Sigma$.\nas{vous etes ok?}

\begin{definition}
  A \emph{protocol} \PP~is a tuple $(Q, \Sigma, \qinit, T)$ such that $Q$ is a finite set of states, $\Sigma$ is a finite alphabet, $\qinit$ is an initial state, and $T \subseteq Q \times \Op_\Sigma \times Q$ is the transition relation.
  \end{definition}

In this work, we focus on some syntactical restriction on such protocols. We say that a protocol is  \emph{Wait-Only} when for all $q\in Q$, either $\set{q' \mid (q,\alpha, q') \in T \mbox{ with } \alpha \in ?\Sigma } = \emptyset$, or $\set{q' \mid (q, \alpha, q') \in T \mbox{ with } \alpha \in !!\Sigma \cup !\Sigma \cup \set{\tau} } = \emptyset$. 
%$\set{q' \mid \textrm{there exists } m \in \Sigma, (q, ?m, q') \in T } = \emptyset$, either $\set{q' \mid (q, \alpha, q') \in T, \alpha \in\Act_\Sigma } = \emptyset$. 
We call a state respecting the first or both conditions  an \emph{action} state and a state respecting the second condition a \emph{waiting} state. In the following,
we denote by $\activeset{Q}$ the set of action states of $\PP$ and $\waitingset{Q}$ its set of waiting states, which provides a partition of the total set of states. 

If the protocol 
%does not contain any rendez-vous transition of the form $(q,!m,q')$, we say that the protocol is a \emph{broadcast protocol} and if it 
does not contain any broadcast transition of the form  $(q,!!m,q')$ (resp. sending transition of the form $(q, !m, q')$), we call it a \emph{Rendez-vous protocol} (\rdvprot) (resp. a \emph{Broadcast protocol}).

	\begin{figure}[htbp]
		\begin{center}
			\input{Figures/example-1.tex}
		\end{center}
		\caption{Example of a protocol denoted $\PP_{\mathsf{dashed}}$ (we note $\PP$ the protocol $\PP_{\mathsf{dashed}}$ without the dashed arrow between $q_2$ and $q_3$)}
        \label{fig:example-1}
      \end{figure}
      
	\begin{example}
		An example of protocol is depicted on \cref{fig:example-1}. We name $\PP$ the protocol drawn without the dashed arrow between $q_2$ and $q_3$, and $\PP_{\mathsf{dashed}}$ the complete protocol. Note that $\PP$ is a \emph{Wait-Only} protocol, indeed each state is either an action state %(i.e. all its outgoing transitions are labelled by action: broadcasts, rendez-vous or internal), 
		($\qinit, q_2, q_3, q_5$ and $q_6$), or a waiting state, ($q_1$ and $q_4$).
		However, $\PP_{\mathsf{dashed}}$ is not a Wait-Only protocol, since $q_2$ is neither an action state nor a waiting state as it has an outgoing transition labelled with an action $!!c$, and an outgoing transition labelled  with an action $?a$. 
	\end{example}
	
      We shall now present the semantics associated with protocols. Intuitively, we consider networks of processes, each process being in a state of the protocol and changing its state according to the transitions of the protocol with the following assumptions. A process can perform on its own an internal action $\tau$ and this does not change the state of the other processes. When a process sends a broadcast with the action $!!m$, then all the processes in the network which are in a state from which the message $m$ can be received (i.e. with an outgoing transition labelled by $?m$) have to take such a transition. And when a process sends  a rendez-vous with the action $!m$, then \emph{at most} one process receives it: in fact, if there is at least one process in a state from which the message $m$ can be received, then exactly one of these processes has to change its state, along with the receiver (while the other processes do not move), but if no process can receive the message
      $m$, only the sender performs the action $!m$. This is why we call this communication mechanism a \emph{non-blocking} rendez-vous.       

  We move now to the formal definition of the semantics. Let  $\PP=(Q, \Sigma, \qinit, T)$ be a protocol.
A \emph{configuration} $C$ over \PP~is a non-empty multiset over $Q$. It is \emph{initial} whenever $C(q) = 0$ for all $q \in Q \setminus \set{\qinit}$.
	%of $n$ agents is a function $C : [n] \mapsto Q$. It is initial whenever for all $i \in [n]$, $C(i) = \qinit$.
	We note $\CCprot{}$ the set of all configurations over \PP, and $\Cinitprot{}$ the set of all initial configurations over \PP. 
    For $q \in Q$, we let $R(q)=\set{m \in \Sigma \mid \textrm{there exists } q'\in Q, (q, ?m, q') \in T}$ be the set of messages that can be received when in the state $q$.
%Let $a \in \Sigma$, we note $\br(a) = \set{!a, !!a}$.
    Given  a transition $t = (q, \alpha, q') \in T$, we define the relation $\transup{t} \subseteq \CC \times \CC$ as follows: for two configurations $C, C'$ we have $C \transup{t} C'$ iff %$||C|| = ||C'||$ and 
    one of the following conditions holds:
	\begin{itemize}
%		\item $\alpha = \tau$ and $\exists i \in [n]$ such that $C(i)  = q$, $C'(i) = q'$ and for all $j \neq i$, $C'(j) = C(j)$;
%		\item $\alpha = !!a$ and $\exists i \in [n]$ such that $C(i)  = q$, $C'(i) = q'$ and for all $j \neq i$, if $\exists (C(j), ?a, p') \in T$ for some $p'$, then $(C(j) , ?a , C'(j)) \in T$, otherwise, $C'(j) = C(j)$.
			\item[(a)] $\alpha = \tau$, $C(q) >0$ and $C' = C -\mset{q}+ \mset{q'}$;
			\item[(b)] $\alpha = !!m$, $C = \mset{q_1,q_2, \dots,q_{n}, q}$ for some $n \in \nat$, and $C' = \mset{q'_1 ,q'_2,\dots,q'_{n},q'} $ where for all $1 \leq i \leq n$, either $m \nin R(q_i)$ and $q'_i=q_i$, or $(q_i, ?m, q'_i) \in T$;
			%\item[(c)] $\alpha = !m$, and  $C(q) > 0$, and $(C-\mset{q})(p) > 0$ for all $p \in Q$ such that  $m \nin R(p)$, and $C' = C -\mset{q} + \mset{q'}$;
			\item[(c)] $\alpha = !m$, $C(q) > 0$, $(C-\mset{q})(p) = 0$ for all $p \in Q$ such that  $m \in R(p)$, and $C' = C -\mset{q} + \mset{q'}$;
			\item[(d)] $\alpha = !m$, $C(q) > 0$, there exists $p \in Q$ such that $(C-\mset{q})(p) > 0$ and $(p, ?m, p') \in T$ for some $p' \in Q$, and $C' = C -\mset{p,q} + \mset{q', p'}$.
            \end{itemize}
Observe that when $C\transup{t} C'$, we necessarily have $|C|=|C'|$.
%\nas{j'ai supprimé $||C||=||C'||$ des conditions de la transition car j'ai l'impression que ça découle des autres conditions.}

 The case (a) corresponds to an internal action of a single process, the case (b) to a broadcast emissio, hence all the processes that can receive the message have to receive it. The case (c) corresponds to the case where a process sends a rendez-vous and there is no process to answer to it, hence only the sender changes its state. The case (d) corresponds to a classical rendez-vous where a process sends a rendez-vous and another process receives it. Note that 
 for both the broadcast and the rendez-vous, the absence of a receiver does not prevent a sender from its action. We call our semantics non-blocking because
  of the case (c), which contrasts with the broadcast model of \cite{esparza-verif-lics99} for instance, where this case is not possible. 
  %\nas{ça vous va?}

We write $C \trans C'$ whenever there exists $t \in T$ such that $C \transup{t} C'$, and denote by $\trans^\ast$ [resp. $\trans^+$] the reflexive and transitive [resp. transitive] closure of $\trans$. An \emph{execution} $\rho$ is then a finite sequence of the form $C_0 \transup{t_1} C_1 \transup{t_2} \ldots  \transup{t_n} C_n$, and it is said to be \emph{initialized} when $C_0$ is an initial configuration in $\Cinit$.
	
	\begin{example}\label{example-2}
      We consider the protocol \PP\ of \cref{fig:example-1}. We then have the following execution starting at the initial configuration  $\mset{\qinit, \qinit, \qinit}$ with three processes:
      \begin{align*}
      \mset{\qinit, \qinit,\qinit} & \transup{(\qinit,!!a,q_1)} \mset{q_1, \qinit, \qinit} \transup{(\qinit,!b,q_4)} \mset{q_2, q_4, \qinit} \transup{(\qinit,!b,q_4)} \mset{q_2, q_4, q_4} \\
       &  \transup{(q_2,!!c,q_1)} \mset{q_1, q_5, q_5} \transup{(q_5,!!a,q_6)}  \mset{q_3, q_6, q_5}.
   \end{align*}

It corresponds to  the following sequence of events: one of the agents broadcasts message $a$ (not received by anyone), then another agent sends message $b$ 
which leads to a rendez-vous with the first agent on $q_1$, the last agent sends message $b$ which is not received by anyone (the sending is possible thanks to
the non-blocking semantics), the agent in state $q_2$ broadcasts message $c$ which is received by the two other agents, and finally, one of the agents in 
$q_5$ broadcasts letter $a$ which is received by the process
on $q_1$.
	\end{example}

\begin{remark}
	Observe that internal transitions labelled by $\tau$ can be replaced by broadcast transitions of the form $!!\tau$. Since no transition is labelled by $?\tau$,
	when $\tau$ is broadcasted, no process is ready to receive it and the semantics is equivalent to the one of an internal transition. Observe also that
	since $\tau\in\Act_\Sigma$, transforming internal transitions into broadcasts keeps a protocol Wait-Only. 
	%a special message $\#$ such that $\# \nin \Sigma$. When $\#$ is broadcast, no process is ready to receive it as there is no reception transition for message $\#$. As a consequence, only the sender changes state, and all other processes remain still. Note that when considering Wait-Only Protocols, this transformation maintains the partition of active and waiting states: a state with an internal transition is by definition an \emph{active} state, and so transforming the internal transition into a broadcast one preserves the nature of the state (i.e. it is still an action state).
	\end{remark}
Following this remark, we will omit internal transitions in the rest of this work.

 % \arstext{Je suis pas fan du tableau suivant... je vous expliquerai pourquoi}
 %    We are now able to define several types of networks, all are define with the semantics given above:

 %    \begin{center}
 %    	\begin{tabular}{|c|c|}
 %    		\hline
 %    		{\bf Network name} & {\bf Input Protocol}\\
 %    		\hline
 %    		\hline
 %    		Br+Nb-RDV-Networks & Protocol\\
 %    		\hline
 %    		 Br-Networks & Broadcast protocol\\
 %    		\hline
 %    		Nb-RDV-Networks & Rendez-Vous protocol \\
 %    		\hline
 %    		Wait-only Br-Networks & Wait-only Broadcast protocol \\
 %    		\hline
 %    		Wait-only Nb-RDV-Networks & Wait-only Rendez-Vous protocol\\
 %    		\hline
 %    		Wait-only Br+Nb-RDV-Networks & Wait-only protocol \\
 %    		\hline
 %    	\end{tabular}
 %    \end{center}
	
    \subsection{Verification Problems}

    We present now the verification problems we are interested in. Both these problems consist in ensuring a safety property:
    we want to check that, no matter the number of processes in the network, a configuration exhibiting a specific pattern 
    %, given by a control state in the first case or by a multi-set of control sets in the second case, 
    can never be reached. If the answer to the problem is positive, it means in our context that the protocol is not safe.

 The state coverability problem \SCover~ is stated as follows:

    \begin{decproblem}
  \problemtitle{$\SCover$~}
  \probleminput{A protocol  $\PP$ and
  a  state $q_f \in Q$;} 
  \problemquestion{Do there exist $C\in \Cinit$ and $C' \in \CC$ such that $C \trans^\ast C'$, and $C'(q_f)>0$ ?}
\end{decproblem}
\vspace{1em}
%<<<<<<< Updated upstream
When the answer is positive, we say that $q_f$ is \emph{coverable} by $\PP$. The second problem, called the configuration coverability problem  \CCover, is a generalisation of the first one where we look for a multi-set to be covered.
%=======
%When the answer is positive, we say that $q_f$ is \emph{coverable} by $\PP$. The second problem, called the configuration coverability problem  \SCover, is a generalisation of the first one where we look for a multi-set to be covered.
%>>>>>>> Stashed changes
 \begin{decproblem}
  \problemtitle{$\CCover$~}
  \probleminput{A protocol  $\PP$ and
  a  configuration $C_f \in \CC$;} 
  \problemquestion{Do there exist $C\in \Cinit$ and $C' \in \CC$ such that $C \trans^\ast C'$, and $C_f \preceq C'$ ?}
\end{decproblem}
% \vspace{1em}
% 	The \SCover~problem asks given a protocol \PP~and a state of the protocol $q_f$ if there exist $C\in \Cinit$, $C' \in \CC$ such that $C \trans^\ast C'$, and there exists $i \in [|C|]$ with $C'[i] = q_f$.
	
% 	The \CCover~problem asks given a protocol \PP~and a configuration $C_f \in \CC$ if there exist $C\in \Cinit$, $C' \in \CC$ such that $C \trans^\ast C'$, and $C' \succeq C_f$.
	
	\begin{remark}
		Note that if $\PP$ is a Wait-Only protocol and its initial state $\qinit$ is a waiting state, then no state besides $\qinit$ is coverable and the only coverable configurations are the initial ones. Hence, when talking about Wait-Only protocols, we assume in the rest of this work that the initial state $\qinit$ is always an \emph{action state}.
	\end{remark}

	\begin{example}
		In the protocol $\PP$ of \cref{fig:example-1}, configuration $\mset{q_3,q_6}$ is coverable as $\mset{\qinit, \qinit, \qinit} \trans^\ast \mset{q_3,q_6,q_5}$ 
		(see Example \ref{example-2}) and $ \mset{q_3,q_6}\preceq \mset{q_3,q_6,q_5}$.
		 %\color{red} peut etre pas utile comme example\color{black}
	\end{example}

	\paragraph*{Results.}
	We summarize in \cref{tab:results} the decidability status and complexity classes for these problems on the different types on protocols, where the results presented in this paper appear in red. 
	Note that concerning the lower bounds for protocols, they have been proved in \cite{schmitz-power-concur13}[Fact16, Remark 17] for protocols with broadcast and a classical ``blocking'' rendez-vous semantics, i.e. where a process requesting a rendez-vous cannot take the transition if no process answers the rendez-vous. However, it is possible to retrieve the lower bound for protocols without rendez-vous by using the fact that ``blocking'' rendez-vous can be simulated by broadcast as shown in \cite{esparza-verif-lics99,aminof-expressive-lpar15}.

	\begin{table}
		\begin{tabular}{|c|c|c|}
			\hline
			\textbf{Type of protocols} & ~~~\SCover~~~~ & \CCover \tabularnewline
			\hline
			\hline
			(Broadcast) Protocol &  \multicolumn{2}{c|}{Decidable\cite{EmersonK03} and  Ackermann-hard \cite{schmitz-power-concur13,aminof-expressive-lpar15,esparza-verif-lics99}} \\
			 \hline
			 	\rdvprot & \multicolumn{2}{c|}{\textsc{ExpSpace}-complete \cite{guillou-safety-concur23}} \tabularnewline
			\hline
			Wait-Only (broadcast) protocol & \textcolor{red}{\textbf{P-complete}}& \textcolor{red}{\textbf{\pspace-complete}}\tabularnewline
			\hline
                             Wait-Only \rdvprot & \multicolumn{2}{c|}{\textcolor{red}{\textbf{P-complete}} }\tabularnewline
			\hline
			
		\end{tabular} \caption{Coverability  in protocols}\label{tab:results}
	
      \end{table}

%% file: Figures/example-1.tex
\tikzset{box/.style={draw, minimum width=4em, text width=4.5em, text centered, minimum height=17em}}

\begin{tikzpicture}[->, >=stealth', shorten >=1pt,node distance=2cm,on grid,auto, initial text = {}] 
	\node[state, initial] (q0) {$\qinit$};
	\node[state] (q1) [right = of q0, yshift = 20] {$q_1$};
	\node[state] (q2) [right = of q1] {$q_2$};
	\node[state] (q3) [right  = 1 of q1, yshift = 35] {$q_3$};
	\node[state] (q4) [right  = of q0, yshift = -20] {$q_{4}$};
	\node[state] (q5) [right  = of q4] {$q_{5}$};
	\node[state] (q6) [right  = of q5] {$q_{6}$};
	
	\path[->] 
	(q0) edge [thick,bend right = 0] node  []{$!b$} (q4)
			edge [thick,bend left = 0] node  []{$!!a$} (q1)
	(q1) edge [thick,bend left = 0] node  []{$?a$} (q3)
			edge [thick,bend left = 20] node  [above]{$?b$} (q2)
	(q2) edge [thick,bend left = 0,dashed] node  [right, yshift = 5]{$?a$} (q3)
			edge [thick,bend left = 20] node  [below]{$!!c$} (q1)
	(q4) edge [thick,bend left = 0] node  [above]{$?c$} (q5)
	(q5) edge [thick] node  [above]{$!!a$} (q6)
	
	;
\end{tikzpicture}

%
%\begin{tikzpicture}[->, >=stealth', shorten >=1pt,node distance=2cm,on grid,auto, initial text = {}] 
%	\node[state, initial] (q0) {$\qinit$};
%	\node[state] (q1) [right = of q0, yshift = 25] {$q_1$};
%	\node[state] (q2) [right = of q1] {$q_2$};
%	\node[state] (q3) [right  = of q0, yshift = -25] {$q_3$};
%	\node[state] (q4) [right  = of q3] {$q_{4}$};
%	\node[state, fill=green] (q5) [right  = of q4] {$q_{5}$};
%	
%	\path[->] 
%	(q0) edge [thick,bend right = 20] node  [below]{$\tau$} (q3)
%	edge [thick,bend left = 20] node  [above]{$!!c$} (q3)
%	edge [thick,bend left = 0] node  [above]{$!!b$} (q1)
%	(q1) edge [thick,bend left = 20] node  [above]{$?a$} (q2)
%	(q2) edge [thick,bend left = 20] node  [below]{$!!b$} (q1)
%	(q3) edge [thick,bend left = 20] node  [above]{$?b$} (q4)
%	(q4) edge [thick,bend left = 20] node  [below]{$!a$} (q3)
%	(q4) edge [thick, draw = green] node  [above]{$?c$} (q5)
%	
%	;
%\end{tikzpicture}

%% file: prelim.tex
\section{Preliminaries properties}\label{sec:prelim}
	
%	\subsection{General case}\label{subsection:general-case}
%	\nas{dire que la preuve de décidabilité a été faite pour les Broadcast Protocols (avec rendez-vous bloquants) dans LICS99, et qu'on l'adapte.}
%	In the general case (i.e. when there is no restriction on the input protocol) we argue that both problems are decidable. The proof is based on the theory of well quasi orders: we show that, given a protocol \PP, $(\CCprot, \preceq)$ form a well-quasi order, and that $(\CC_\PP, \preceq, \trans)$ form a well structured transition system. %The former is proved by showing a monoticity property 
%	The decidability of both problems then follows from the properties of well-structured transitions systems. 
%	
%	To contrast, we also show that the two problems are \emph{non-primitive recursives}. \color{red} explain how \color{black}
%	
%	
%	 

%	\subsection{Property on coverable states}
	
	Wait-Only protocols enjoy a nice property on coverable states.
	The property makes a distinction between action states and waiting states. First, we show that when an action state is coverable, then it is coverable by a
	number of processes as big as one wants, whereas this is not true for a waiting state. Indeed, it is possible that a waiting state can be covered by exactly one process at a time,
	and no more.  However, we show that if two action states, or if an action state and a waiting state are coverable, then there is an execution that reaches a configuration where they are both covered. 
	
	This property relies on the fact that once the action state has been covered in an execution, it will not be emptied while performing the sequence of actions
	allowing to cover the second (waiting state), since no reception of message can happen in such a state. 
	As we will see, this phenomenon can be generalised to a subset of action states. 
		%The property can be explained as follows: if one is able to cover \emph{an action state} with one execution, and another state with another execution, then one is able to cover both states with one execution, i.e\ in the same time.
%	
	
%	In order to construct this execution, first observe that in the wait-only protocol, take the first execution covering the action state, and add all processes of the second execution. They remain in the initial state during all the first part (note that this holds because the initial state is an action state). Once the first execution is completed, (at least) one process is on the first action state. The left processes on the initial state can now complete the second execution. As the first state is an action state, the first process remains on its action state (it will not receive any message sent in this second step). At the end we get one process on the action state and one process in the second state at the same time.

	\begin{example}
		
%		
%		\color{blue} here should come an example \color{black}
%		\nastext{Ce serait bien que l'exemple montre que la propriété n'est pas vraie si le protocole n'est pas wait-only.}
%		\lugtext{Dans ce cas là il faut soit changer le protocole du premier exemple soit en montrer un autre }
%		
		Going back to the protocol \PP\ of \cref{fig:example-1}, consider the action state $q_2$. It is coverable as shown by the execution
		$
		\mset{\qinit, \qinit} \transup{(\qinit, !!a, q_1)} \mset{q_1, \qinit} \transup{(\qinit, !b, q_4)} \mset{q_2, q_4}
		$.
		From this execution, for any integer $n\in \nat$, one can build an execution leading to a configuration covering $\mset{n \cdot q_2}$. For instance, for $n=2$, we build the following execution:
		\begin{align*}
			\mset{\qinit, \qinit, \qinit, \qinit} &\transup{(\qinit, !!a, q_1)} \mset{q_1, \qinit,\qinit, \qinit} \transup{(\qinit, !b, q_4)} \mset{q_2, q_4,\qinit, \qinit}\\
			&\transup{(\qinit, !!a, q_1)} \mset{q_2, q_4,q_1, \qinit} \transup{(\qinit, !b, q_4)} \mset{q_2, q_4,q_2, q_4}.
		\end{align*}
		Furthermore each coverable waiting state is coverable by a configuration that also contains $q_2$. For instance, $\mset{q_2, q_4}$ is coverable as shown by the above execution.
	
		Note that when considering $\PP_{\mathsf{dashed}}$, which is not Wait-Only, such an execution is not possible as the second broadcast of $a$ should be received by the process on $q_2$. In fact, $q_2$ is coverable by only one process and no more. This is because $q_1$ is coverable by at most one process at a 
		time; every new process arriving in state $q_1$ will do so by broadcasting $a$, message that will be received by the process already in $q_1$. Then
		any attempt to send two processes in $q_2$ requires a broadcast of $a$, hence the reception of $a$ by the process already in $q_2$.
				For the same reason, in $P_{\mathsf{dashed}}$, $\mset{q_1, q_2}$ is not coverable whereas $q_1$ is a coverable waiting state and $q_2$ a coverable action state.
	\end{example}

	Before stating the main lemma of this section (Lemma \ref{lemma:copycat-action-state}), we need an additional definition. For each \emph{coverable} state $q \in Q$, let $\mathsf{min}_q$ be the minimal number of processes needed to cover $q$. More formally, $\mathsf{min}_q = \min\set{n \mid n \in \nat, \textrm{ there exists } C\in \CC \text{ s. t. } \mset{n.\qinit} \trans^\ast C \text{ and } C(q) > 0}$. Note that $\mathsf{min}_q$ is defined only when $q$ is coverable.

%% file: copypaste.tex
Along with ensuring the covering of multiple action states (and a waiting state) at a time, the copypaste property also presents a bound on the number of processes to do so. In particular, to cover a configuration $\mset{n \cdot q}$, we need at most $n\cdot \mathsf{min}_q$ when state $q$ is coverable. 
This is formalized in the lemma below.

\begin{lemma}[The Copypaste property]\label{lemma:copycat-action-state}
	Let $\PP=(Q, \Sigma, \qinit, T)$ be a Wait-Only protocol, $A = \set{q_1, \dots, q_n}\subseteq {Q}_A$ a subset of coverable \emph{action} states and $p\in \waitingset{Q}$ a coverable \emph{waiting} state. Then, for all $N \in \nat$, there exists an execution $\mconfInit \mtrans^\ast \mconf$ such that $\mconfInit \in\mconfsInit$, and $\set{N \cdot q_1, \dots, N\cdot q_n, p}\leq \mconf$.
	Moreover, $|\mconfInit| = N\cdot \sum^n_{i=1} \mathsf{min}_{q_i} + \mathsf{min}_p$.
\end{lemma}

\ifappendix\else Let us explain the lemma before proving it formally.\fi
This lemma states that for any subset of coverable action states, and for any coverable waiting state, one can build an execution covering all those states at the same time. Moreover, the execution can put an arbitrary number of processes on each action states. To clarify why we refer to this as the copypaste property, let us outline the proof intuition.

%\begin{tikzpicture}[
%	node distance=0.4cm and 1cm,
%	every node/.style={font=\footnotesize},
%	q/.style={, rounded corners=8pt, fill=red!20, minimum width=4cm, minimum height=3.8cm},
%	state/.style={anchor=west},
%	boxq/.style={draw, thick, minimum size=0.6cm}
%	]
%	
%	% Box
%	\node[q] (box) {};
%	
%	% States on the left
%	\node[state] at ([xshift=-1.8cm, yshift=1.4cm]box.center) {$\qinit$};
%	\node[state] at ([xshift=-1.8cm, yshift=0.9cm]box.center) {$\qinit$};
%	\node[state] at ([xshift=-1.8cm, yshift=0.2cm]box.center) {$\vdots$};
%	\node[state] at ([xshift=-1.8cm, yshift=-0.5cm]box.center) {$\qinit$};
%	\node[state] at ([xshift=-1.8cm, yshift=-1.0cm]box.center) {$\qinit$};
%	
%	% Arrow in the middle
%	\draw[->, thick] ([xshift=-0.8cm]box.center) -- ([xshift=0.8cm]box.center) node[midway, above] {*};
%	
%	% States on the right
%	\node[state] at ([xshift=1.2cm, yshift=0.9cm]box.center) {$p_1$};
%	\node[state] at ([xshift=1.2cm, yshift=0.2cm]box.center) {$\vdots$};
%	\node[state] at ([xshift=1.2cm, yshift=-0.5cm]box.center) {$p_{k-1}$};
%	\node[state] at ([xshift=1.2cm, yshift=-1cm]box.center) {$p_k$};
%	
%	% Boxed q
%	\node[boxq] at ([xshift=1.5cm, yshift=1.4cm]box.center) {$q$};
%	
%\end{tikzpicture}
%%\lulutex{dans les figure : entourer le $q_a$ qui ne bouge pas dans fig .7, trouver autre chose que '?'}

Let $\PP=(Q, \Sigma, \qinit, T)$ be a Wait-Only protocol.
Take a coverable action state $q_a \in {Q}_A$ and any coverable state $q \in Q$ (it might be that $q =q_a$).
Consider the execution covering $q_a$ [resp. $q$] and denote it $\rho_a$ [resp. $\rho$].
We are able to build an execution $\rho'$ with as many processes as in $\rho_a$ and $\rho$ combined, where
$\rho_a$ is glued before $\rho$, and such that a process remains in $q_a$ when $\rho$ is executed. This can be done thanks to the action states property. 

The proof of this lemma relies on two properties on executions of Wait-Only protocols. The first one formalizes that given any initial configuration $\mconfInit$ (for example the starting configuration of $\rho$) and execution $\rho_a$, one can build a new execution $\rho'$ which is composed of processes in $\mconfInit$ and processes of $\rho_a$. It executes $\rho_a$ while processes of $\mconfInit$ remain idle in the initial state.	
This is possible as $\qinit$ is an action state, hence processes of $\mconfInit$ will not receive any broadcast messages or messages sent with a sending step. 

The second auxiliary lemma formalizes the fact that from an execution $\mconf  \mtrans^\ast \mconf'$, if one considers a configuration $\mconf''$, then from the configuration $\mconf + \mconf''$, one can execute the same sequence of events happening in the execution $\mconf  \mtrans^\ast \mconf'$ and reach a configuration $\mconf'''$ such that:
\begin{itemize}
	\item $\mconf'''(q) \geq \mconf'(q)$ for all $q \in \waitingset{Q}$;
	\item $\mconf'''(q) \geq \mconf'(q) + \mconf''(q)$ for all $q \in {Q}_A$.
\end{itemize}
This is due to the fact that from states in ${Q}_A$, there is no outgoing reception transition, and hence, no broadcast nor sending happening between $\mconf$ and $\mconf'$ can make a process in an action state leave it.
\ifappendix These two lemmas along with the formal proof of Lemma~\ref{lemma:copycat-action-state} can be found in the appendix~\ref{sec:app-copypaste}.
\fi

%% file: scover.tex
	\section{\SCover~for Wait-Only protocols is P-complete} \label{sec:Scover:in:P}

    \subsection{Upper bound}
	We present here a polynomial time algorithm  to solve the state
    coverability problem when the considered protocol is Wait-Only. Our algorithm computes in a greedy manner
    the set of coverable states using Lemma
    \ref{lemma:copycat-action-state}.

    Given a Wait-Only protocol $\PP=(Q, \Sigma, \qinit, T)$, we compute iteratively
     a set of states $S \subseteq Q$ containing all the
    states that are coverable by $\PP$, by relying on a family $(S_i)_{i \in \nat}$ of subsets of
    $Q$ formally defined as follows (we recall that $\Act_\Sigma=!!\Sigma \cup !\Sigma $):
    \begin{align*}
    	&S_0 =\set{\qinit} \\
    	&S_{i+1} = S_i \ \cup \
		 \set{q \mid  \textrm{there exists } q' \in S_i, (q', \alpha ,q) \in T,
                      \alpha \in \Act_\Sigma} \\
      	  &\cup \set{q'_2 \mid \textrm{there exist } q_1,q_2 \in S_i, q'_1\in Q,  a \in
           \Sigma \textrm{ s. t. }  (q_1, !a, q'_1)\in T \textrm{ and } (q_2, ?a, q'_2) \in T}\\
       &\cup \set{q'_2 \mid \textrm{there exist } q_1,q_2 \in S_i, q'_1 \in Q, a \in
           \Sigma \textrm{ s. t. }  (q_1, !!a, q'_1)\in T \textrm{ and } (q_2, ?a, q'_2) \in T}
      \end{align*}
     
Intuitively at each iteration, we add some control states to $S_{i+1}$ either if
they can be reached from a transition labelled with an action (in
$\Act_\Sigma$) starting at a state in $S_i$ or if they can be reached
by two transitions corresponding to  a communication by broadcast or
by rendez-vous starting from states in $S_i$. We then define  $S=\bigcup_{n \in \nat} S_n$. Observe that $(S_i)_{i\in\mathbb{N}}$
is an increasing sequence such that $|S_i|\leq |Q|$ for all $i\in\mathbb{N}$. Then we reach a fixpoint $M\leq |Q|$ such that 
$S_M=S_{M+1}=S$. Hence $S$ can be computed in polynomial time. 

The two following lemmas show correctness of this algorithm. We first prove that any state 
$q\in S$ is indeed coverable by $P$. Moreover, we show that $\textsf{min}_q$ the minimal
number of processes necessary to cover $q\in Q$ is smaller than $2^{|Q|}$. 

\begin{lemma}\label{lemma:scover-sound}
If $q \in S$, then there exists $C \in  \mathcal{I}$ and $C' \in
\mathcal{C}$ such that $C \trans^\ast C'$, $C'(q) >0$ and $|C| \leq 2^{|Q|}$.
		% If $q \in S$, then $q$ is coverable. Furthermore, it is coverable from an initial configuration with at most $2^n$ processes where $n = |Q|$.
	\end{lemma}
	\begin{proof}
		Let $M \in \nat$ be the first natural such that
        $S_M=S_{M+1}$. We have then $S_M=S$ and  $M \leq |Q|$.  We
        prove by induction that for all $0\leq i \leq M$, for all $q \in
        S_i$, there exists $C \in  \mathcal{I}$ and $C' \in
\mathcal{C}$ such that $C\trans^\ast C'$, $C'(q) >0$ and $|C| \leq 2^i$.
		
		As $S_0 = \set{\qinit}$, the property trivially holds for $i =
        0$, since $\mset{\qinit} \in \mathcal{I}$
        and  $\mset{\qinit}(\qinit) >0$.
        
		Assume now the property to be true for $i < M$ and let $q \in
        S_{i+1}$. If $q \in S_i$, then by induction hypothesis, we
        have that  there exists $C \in  \mathcal{I}$ and $C' \in
\mathcal{C}$ such that $C\trans^\ast C'$, $C'(q) >0$ and $|C| \leq 2^i<
2^{i+1}$. We suppose that $q \notin S_i$ and proceed by a case
analysis on the way $q$ has been added to $S_{i+1}$.
		\begin{enumerate}
			\item there exists $q' \in S_i$ and $t =(q', \alpha,
              q)\in T$ with $\alpha \in \Act_\Sigma$. By induction
              hypothesis, there exists an execution $C \trans^\ast C'$
              such that $C'(q') >0$ and $|C| \leq 2^i$. But we have
              then $C' \transup{t} C''$ with $C''(q)>0$, and
              consequently as well $C
              \trans^\ast C''$. This is true because of the ``non-blocking'' nature
              of both broadcast and rendez-vous message in this model. Hence
              there is no need to check for a process to receive the message to ensure
              the execution $C\trans^\ast C''$.

			\item there exist $q_1 ,q_2 \in S_i $ and $q'_1 \in
              Q$ and  there exists $a \in \Sigma$ such that $(q_1,
              !a, q'_1),(q_2, \linebreak[0] ?a, q)\in T$. %and
              %$q\in\set{q'_1,q'_2}$. 
              By induction hypothesis, we have
                that there exists $C_1,C_2 \in \mathcal{I}$ and
                $C'_1,C'_2 \in \mathcal{C}$ such that $C_1 \trans^\ast
                C'_1$ and $C_2 \trans^\ast C'_2$  and $C'_1(q_1)>0$
                and $C'_2(q_2)>0$ and $|C_1|\leq 2^i$ and $|C_2|
                \leq 2^i$. Note furthermore that by definition of an action state, $q_1$
                is in $\activeset{Q}$ and, as $(q_2, ?a, q) \in T$, $q_2$ does not belong to $\activeset{Q}$.
                Hence $q_1\neq q_2$. By~Lemma
                \ref{lemma:copycat-action-state}, we know that there
                exist $C \in \mathcal{I}$ and $C' \in \mathcal{C}$
                such  that $C  \trans^\ast C'$ and $C'(q_1) >0$ and
                $C'(q_2)>0$. 
               % \lug{nouvelle proposition avec le nouveau lemme 1}
                Furthermore, recall that $\textsf{min}_{q_i}$ for $i \in \set{1,2}$ is
                the minimal number of processes needed to cover $q_i$, by
                Lemma \ref{lemma:copycat-action-state}, $|C| \leq \textsf{min}_{q_1} + \textsf{min}_{q_2}$.
                By induction hypothesis, $\textsf{min}_{q_1}+\textsf{min}_{q_2}\leq 2^i+2^i$,
                hence $|C|\leq 2^{i+1}$. 
              %  As $C_1 \trans^\ast C_1'$ with $C_1'(q_1) > 0$ and 
                %$C_2 \trans^\ast C_2'$ with $C_2'(q_2) >0 $ and $C_1, C_2 \in \mathcal{I}$, 
               % it holds that $\textsf{min}_{q_1} + \textsf{min}_{q_2} \leq |C_1| + |C_2| \leq 2 \times 2^i$, hence $|C| \leq 2^{i+1}$.
%                Furthermore, the proof of this latter
%                lemma constructs the execution by putting one
%                execution after the other starting by the sum of the
%                two initial configurations\ars{est-ce-que ca vous va
%                  comme ca ?}, hence the number of
%                processes in $C$ is equal to the sum of the number of
%                processes in $C_1$ and the number of processes in
%                $C_2$, hence $|C|=|C_1|+|C_2|\leq 2^i + 2^i =
%                2^{i+1}$. 

                We then have $C' \transup{(q_1,!a, q'_1)} C''$
                with  $C'' = C' -\mset{q_1,q_2} + \mset{q'_1,
                  q}$. Hence $C\trans^\ast C''$ with $C''(q)>0$.
                  
\item there exist $q_1 ,q_2 \in S_i $ and $q'_1 \in
              Q$ and  there is some $a \in \Sigma$ such that $(q_1,
              !!a, q'_1),(q_2, \linebreak[0]?a, q)\in T$.
              As above, we obtain the existence of an execution 
              $C\trans^\ast C'$ with $C'(q_1)>0$ and $C'(q_2)>0$, and $|C|\leq 2^{i+1}$. 
              Then $C'=\mset{q_1,q_2,\dots, q_k}$ and $C'\transup{(q_1,!!a,q'_1)}C''$
              with $C''=\mset{q'_1, q,\dots, q'_k}$ with, for all $3\leq j\leq k$, either $a\notin R(q_j)$
              and $q_j=q'_j$ or $(q_j,?a,q'_j)\in T$. In any case, we have $C\trans^\ast C''$ with 
              $C''(q)>0$. 

              %   there exist two executions: $C_0 \trans^\ast C$ such that $C(q_1) >0$, $|C_0|\leq 2^i$ and $C'_0 \trans^\ast C'$ such that $C'(q_2) >0$, $|C'_0| \leq 2^i$. Then, by \cref{lemma:copycat-action-state}, there exists an execution $C''_0 \trans^\ast C''$ such that $C'''(q_1) >0$ and $C''(q_2) > 0$. Furthermore, the proof of the copycat lemma constructs the execution by putting one execution after the other starting by the sum of the two initial configurations, hence the number of processes in $C''$ is equal to the number of processes in $C_0$ plus the number of processes in $C'_0$. By induction hypothesis $|C''|\leq 2^i + 2^i = 2^{i+1}$.
			% Hence, there exists $C''_2$ such that $C''\transup{(q_1,\alpha,p)} C''_{2}$ and where at least one agent (exactly one if $\alpha =!a$ and all if $\alpha =!!a$) on state $q_2$ in $C''$ takes the transition $(q_2, ?a, q)$, hence $C''_{2}(q) >0$, and $|C''_{2}| \leq 2^{i+1}$ which concludes the induction.
			
			% Hence, for all $q \in S$, $q$ is coverable, furthermore, as $m \leq |Q|$, $q$ is coverable with at most $2^{|Q|}$ processes.
			
				\end{enumerate}
				So, for any $q\in S$, $q\in S_M$ and we have an execution $C\trans^\ast C'$ with $C\in\Cinit$ such that $C'(q)>0$ and $|C|\leq 2^M\leq
				2^{|Q|}$. 
      \end{proof}

 We now prove the completeness of our algorithm by showing that every
 state coverable by $\PP$ belongs to $S$.

	\begin{lemma}\label{lemma:scover-complete}If there exists $C \in  \Cinit$ and $C' \in
\CC$ such that $C \trans^\ast C'$ and $C'(q) >0$, then $q \in S$.
\end{lemma}
	\begin{proof}
		We consider the initialized execution $C_0 \transup{t_1} C_1
        \transup{t_2} \ldots  \transup{t_n} C_n$ with $C=C_0$ and $C_n=C'$.
		% Let us note $S_0, S_1, \dots, S_m$ the sequence of sets computed by the algorithm. We adopt the convention that if $n > m$, $S_i = S_m$ for all $m < i \leq n$.
		% Note that by definition of the algorithm, $S_{i+1}$ is the
        % resulting set of the transformation made by the algorithm of
        % $S_i$. Indeed, as the last pair is $S_m$, if $S'$ is the
        % result of the algorithm after the $m+1$th iteration, $S' =
        % S_{m+1} = S_m$, as the algorithm stopped.
We will prove by induction on $0 \leq i \leq n$ that for all
        $q$ such that $C_i(q) > 0$, we have $q \in S_i$.
		
		For $i = 0$, we have $C_0 = \mset{|C| \cdot \qinit}$, and $S_0 = \set{\qinit}$. Hence the property holds.
		
		Assume the property to be true for $i< n$, and let $q \in Q$
        such that $C_{i+1}(q) > 0$. If $C_i(q) >0$, then by induction
        hypothesis we have $q \in S_i$
        and since $S_i \subseteq S_{i+1}$, we deduce that $q \in
        S_{i+1}$.  Assume now that $C_i(q)=0$. 
        %(otherwise we
        %would have $q \in S_i$). 
        We  proceed by a
        case analysis.
		\begin{enumerate}
			\item $t_{i+1}=(q',!a,q)$ or $t_{i+1}=(q',!!a,q)$  for some
              $a \in \Sigma$ and $q' \in Q$. Since $C_i
              \transup{t_{i+1}} C_{i+1}$, we have necessarily $C_i(q')>0$. By induction hypothesis,
              $q'\in S_i$, and by construction of $S_{i+1}$, we deduce that
              $q' \in S_{i+1}$.
%            \item $t_{i+1}=(q_1,\tau,q'_1)$ with $q'_1\neq q$. This case
%              is not possible because $C_i(q)=0$ and $C_{i+1}(q)>0$.
			\item  $t_{i+1}=(q_1,!a,q'_1)$ or $t_{i+1}=(q_1,!!a,q'
              _1)$
              with $q'_1 \neq q$. Since $C_i(q)=0$ and $C_{i+1}(q)>0$,
              there exists a transition of the form $(q_2,?a,q)$ with $q_1\neq q_2$
               (because $q_1\in \activeset{Q}$ and $(q_2, ?a, q) \in T$ hence $q_2 \nin
              \activeset{Q}$). Consequently, we know that we have
              $C_i(q_1)>0$ and $C_i(q_2)>0$. By induction hypothesis
              $q_1,q_2$ belong to $S_i$ and by construction of
              $S_{i+1}$ we deduce that $q \in S_{i+1}$.
		\end{enumerate}
	\end{proof}

The two previous lemmas show the soundness and completeness of our
algorithm to solve \SCover\ based on the computation of the set
$S$. Since this set of states can be computed in polynomial time, we
obtain the following result.

\begin{theorem}\label{thm:scover-in-p}
	\SCover~is in \textsc{P}\ for Wait-Only protocols.
\end{theorem}

Furthermore, completeness of the algorithm along with the bound on the number of processes established in~Lemma \ref{lemma:scover-sound} gives the following result.

\begin{corollary} \label{cor:scover-bound}
Given a Wait-Only protocol $P=(Q,\Sigma, \qinit, T)$, for all $q\in Q$ coverable by $P$, then $\mathsf{min}_q$ the minimal number of processes necessary
to cover $q$ is at most $2^{|Q|}$. 
\end{corollary}
% \begin{proof}
% 	The following algorithm answers to \SCover\ in polynimial time:
% 	given a protocol \PP~and a state $q$, compute $S$ in polynomial time, answer yes if $q \in S$ and no otherwise.
% 	Indeed, if $q$ is coverable then, from \cref{lemma:scover-completeness}, $q \in S$, and so the answer is yes, and if the answer is yes then $q \in S$ and from \cref{lemma:scover-correctness}, $q$ is coverable.
	
% 	Furthermore, let $q$ a coverable state, then from \cref{lemma:scover-correctness}, it is coverable from an initial configuration with at most $2^{|Q|}$ processes.
	
% 	\qed
% \end{proof}

\subsection{Lower Bound} \label{subsec:p-hard}

%\lugtext{ajouter une remarque sur la p dureté dans le cas seulement RDV non bloquants }

We show that \SCover\ for Wait-Only protocols is \textsc{P}-hard. For this, we provide a
reduction from the Circuit Value Problem (CVP) which is known to be P-complete \cite{Ladner75}. CVP is defined as follows: given an acyclic Boolean circuit with $n$ input variables, one output variable, $m$ boolean gates of type \emph{and}, \emph{or}, \emph{not}, and a truth assignment for the input variables, is the value of the output equal to a given boolean value? 
Given an instance of the CVP, we build a protocol in which the processes broadcast variables (input ones or associated with gates) along with their boolean
values. These broadcasts will be received by other processes that will use them to compute boolean value of their corresponding gate, and broadcast
the obtained value. Hence, different values are propagated through the protocol representing the circuit, until the state representing the output variable value
we look for
is covered.  
%at least one process  broadcasts variables along with their boolean inputs, and other processes 
%compute the boolean values of gates, depending on the received messages. Once computed, the process broadcasts the boolean value of the gate at will. 

Take for example a CVP instance $\textsf{C}$ with two variables $v_1, v_2$, and two gates: one \emph{not} gate on variable $v_1$ denoted $g_1(v_1, \lnot, o_1)$ (where 
$o_1$ stands for the output variable of $g_1$), and one \emph{or} gate on variable $v_2$ and $o_1$ denoted $g_2(o_1, v_2, \lor, o_2)$ (where $o_2$ stands for the output variable of gate $g_2$). Assume the input boolean value for $v_1$ [resp. $v_2$] is $\top$ [resp. $\bot$]. The protocol associated to $\textsf{C}$ is displayed on \cref{fig:cvp-prot}. Assume the output value of $\textsf{C}$ is $o_2$, we will show that $q_\top^2$ [resp. $q_\bot^2$] is coverable if and only if $o_2$ evaluates to $\top$ [resp. $\bot$].
Note that with the truth assignment $v_1 = \top$ and $v_2 = \bot$, $o_2$ evaluates to $\bot$, and indeed one can build an execution covering $q_\bot^2$ with three processes:
\begin{align*}
	\mset{3.\qinit} & \trans \mset{2.\qinit, q_0^2} \trans \mset{\qinit, q_0^1, q_0^2} \transup{(\qinit, !!(v_1, \top), \qinit)} \mset{\qinit, q_\bot^1, q_0^2} \\
	&\transup{(q_\bot^1, !!(o_1, \bot), q_\bot^1)} \mset{\qinit, q_\bot^1, q_1^2} \transup{(\qinit, !!(v_2, \bot), \qinit)} \mset{\qinit, q_\bot^1, q_\bot^2}
\end{align*}
%Note that the fact $q_\bot^1$ [resp. $q_\bot^2$]
%: we start with three processes from $\mset{3.\qinit}$, one process stays on $\qinit$, whereas the two other processes reaches respectively $q_0^1$ and $q_0^2$. Then, the process on $\qinit$ broadcasts $(v_1, \top)$, and the process on $q_0^1$ reaches $q_\bot^1$. Note that this corresponds to the fact that $o_1$ evaluates to $\bot$. Then, the process on $q_\bot^1$ broadcasts $(o_1, \bot)$, received by the process on $q_0^2$, finally process on $\qinit$ broadcats $(v_2, \bot)$, received by process on $q_1^2$, the latter reaches $q_\bot^2$. This corresponds to the fact that $o_2$ evaluates to $\bot$.
%and that we want check that $o_2$ evaluates to $\bot$. Then, one will ask if there is an execution covering state $q_\bot^2$.

\begin{figure}
	\centering
	\input{Figures/fig-cvp-prot.tex}
	\caption{Protocol for a CVP instance with two variables $v_1, v_2$, two gates $g_1(\lnot, v_1, o_1)$ and $g_2(\vee, o_1, v_2, o_2)$, and input $\top$ for $v_1$ and $\bot$ for $v_2$ and output variable $o_2$.
	Depending on the truth value of $o_2$ to test, the state we ask to cover can be $q_\bot^2$ or $q_\top^2$.
	}\label{fig:cvp-prot}
\end{figure}

We now present the formal proofs.
We start by introducing some notations.
We denote  an instance of CVP: $\textsf{C} = (V, o, G, B ,b)$ where $V = \set{v_1, \dots,\linebreak[0] v_n}$ denotes the $n$ input variables, $o$ is the output variable, $G = \set{g_1, \dots ,g_m}$ the $m$ boolean gates,  $B = \set{b_1, \dots, b_n}$ the boolean assignment such that for all $1 \leq i \leq n$, boolean $b_i \in \set{\top,\bot}$ is the assignment of variable $v_i$, and $b$ the boolean output value to test.
Let $V' = V \cup \set{o_1, \dots, o_m}$ where $o_j$ is the output variable of gate $j$ for $1 \leq j \leq m$. Wlog we can assume that $o = o_m$. For $1 \leq j \leq m$, we denote gate $g_j$  by $g_j(\diamond, x_1, x_2,o_j)$ with $x_1,x_2 \in V'$ and $\diamond \in \set{\vee, \land}$ or  by $g_j(\lnot, x, o_j)$ with $x \in V'$. As $C$ is acyclic, one can assume that $x_1, x_2, x \in V \cup \set{o_1, o_2, \dots, o_{j-1}}$. 

Let some $x \in V'$, we denote $\mathsf{bv}(x)$ the boolean value of $x$ with respect to the input $B$. Note that if $x \in V$, there exists $1 \leq i \leq n$ such that $x = v_i$ and so $\mathsf{bv}(x) = b_i$.

Let $ 1 \leq j \leq m$, we describe $\PP_j = (Q_j, \Sigma_j, T_j)$ where $\Sigma_j = V' \times \set{\top, \bot}$ as follows: 
\begin{itemize}
	\item if $g_j(\vee, x_1, x_2,o_j)$, $Q_j = \set{q_0^j, q_\top^j, q_1^j, q_\bot^j}$, and 
	\begin{align*}
	T_j &=  \set{(q_0^j, ?(x_k, \top), q_\top^j) \mid k = 1,2}   \cup \set{(q_0^j , ?(x_1,\bot), q_1^j), (q_1^j, ?(x_2, \bot), q_\bot^j)} \\
	 & \cup \set{(q_\top^j, !!(o_j, \top), q_\top^j), (q_\bot^j, !!(o_j,\bot), q_\bot^j)};
	\end{align*}
	
	\item if $g_j(\land, x_1, x_2,o_j)$, $Q_j = \set{q_0^j, q_\top^j, q_1^j, q_\bot^j}$, and 
		\begin{align*}
	T_j &= \set{(q_0^j, ?(x_k, \bot), q_\bot^j) \mid k = 1,2}  \cup \set{(q_0^j , ?(x_1,\top), q_1^j), (q_1^j, ?(x_2, \top), q_\top^j)} \\
	& \cup \set{(q_\top^j, !!(o_j, \top), q_\top^j), (q_\bot^j, !!(o_j,\bot), q_\bot^j)};
		\end{align*}
	
	\item if $g_j(\lnot, x, o_j)$, $Q_j = \set{q_0^j, q_\top^j, q_\bot^j}$, and 
			\begin{align*}
				T_j = \set{(q_0^j, ?(x, \bot), q_\top^j) ,(q_0^j , ?(x,\top), q_\bot^j)}\cup \set{(q_\top^j, !!(o_j, \top), q_\top^j), (q_\bot^j, !!(o_j,\bot), q_\bot^j)}.
					\end{align*}
\end{itemize}

We are now ready to define the protocol associated to $\textsf{C}$, $\PP_{\textsf{C}}= (Q, \Sigma, \qinit, T)$:
% displayed in \cref{fig:cvp}.
%\begin{figure}
%	%	\input{Figures/fig-cvp.tex}
%	\begin{minipage}[c]{0.45\linewidth}
	%		\input{Figures/fig-cvp.tex}
	%%		\caption{Display of $\PP_C$.}\label{fig:cvp}
	%	\end{minipage}\hfill
%	\begin{minipage}[c]{0.45\linewidth}
	%		\input{Figures/fig-or-prot.tex}
	%%		\caption{Display of $\PP_j$ for a gate $g_j(\lor, x_1, x_2, o_j)$.}\label{fig:cvp:or-gate}
	%	\end{minipage}
%	\caption{Display of the protocol $\PP_C$ (to the left) and of a protocol $\PP_j$ for a gate $g_j(\lor, x_1, x_2, o_j)$ (to the right).}\label{fig:cvp}
%\end{figure}
\begin{itemize}
	\item $Q = \set{\qinit} \cup \bigcup_{1 \leq j \leq m} Q_j$;
	\item $\Sigma =V' \times \set{\top, \bot}$;
	\item $T = \set{(\qinit, !!(v_j, b_j), \qinit) \mid 1 \leq j \leq n} \cup \set{(\qinit, !!\tau, q_0^j) \mid 1 \leq j \leq m} \bigcup_{1 \leq j \leq m} T_j$.
\end{itemize}

Observe that $\PP$ is Wait-Only: $\qinit$ is an active state, and for all $1 \leq j \leq m$: if $g_j$ is a "not" gate, $q_0^j$ is a waiting state, and $q_\bot^j, q_\top^j$ are active states, and if $g_j$ is an "and" gate or an "or" gate, $q_\bot^j, q_\top^j$ are active states and $q_0^j, q_1^j$ are waiting states. 

We show that $\mathsf{bv}(o) = b$ if and only if there is an initial configuration $C_0 \in \mathcal{I}$ and $C_f \in \CC$ such that $C_0 \trans^\ast C_f$ and $C_f(q_b^m) > 0$. 

\begin{lemma}\label{p-hard:completeness}
	If $\mathsf{bv}(o) = b$, then there exists $C_0 \in \mathcal{I}$ and $C_f \in \CC$ such that $C_0 \trans^\ast C_f$ and $C_f(q_b^m) > 0$. 
\end{lemma}
\begin{proof}
	Assume that $\mathsf{bv}(o) = b$, and take $C_0 = \mset{(m+1).\qinit}$. There exists an execution $C_0 \trans^\ast C_f$ with $C_f = \mset{\qinit, q_{y_1}^1, q_{y_2}^2, \dots q_{y_m}^m}$ where $y_j = \mathsf{bv}(o_j)$ with the input boolean values $B$ for $1 \leq j \leq m$. By definition, $y_m = \mathsf{bv}(o) = b$.
	The execution is: $C_0  \trans^+ C_1 \trans^+ \dots \trans^+ C_m$ where $C_j = \mset{\qinit, q_{y_1}^1, \dots q_{y_j}^j, \qinit, \dots, \qinit}$ for all $0 \leq j < m$. Between $C_j$ and $C_{j+1}$, the sequence of transitions is:
	\begin{itemize}
		
		\item if $g_{j+1}(\vee,x_1,x_2,o_{j+1})$ with $\mathsf{bv}(x_k) = \top$ for some $k \in\set{ 1, 2}$, then $\mathsf{bv}(o_{j+1}) = \top$. Either $x_k = v_i$ (and $b_i = \top$) for some $1 \leq i \leq n$, or $x_k = o_i$ (and $y_i = \top$) for some $1 \leq i \leq j$. In the first case, as $C_j(\qinit) \geq 2$, then consider the sequence $C_j \transup{(\qinit, !!\tau, q_0^{j+1})} C'_{j} \transup{(\qinit, !!(v_i, \top), \qinit)} C_{j+1}$. It holds that $C'_{j} = C_j -\mset{\qinit} + \mset{q_0^{j+1}}$ and $C_{j+1} = C'_{j} - \mset{q_0^{j+1}} + \mset{q_\top^{j+1}}$. Hence $C_{j+1} = \mset{\qinit, q_{y_1}^1, \dots q_{y_j}^j, q_{y_{j+1}}^{j+1}, \dots, \qinit}$. 
		
		In the second case, $C_j(q_{\top}^i) > 0$ as $i \leq j$ and $\top = \mathsf{bv}(o_i)$, and $C_j(\qinit) > 0$. Consider the sequence $C_j \transup{(\qinit, !!\tau, q_0^{j+1})} C'_{j} \transup{(q_{\top}^i, !!(o_i, \top), q_{\top}^i)} C_{j+1}$.  It holds that $C'_{j} = C_j -\mset{\qinit} + \mset{q_0^{j+1}}$ and $C_{j+1} = C'_{j} - \mset{q_0^{j+1}} + \mset{q_\top^{j+1}}$. Hence $C_{j+1} = \mset{\qinit, q_{y_1}^1, \dots q_{y_j}^j, q_{y_{j+1}}^{j+1}, \dots, \qinit}$.

		\item if $g_{j+1}(\land,x_1,x_2,o_{j+1})$ with $\mathsf{bv}(x_k) = \bot$ for some $k = 1, 2$, then $\mathsf{bv}(o_{j+1}) = \bot = y_{j+1}$. The sequence of transitions is built in an analogous way than the previous case, however this time the broadcast messages are $\tau$ and $(x_k, \bot)$ and the reached state is $q_\bot^{j+1}$.
		
		\item if $g_{j+1}(\vee,x_1,x_2,o_{j+1})$ with $\mathsf{bv}(x_1) = \mathsf{bv}(x_2) = \bot$, then $\mathsf{bv}(o_{j+1}) = \bot$. Either $x_1 = v_i$ (and $b_i = \bot$) for some $1 \leq i \leq n$, or $x_1 = o_i$ (and $y_i = \bot$) for some $1 \leq i \leq j$. In the first case, as $C_j(\qinit) \geq 2$, then consider the sequence $C_j \transup{(\qinit, !!\tau, q_0^{j+1})} C_{j,1} \transup{(\qinit, !!(v_i, \bot), \qinit)} C_{j,2}$. It holds that $C_{j,1} = C_j -\mset{\qinit} + \mset{q_0^{j+1}}$ and $C_{j,2} = C_{j,1} - \mset{q_0^{j+1}} + \mset{q_1^{j+1}}$. 
		
		If $x_1 = o_i$ for some $i \leq j$, $C_j(q_{\bot}^i) > 0$ as $i \leq j$ and $\bot = \mathsf{bv}(o_i)$, and $C_j(\qinit) > 0$. Consider the sequence $C_j \transup{(\qinit, !!\tau, q_0^{j+1})} C_{j,1} \transup{(q_{\bot}^i, !!(o_i, \bot), q_{\bot}^i)} C_{j,2}$.  It holds that $C_{j,1} = C_j -\mset{\qinit} + \mset{q_0^{j+1}}$ and $C_{j,2} = C_{j,1} - \mset{q_0^{j+1}} + \mset{q_1^{j+1}}$.
		
		In both cases, $C_{j,2} = \mset{\qinit, q_{y_1}^1, \dots q_{y_j}^j, q_{1}^{j+1}, \dots, \qinit}$. 
		
		We make the same cases distinctions in order to build the configuration $C_{j+1}$ such that $C_{j,2} \transup{(q, !!(x_2, \bot),q)} C_{j+1}$ where $q = \qinit$ if $x_2 \in V$, and otherwise $x_2 = o_i$ for some $i \leq j$ and $q = q_{\bot}^i$.
		
		It holds that $C_{j+1} = C_{j,2} - \mset{q_1^{j+1}} + \mset{q_\bot^{j+1}}$, hence $C_{j+1} = \mset{\qinit, q_{y_1}^1, \dots q_{y_j}^j, q_{y_{j+1}}^{j+1}, \linebreak[0] \dots, \qinit}$. 
		
		\item if $g_{j+1}(\land, x_1, x_2, o_{j+1})$ with $\mathsf{bv}(x_1) = \mathsf{bv}(x_2) = \top$, then $\mathsf{bv}(o_{j+1}) = \top$. The sequence of transitions is built in an analogous way than the previous case, however this time the broadcast messages are $\tau$, $(x_1, \top)$ and $(x_2, \top)$ and the reached state is $q_\top^{j+1}$.

		\item if $g_{j+1}(\lnot, x, o_{j+1})$ with $\mathsf{bv}(x) = \top$ (resp. $\bot$), then $\mathsf{bv}(o_{j+1}) = \bot$ (resp. $\top$). Either $x = v_i$ for some $1 \leq i \leq n$, and as $C_j(\qinit) \geq 2$, we build the following sequence: $C_j \transup{(\qinit, !!\tau, q_0^{j+1})} C'_j \transup{(\qinit, (v_i, b_i), \qinit)} C_{j+1}$. It holds that $C'_j = C_j - \mset{\qinit} + \mset{q_0^{j+1}}$ and $C_{j+1} = C'_j - \mset{q_0^{j+1}} + \mset{q_{\bar b_i}^{j+1}}$ where $\bar b_i = \bot$ (resp. $\bar b_i = \top$).
		
		Otherwise, $x = o_i$ for some $i \leq j$, and so $C_j(q_{\top}^i) > 0$ (resp. $C_j(q_{\bot}^i) > 0$) and $C_j(\qinit) > 0$. Hence, we can build the following sequence: $C_j \transup{(\qinit, !!\tau, q_0^{j+1})}, C'_j \transup{(\qinit, (o_i, \mathsf{bv}(o_i)), \qinit)} C_{j+1}$. It holds that $C'_j = C_j - \mset{\qinit} + \mset{q_0^{j+1}}$ and $C_{j+1} = C'_j - \mset{q_0^{j+1}} + \mset{q_{\bar y_i}^{j+1}}$ where $\bar y_i = \bot$ (resp. $\bar y_i = \top$).
		
		In both cases, $C_{j+1} = \mset{\qinit, q_{y_1}^1, \dots q_{y_j}^j, q_{y_{j+1}}^{j+1}, \dots, \qinit}$. 
	\end{itemize}
	Hence, $C_0 \trans^+ C_m$ where $C_m(q_{y_m}^m)>0$, and so if $b = y_m$, $C_m(q_b^m) >0$. 
\end{proof}

\begin{lemma}\label{p-hard:soundness}
	If there exists $C_0 \in \mathcal{I}$ and $C_f \in \CC$ such that $C_0 \trans^\ast C_f$ and $C_f(q_b^m) > 0$, then $\mathsf{bv}(o) = b$.
\end{lemma}
\begin{proof}
	Let $C_0 \in \mathcal{I}$ and $C_f \in \CC$ such that $C_0 \trans^\ast C_f$ and $C_f(q_b^m) > 0$.
	
	First we show that all broadcast messages $(x,b_x) \in V' \times \set{\top, \bot}$ are such that $b_x = \mathsf{bv}(x)$. We start by proving it for $x \in V$, and then proceed to prove it for $x \in \set{o_1, \dots, o_m}$ by induction on $m$. 
	
	Let $(x, b_x)$ a broadcast message such that $x \in V$, we note $x = v_i$ with $1 \leq i \leq n$. By construction of the protocol \PP, the only broadcast transition labelled with first element $v_i$ is the transition $(\qinit, !!(v_i,b_i), \qinit)$ where $b_i$ is the input boolean value for variable $v_i$, i.e. $\mathsf{bv}(v_i) = b_i$.
	Hence, all broadcast messages $(x, b_x)$ with $x \in V$, are such that $b_x = \mathsf{bv}(x)$.
	
	We prove now than for all $(x, b_x)\in \set{o_1, \dots, o_m} \times \set{\top, \bot}$, $b_x = \mathsf{bv}(x)$ and we do so by induction on $m$. 
	For $m = 1$, we have that $x = o_1$. Note that tuples containing $o_1$ can only be broadcast from $q_\bot^1$ or $q_\top^1$. 	
	Denote $g_1(\diamond, x_1, x_2, o_1)$ or $g_1(\lnot, x_3, o_1)$ with $\diamond \in \set{\vee, \land}$ and $x_1, x_2, x_3 \in V$ by acyclicity. 
	\begin{itemize}
		\item if $\diamond = \vee$, then a process reaching $q_\top^1$ has necessarily received $(x_1, \top)$ or $(x_2, \top)$, and a process reaching $q_\bot^1$ has necessarily received $(x_1, \bot)$ and $(x_2, \bot)$. As we proved, all broadcast messages containing $x_1$ (resp. $x_2$) are of the form $(x_1, \mathsf{bv}(x_1))$ (resp. $(x_2, \mathsf{bv}(x_2))$). Hence, only one state between $q_\top^1$ and $q_\bot^1$ is reachable. If it is $q_\top^1$ (resp. $q_\bot^1$), the only messages containing $o_1$ which can be broadcast are $(o_1, \top)$ (resp. $(o_1, \bot)$) and it holds that $\top = \mathsf{bv}(x_1) \vee \mathsf{bv}(x_2) = \mathsf{bv}(o_1)$ (resp. $\bot$) as the process on $q_\top^1$ (resp. $q_\bot^1$) received either $(x_1, \top)$ or $(x_2, \top)$ (resp. $(x_1, \bot)$ and $(x_2, \bot)$);~
		
		\item if $\diamond = \land$, the argument is analogous to the previous case;
		
		\item if $\diamond = \lnot$, then a process reaching $q_\top^1$ has necessarily received $(x_3, \bot)$ and a process reaching $q_\bot^1$ has necessarily received $(x_3, \top)$. As we proved, all broadcast messages containing $x_3$ are of the form $(x_3, \mathsf{bv}(x_3))$. Hence, only one state between $q_\top^1$ and $q_\bot^1$ is reachable.  If it is $q_\top^1$ (resp. $q_\bot^1$), the only messages containing $o_1$ which can be broadcast are $(o_1, \top)$ (resp. $(o_1, \bot)$) and it holds that $\top = \lnot \mathsf{bv}(x_3) = \mathsf{bv}(o_1)$ (resp. $\bot$) as the process on $q_\top^1$ (resp. $q_\bot^1$) received $(x_3, \bot)$ (resp. $(x_3, \top)$).
	\end{itemize}
	Assume the property true for $m$ gates, and let $(o_{m+1}, b_{m+1})$ a broadcast message and note $g_{m+1}(\diamond, x_1, x_2, o_{m+1})$ with $\diamond \in \set{\land, \lor}$, or $g_{m+1}(\lnot, x_3, o_{m+1})$ with $x_1, x_2, x_3 \in V \cup \set{o_1, \dots, o_m}$. By induction hypothesis, the only broadcast messages containing $x_1$, $x_2$ or $x_3$ are $(x_1, \mathsf{bv}(x_1))$, $(x_2, \mathsf{bv}(x_2))$, and $(x_3, \mathsf{bv}(x_3))$. The arguments are then the same than in the case $m = 1$.
	%	\lugtext{todo}
\end{proof}

Hence, we proved with the two previous lemmas that $\mathsf{bv}(o) = b$ if and only if $q_b^m$ is coverable and we get the following theorem.
\begin{theorem}
	\SCover\ for Wait-Only (broadcast) protocols is P-hard.
\end{theorem}

\begin{remark}
	If we transform the Wait-Only broadcast protocol presented here into a Wait-Only \rdvprot~(by transforming all the broadcast transitions into sending transitions on the same message), the reduction remains sound and complete. Indeed, in the execution of $\PP_{\textsf{C}}$ built in Lemma \ref{p-hard:completeness}, all the broadcasts are received by only one process. Hence, the same execution is possible by replacing broadcasts transitions by sending transitions. The proof of Lemma \ref{p-hard:soundness}\ works exactly the same way if the broadcasts transitions become sending transitions.
	Hence, the \SCover~and \CCover~problems for Wait-Only \rdvprot~are P-hard.
\end{remark}
\begin{theorem}\label{thm:scover-wo-rdv-phard}
	\SCover\ for Wait-Only \rdvprot s is P-hard.
\end{theorem}

Together with Theorem \ref{thm:scover-in-p}, we get the following theorem.
\begin{theorem}
	\SCover\ for Wait-Only (broadcast / RDV) protocols is P-complete.
\end{theorem}
%This reduction can be adapted to Wait-Only NB-Rendez-vous protocols, which leads to the following theorem, proving that the upper bound presented in \cite{guillou-safety-concur23}\ is tight.

%\lugtext{this way of formulating things made sense before when we were referring to open question left in Concur 23 but now it is a bit weird... to rephrase}

%\lugtext{repasser sur les notations, je pense que ça peut etre un peu harmoniser}

%% file: Figures/fig-cvp-prot.tex
\tikzset{box/.style={draw, minimum width=4em, text width=4.5em, text centered, minimum height=17em}}

\begin{tikzpicture}[->, >=stealth', shorten >=1pt,node distance=2cm,on grid,auto, initial text = {}] 
	\node[state, initial] (q0) {$\qinit$};
	
	\node[state] (p1) [right = of q0, xshift = 0, yshift = 20]{$q_0^1$};
	\node[state] (p2) [right = of p1, xshift = 15, yshift = 0]{$q_\top^1$};
	\node[state] (p) [right = of p2, xshift = 30, yshift = 0]{$q_\bot^1$};
	
	\node[state] (q) [right = of q0, yshift = -25]{$q_0^2$};
	\node[state] (q1) [right = of q, xshift = 15, yshift = 18] {$q_\top^2$};
	\node[state] (q2) [right = of q, xshift = 15, yshift =-15] {$q_1^2$};
	\node[state] (q3) [right  = of q2] {$q_\bot^2$};
	
	\path[->] 
	
	(q0) edge [thick] node {$!!\tau$} (p1)
	edge [thick] node {$!!\tau$} (q)
	edge [thick, loop below] node [] {
		\begin{tabular}{l}
			$!!(v_1, \top)$ \\
			$!!(v_2, \bot)$ 
		\end{tabular}
} ()
	
	(p1) edge [thick,bend left = 0] node  [below]{$?(v_1, \bot)$} (p2)
	(p1) edge [thick,bend left = 18] node  []{$?(v_1, \top)$} (p)
	(p) edge [thick, loop right] node [] {$!!(o_1, \bot)$} ()
	(p2) edge [thick, loop right] node [] {$!!(o_1, \top)$} ()
	
	(q) edge [thick,bend left = 10] node  [yshift = -5]{$?(o_1, \top)$} (q1)
	edge [thick,bend right = 10] node  [yshift = -15, xshift = 25]{$?(v_2, \top)$} (q1)
	edge [thick,bend right = 10] node  [yshift = -20, xshift = -25]{$?(o_1, \bot)$} (q2)
	(q1) edge [thick,loop right] node  []{$!!(o_2, \top)$} ()
	(q2) edge [thick,bend left = 0] node  []{$?(v_2, \bot)$} (q3)
	(q3) edge [thick,loop right] node  []{$!!(o_2, \bot)$} ()

	;
\end{tikzpicture}

%% file: ccover.tex
	\section{\CCover~for Wait-Only protocols is \pspace-complete}\label{sec:CCover}

 We present here an algorithm to solve the configuration coverability~problem for Wait-Only protocols in polynomial space.

 \subsection{Main ideas}

For the remaining of the section, we fix a Wait-Only protocol $\PP=(Q, \Sigma, \qinit, T)$ and a configuration $C_f \in \mathcal{C}$ to cover, and we let $K=|C_f|$.
The intuition is the following: we (only) keep track of the $K$ processes that will cover $C_f$. Of course, they might need other processes to reach
the desired configuration, if they need to receive messages. That is why we also maintain the set of reachable states along the execution. An
abstract configuration will then be a multiset of $K$ states (concrete part of the configuration) and a set of all the reachable states (abstract part of the configuration).
 Lemma \ref{lemma:copycat-action-state}
ensures that it is enough to know which action states are reachable to ensure that both the concrete part and the action states of the abstract part are coverable at the same time. However, there is a case where this 
abstraction would not be enough: assume that one of the $K$ processes has to send a message, and this message \emph{should not} be received by the other 
$K-1$ processes. This can happen when the message is received by a process in the part of the configuration that we have abstracted away. In that case, 
even if the (waiting) state is present in the set of reachable states, 
Lemma \ref{lemma:copycat-action-state} does not guarantee that the entire configuration is reachable, so the transition to an abstract configuration where none
of the $K-1$ processes has received the message might be erroneous. This is why in that case we need to precisely keep track of the process that will receive the message,
even if in the end it will not participate in the covering of $C_f$. This leads to the definition of the $\abtransSwitch$ transition below. 

%We define the object of \emph{abstract configuration} as a configuration \emph{of fixed size $K$} together with a set of states.
%We also present an abstract transition relation among abstract configurations and prove it to be sound and complete.
%
%The intuition is the following: we (only) keep track of the processes needed to cover $C_f$ (i.e. $K$ processes). However as we will see, those processes are not necessarily the same throughout the abstract execution that we consider. The main reason is when going from a concrete execution to an abstract one, forgetting some processes might lead to some not desired rendez-vous, i.e. consisting of the same sender but a different receiver (as the receiver in the concrete execution might have been forgotten). We will prevent this situation by keeping track of some \emph{receivers} while they are necessary, and then substituting them with the processes needed to cover $C_f$.

%\ars{Pourquoi il reste des ?? à la place de références ?}
This proof is structured as follows: we present the formal definitions of the abstract configurations and semantics in \cref{subsection:in-pspace:definitions}. In \cref{subsec:CCover:in-pspace:completeness}, we present the completeness proof, \cref{subsection:in-pspace:soundness}\ is devoted to prove the soundness of the construction. In the latter, we also give some ingredients to prove an upper bound on the number of processes needed to cover the configuration. In \cref{subsec:Ccover:pspace}\ one can find the main theorem of this section: it states that the \CCover\ problem is in \pspace\ and if the configuration is indeed coverable, it presents an upper bound on the number of processes needed to cover it.
In \cref{subsec:Ccover:pspace-hard}, we prove that this lower bound is tight as the problem is \pspace-hard.

\subsection{Reasoning with Abstract Configurations}\label{subsection:in-pspace:definitions}

We present the abstract configurations we rely on. Let us fix $K = |C_f|$. An \emph{abstract configuration} $\gamma$ is a pair $(M,S)$ where $M$ is a configuration in $\CC$ such that $|M|=K$  and $S \subseteq Q$ is a subset of control states such that $\set{q \in Q \mid M(q)>0} \subseteq S$. We call $M$ the $M$-part of  $\gamma$ and $S$ its $S$-part. We denote by $\Gamma$ the set of abstract configurations and by $\gammainit$ the initial abstract configuration $\gammainit=(\mset{K \cdot \qinit},\set{\qinit})$. An abstract configuration $\gamma=(M,S)$ represents a set of configurations $\interp{\gamma} =\set{C  \in \CC \mid M \preceq C \mbox{ and } C(q) > 0 \mbox{ implies } q\in S}$. Hence in $\interp{\gamma}$, we have all the configurations $C$ that are bigger than $M$
as long as the states holding processes in $C$ are stored in $S$ (observe that this implies that all the states in $M$ appear in $S$).

% In this subsection, we give the formal definitions of abstract configurations and the abstract semantics.

% Let us fix $K = |C_f|$, an abstract configuration $\gamma$ is defined as a tuple $(M,S)$ where $M \in \CC$ such that $|M| = K$, and $S \subseteq Q$. We note $\Gamma$ the set of all abstract configurations, and we denote $\gamma_0$ as the initial abstract configuration $(M_0, \set{\qinit})$ where $M_0 = K\cdot \mset{\qinit}$. We refer to $M$ as the multiset part of $\gamma$ and $S$ as its set part.

% We define the \emph{support} of an abstract configuration $\gamma = (c ,S)$ as the set $\llbracket \gamma \rrbracket = \set{C \mid c \preceq C \wedge \forall q \in Q \text{ such that }C(q) > 0, q\in S}$.

We now define an abstract transition relation for abstract configurations. For this matter, we define three transition relations $\abtransStep$,$\abtransExt$ and  $\abtransSwitch$  and let $\abtrans$ be defined by $\abtransStep\cup\abtransExt\cup\abtransSwitch$.
Let $\gamma=(M,S)$ and $\gamma' =(M',S')$ be two abstract configurations and
$t = (q, \alpha, q')$ be a transition in $T$ with $\alpha=!a$ or $\alpha=!!a$. For $\kappa \in \set{\mathsf{step}, \mathsf{ext}, \mathsf{switch}}$, we have  $\gamma \abtransWUp{t} \gamma'$ iff all the following conditions hold:
\begin{itemize}
\item $S \subseteq S'$, and,
\item for all $p \in S'\setminus S$, either $p=q'$ or there exist $p' \in S$ and $(p',?a,p)$ in $T$, and,
\item one of the following cases is true:
  \begin{itemize}
  \item $\kappa=\mathsf{step}$ and $M \transup{t} M'$. This relation describes a message emitted from the $M$-part of the configuration;
  \item $\kappa=\mathsf{ext}$ and $q \in S$ and $M+\mset{q} \transup{t} M'+\mset{q'}$;
  \item $\kappa=\mathsf{ext}$ and there exists $(p,?a,p')$ in $T$ such that  $q,p \in S$ and $M+\mset{q,p} \transup{t} M+\mset{q',p'}$ (note that in that case $M = M'$). The relation $\mathsf{ext}$ hence describes a message emitted from the $S$-part of the configuration;
  \item $\kappa=\mathsf{switch}$ and $q\in S$ and  $\alpha=!a$ and there exists $t' = (p, ?a, p') \in T$ such that $\mset{p} \preceq M$ and $ \mset{q'} \preceq M'$ and $M - \mset{p} = M' - \mset{q'}$ and $M+\mset{q} \transup{t} M' + \mset{p'}$. This relation describes a sending from a state in the $S$-part of the abstract configuration leading to a rendez-vous with one process in the $M$-part, and a "switch" of processes: we remove the receiver process of the $M$-part and replace it by the sender.   \end{itemize}
\end{itemize}

Note that in any case, $q$, the state from which the message is sent, belongs to $S$.
We then write $\gamma \abtransup{t} \gamma'$ whenever  $\gamma \abtransWUp{t} \gamma'$ for $\kappa \in \set{\mathsf{step}, \mathsf{ext}, \mathsf{switch}}$ and we do not always specify the used transition $t$ (when omitted, it means that there exists a transitions allowing the transition). We denote by $\abtrans^\ast$ the reflexive and transitive closure of $\abtrans$.

\begin{figure}
  \begin{center}
	\input{Figures/example-2}
  \end{center}
	\caption{A Wait-Only protocol  $\PP'$.}\label{figure:example-2}
\end{figure}

\begin{example}\label{example:in-pspace:abstract-semantics}
We consider the Wait-Only protocol $\PP'$ depicted on \cref{figure:example-2} with set of states $Q'$. We want to cover $C_f = \mset{q_3, q_3, q_6}$ (hence $K=3$). 
%	$C$ is indeed coverable, in order to see it, one can consider the following execution:
%	
%	\begin{align*}
%		\mset{\qinit, \qinit, \qinit, \qinit, \qinit} &\trans^+  \mset{\qinit, q_4, q_4, q_4, \qinit}\\
%		& \trans \mset{q_1, q_5, q_5, q_5, \qinit} \\
%		& \trans \mset{}
       %     \end{align*}
In this example, the abstract configuration $\gamma = (\mset{q_2, q_2, q_4}, \set{\qinit, q_1, \linebreak[0] q_2, q_4, q_5, q_6})$ represents all the configurations of 
$\PP'$ with at least two processes on $q_2$ and one on $q_4$, and no process on $q_3$ nor $q_7$. 

Considering the following abstract execution, we can cover $C_f$:
\begin{align*}
&\gammainit \abtransStepUp{(\qinit, !!\tau, q_4)} (\mset{q_4,\qinit, \qinit}, \set{\qinit, q_4})\abtransStepUp{(\qinit, !!\tau, q_4)} (\mset{q_4,q_4, \qinit}, \set{\qinit, q_4}) \\
& \abtransStepUp{(\qinit, !!a, q_1)} (\mset{q_5,q_5,q_1}, \set{\qinit, q_1, q_4, q_5})\abtransExtUp{(q_5, !!c, q_6)}(\mset{q_5, q_5, q_2}, Q'\setminus\set{q_3,q_7})\\
&
\abtransStepUp{(q_2, !b, q_3)} (\mset{q_5,q_5, q_3}, Q')\abtransStepUp{(q_5, !!c, q_6)} (\mset{q_6, q_5, q_3}, Q')\\
&
\abtransSwitchUp{(q_2, !b, q_3)} (\mset{q_3, q_5, q_3}, Q')
\abtransStepUp{(q_5, !!c, q_6)} (\mset{q_3, q_6, q_3}, Q')
\end{align*}
%

%$$
%\begin{array}{l}
%		\gamma_{in}  \abtransStep (\mset{q_4, \qinit, \qinit}, S_1) \abtransStep (\mset{q_4, q_4, \qinit}, S_2)  \abtransStep (\mset{q_5, q_5, q_1}, S_3) \\
%		 \abtransExt (\mset{q_5, q_5, q_2},S_4) \abtransStep (\mset{q_5, q_5, q_3}, S_5) \abtransStep (\mset{q_6, q_5, q_3},S_6) \\
%		 \abtransSwitch (\mset{q_3, q_5, q_3},S_7) \abtransStep (\mset{q_3, q_6, q_3}, S_8)
%\end{array}
%$$
%where $S_1 = S_2 = \set{\qinit, q_4}$; $S_3 = \set{\qinit, q_1, q_4, q_5}$; $S_4 = \set{\qinit, q_1, q_2, q_4,q_5, q_6}$; $S_6 = S_7 = S_8=\set{\qinit, q_1, q_2, q_3, q_4,q_5, q_6,q_7}$. Note that all  $C \in \interp{\mset{q_3, q_6, q_3}, S_8}$, verify $C_f \preceq C$.

It corresponds for instance to the following concrete execution:
\begin{align*}
	\mset{\qinit, \qinit,\qinit} + \mset{\qinit,\qinit}  \trans^+ \mset{q_4, q_4,\qinit} + \mset{q_4,\qinit} \transup{(\qinit, !!a, q_1)} \mset{q_5, q_5,q_1} + \mset{q_5,\qinit}\\
	\transup{(\qinit, !!a, q_1)}  \mset{q_5, q_5,q_1} + \mset{q_5,q_1} \transup{(q_5, !!c, q_6)} \mset{q_5, q_5,q_2} + \mset{q_6,q_1} \\
	 \transup{(q_2, !b, q_3)} \mset{q_5, q_5,q_3} + \mset{q_7,q_1}  \transup{(q_5, !!c, q_6)} \mset{q_6, q_5,q_3} + \mset{q_7,q_2} \\
	\transup{(q_2, !b, q_3)} \mset{q_3, q_5,q_3} + \mset{q_7,q_7} \transup{(q_5, !!c, q_6)}  \mset{q_3, q_6,q_3} + \mset{q_7,q_7}
\end{align*}

The $M$-part of our abstract configuration $\mset{q_6, q_5, q_3}$ reached just before the $\abtransSwitch$ transition does not correspond to the set of 
processes that finally cover $C_f$, at this point of time, the processes that will finally cover $C_f$ are in states $q_2$, $q_5$, and $q_3$. But here we
ensure that the process on $q_6$ will actually receive the $b$ sent by the process on $q_2$, leaving the process on $q_3$ in its state. Once this 
has been ensured, process on $q_6$ is not useful anymore, and instead we follow the process that was on $q_2$ before the sending, hence the $\abtransSwitch$
transition.

\end{example}

The algorithm used to solve \CCover, is then to seek in the directed graph $(\Gamma,\abtrans)$ if a vertex of the form $(C_f,S)$ is reachable from $\gammainit$.

Before proving that this algorithm is correct, we establish the following property.

\begin{lemma}\label{lemma:prop-absconf}
  Let $(M,S)$ and $(M',S')$ be two abstract configurations and $\widetilde S \subseteq Q$ such that $S \subseteq \widetilde S$. We have:
  \begin{enumerate}
  \item $\interp{(M,S)} \subseteq \interp{(M,\widetilde S)}$.
  \item If $(M,S) \abtrans (M',S')$ then there exists $S'' \subseteq Q$ such that $(M,\widetilde S) \abtrans (M',S'')$ and $S' \subseteq S''$.
   \end{enumerate}
 \end{lemma}

 \begin{proof}
The first point is a direct consequence of the definition of $\interp{ }$. For the second point, it is enough to take $S''=S' \cup \widetilde S$ and apply the definition of $\abtrans$. 
 \end{proof}

\subsection{Completeness of the algorithm}
\label{subsec:CCover:in-pspace:completeness}

In this subsection we show that if $C_f$ can be covered then there exists an abstract configuration $\gamma=(C_f,S)$ such that $\gammainit  \abtrans^\ast \gamma$. We use  $\CC_{\geq K}$ to represent the set $\set{C \in \CC \mid |C| \geq K}$ of configurations with at least $K$ processes  and $\CC_{=K}$ the set $\set{C \in \CC \mid |C|= K}$ of configurations with exactly $K$ processes. This first lemma shows the completeness for a single step of our abstract transition relation (note that we focus on the $M$-part, as it is the one witnessing $C_f$ in the end).

\begin{lemma}\label{lemma:in-pspace:completeness-local}
  Let $C, C' \in \CC_{\geq K}$ and  $t \in T$ such that $C \transup{t} C'$. Then for all $M' \in \CC_{=K}$ such that $M' \preceq C'$, there exists $M \in \CC_{=K}$ and $S'
  \subseteq Q$ such that $(M,S) \abtrans (M',S')$ with $S=\set{q \in Q \mid C(q)>0}$, $C \in \interp{(M,S)}$ and $C' \in \interp{(M',S')}$.
  \end{lemma}

\begin{proof}
Let $M' \in \CC_{=K}$ such that $M' \preceq C'$. We assume that $t=(q,\alpha,q')$ with $\alpha \in \set{!a,!!a}$. We let $S=\set{p \in Q \mid C(p)>0}$ and $S'=S \cup \set{p \in Q \mid C'(p)>0}$. By definition of $S$ and $S'$, for all $p \in S'\setminus S$, either $p=q'$ or there exists $p' \in S$ and $(p',?a,p)$ in $T$. In fact, let $p \in S'\setminus S$ such that $p \neq q'$. Since  $C \transup{t} C'$, we have necessarily that there exist $p' \in Q$ such that $C(p')>0$ (hence $p'\in S$), and $(p',?a,p)$ in $T$. We now reason by a case analysis to determine $M \in \CC_{=K}$ such that $(M,S) \abtrans (M',S')$ and $C \in \interp{(M,S)}$. The different cases are: (i) $\alpha = !!a$, (ii) $\alpha = !a$ and the message is \emph{not received}, (iii) $\alpha = !a$ and the message is received by a process.
Case (iii) as it exhibits the most different abstract behaviours. 
	\begin{enumerate}[(i)]
		\item if $\alpha = !!a$, then by definition of $\trans$, $C = \mset{q_1, \dots, q_n, q}$ and $C' = \mset{q'_1, \dots, q'_n, q'}$ such that for all $1 \leq i \leq n$, 
		either $(q_i, ?a, q'_i) \in T$ or $a \nin R(q_i)$ and $q_i = q'_i$. We get the two following disjoint cases:
		\begin{itemize}
			\item $M'(q')=0$. In this case, $M' = \mset{q'_{i_1}, q'_{i_2}, \dots q'_{i_K}}$ where $1 \leq i_j \leq n$ for all $1 \leq j \leq K$ and $i_j\neq i_\ell$ if $j \neq \ell$. Let $M= \mset{q_{i_1}, q_{i_2}, \dots q_{i_K}}$. Note that we have $M \preceq C$, hence $C \in \interp{(M,S)}$. Using the definition of $\trans$, we have as well  $M+\mset{q} \trans M'+\mset{q'}$. Furthermore since $C(q)>0$, we have $q \in S$ by definition of $S$. Applying the definition of $\abtransExt$ we have $(M,S) \abtransExtUp{t} (M',S')$. 
			
			\item $M'(q')>0$. In this case, $M' = \mset{q'_{i_1}, q'_{i_2}, \dots q'_{i_{K-1}}, q'}$  where $1 \leq i_j \leq n$ for all $1 \leq j < K$ and $i_j\neq i_\ell$ if $j \neq \ell$. Let $M= \mset{q_{i_1}, q_{i_2}, \dots q_{i_{K-1}},q}$. Note that we have $M \preceq C$, hence  $C \in \interp{(M,S)}$. Using the definition of $\trans$, we have as well  $M \trans M'$. Applying the definition of $\abtransStep$, we get $(M,S) \abtransStepUp{t} (M',S')$.
		\end{itemize}
		
		\item if $\alpha = !a $ and the message is \emph{not} received (i.e. it is a non blocking sending), then using the definition of $\trans$, we have $C' = C -\mset{q} +\mset{q'}$ and $a \nin R(p)$ for all $p \in Q$ such that $(C-\mset{q})(p) > 0$. We obtain the two following disjoint cases:
		\begin{itemize}
			\item $M'(q')=0$. Since $C'(q')>0$ and $M' \preceq C'$, we deduce that $C'=M'+ \mset{q'}+M_2$ for some multiset $M_2$. By definition of $C'$, we have as well $C=C'+\mset{q}-\mset{q'}$, hence $C=M'+\mset{q}+M_2$. We deduce that $M' \preceq C$ and hence $C\in \llbracket (M', S)  \rrbracket$. Furthermore we have $M' + \mset{q} \trans M' +\mset{q'}$ as $a\nin R(q)$ for all states $q\in Q$ such that $M'(q) > 0$. Consequently, $(M', S) \abtransExtUp{t} (M',S')$.
			\item  $M'(q')>0$. Since $M'\preceq C'$, we have $C'=M'+M_2$ for some multiset $M_2$. Let $M=M'+\mset{q}-\mset{q'}$. We have hence $C=C'+\mset{q}-\mset{q'}=M'+\mset{q}-\mset{q'}+M_2=M+M_2$. Hence $M \preceq C$ and $C \in \interp{(M,S)}$. Furthermore we have that $a \nin R(p)$ for all $p \in Q$ such that $(M-\mset{q})(p) > 0$). We hence deduce that $M \transup{t} M'$. Applying the definition of $\abtransStep$, we get $(M,S) \abtransStepUp{t} (M',S')$.
		\end{itemize}

\item if $\alpha = !a$ and the message is received by a process (i.e. it is a rendez-vous), denote by $(p, ?a, p')$ the reception transition issued between $C$ and $C'$. Using the definition of  $\trans$, we get $C' = C - \mset{q, p} + \mset{q', p'}$. We consider the four following disjoint cases:
		\begin{itemize}
			\item $M'(q')=0$ and $M'(p')=0$.  Since $\mset{q',p'} \preceq C'$ and $M' \preceq C'$, we get that $C'=M'+\mset{q',p'}+M_2$ for some multiset $M_2$. We deduce that $C=M'+\mset{q,p}+M_2$. This allows us to deduce that $M' \preceq C$ and consequently $C \in \interp{(M',S)}$. Moreover, $M' + \mset{p,q} \transup{t} M' + \mset{q',p'}$ and $q,p \in S$. Hence $(M',S) \abtransExtUpSmall{t} (M',S')$.
			
			\item $M'(q')=0$ and $M'(p')>0$. In that case $C'=M'+\mset{q'}+M_2$ for some multiset $M_2$. Let $M=M'-\mset{p'}+\mset{p}$. We have then $C=C' + \mset{q, p} - \mset{q', p'}=M'+\mset{q'}+M_2+\mset{q, p} - \mset{q', p'}=M'+M_2-\mset{p'}+\mset{p}+\mset{q}=M+\mset{q}+M_2$. This allows us to deduce that $M \preceq C$ and consequently $C \in \interp{(M,S)}$. Furthermore $q \in S$ and $M+\mset{q} \transup{t} M'+\mset{q'}$. Hence $(M,S) \abtransExtUpSmall{t} (M,S')$.

			\item $M'(q')>0$ and $M'(p')=0$. In that case $C'=M'+\mset{p'}+M_2$ for some multiset $M_2$. 
			Let $M=M'-\mset{q'}+\mset{p}$. We have then $C=C' + \mset{q, p} - \mset{q', p'}=M'+\mset{p'}+M_2+\mset{q, p} - \mset{q', p'}=M'+M_2-\mset{q'}+\mset{p}+\mset{q}=M+\mset{q}+M_2$. This allows us to deduce that $q \in S$ and $M \preceq C$ and consequently $C \in \interp{(M,S)}$. We also have $\mset{p} \preceq M$ and $\mset{q'} \preceq M'$ and $M-\mset{p}=M'-\mset{q'}$ and $M+\mset{q} \transup{t} M'+\mset{p'}$. Hence $(M, S) \abtransSwitchUpSmall{t} (M', S')$. Observe that we need to use the $\abtransSwitch$ transition relation in this case. Assume that $C(s)>0$ for some state $s\in S$ such that $(s,?a,s')\in T$, and 
			that any configuration in $\CC_{=K}$ such that $M\preceq C$ contains such state $s$. Then, applying $\abtransStep$ to such a multiset $M$
			will take away the process on state $s$ and will lead to an abstract configuration with $M'\not\preceq C'$.
			
			\item $M'(q')>0$ and $M'(p')>0$. In that case $C'=M'+M_2$  for some multiset $M_2$. Let $M=M'-\mset{p',q'}+\mset{p,q}$. We have then $C=C' + \mset{q, p} - \mset{q', p'}=M+M_2$. This allows us to deduce that $M \preceq C$ and consequently $C \in \interp{(M,S)}$ and that $M \transup{t} M'$. Hence $(M,S) \abtransStepUpSmall{t} (M',S')$.
	\end{itemize}
	\end{enumerate}
\end{proof}

The two previous lemmas allow us to establish completeness of the construction, by a simple induction on the length of the considered execution. 

\begin{lemma}\label{lemma:in-pspace:completeness-final}
	Let $C_{\textit{in}} \in \Cinit$ and $C \in \CC_{\geq K}$ such that $C_\textit{in} \trans^\ast C$. For all $M \in \CC_{= K}$ such that $M \preceq C$ there exists $S \subseteq Q$ such that  $C \in \interp{(M,S)}$ and $\gammainit \abtrans^\ast (M,S)$.
\end{lemma}
\begin{proof}
	Suppose we have $C_0 \transup{t_1} C_1 \transup{t_2} \ldots  \transup{t_n} C_n$ with $C_0=C_\textit{in}$ and $C_n=C$. We show by induction on $0\leq i\leq n$ that for all $M_i \in \CC_{= K}$ such that $M_i \preceq C_i$ there exists $S_i \subseteq Q$ such that   $C_i \in \interp{(M_i,S_i)}$ and $\gammainit \abtrans^\ast (M_i,S_i)$. 
	
	For $i=0$, we have $C_\textit{in}=\mset{|C|\cdot \qinit}$ and $\mset{K\cdot \qinit}$ is the unique configuration of size $K$ smaller than $C_\textit{in}$.  Since $\gammainit=(\mset{K\cdot \qinit},\set{\qinit})$, we have $C_\textit{in} \in \interp{\gammainit}$. 
	
	Assume now that the property holds for $0 \leq i <n$. Let $M_{i+1} \in \CC_{= K}$ such that $M_{i+1} \preceq C_{i+1}$. Since $C_i \transup{t_{i+1}} C_{i+1}$, by applying Lemma \ref{lemma:in-pspace:completeness-local}, there exists $M_i \in \CC_{=K}$, $S_{i+1}
	\subseteq Q$ such that $(M_i,S_i) \abtrans (M_{i+1},S_{i+1})$ where $S_i=\set{q \in Q \mid C_i(q)>0}$, $C_i \in \interp{(M_i,S_i)}$ and $C_{i+1} \in \interp{(M_{i+1},S_{i+1)}}$. By induction hypothesis, there exists $S'_i \subseteq Q$ such that $C_i \in \interp{(M_i,S'_i)}$ and $\gamma_{in} \abtrans^\ast (M_i,S'_i)$. But since $C_i \in \interp{(M_i,S'_i)}$ and $S_i=\set{q \in Q \mid C_i(q)>0}$, we have $S_i \subseteq S'_i$. Using Lemma \ref{lemma:prop-absconf}.2, we deduce that there exists $S'_{i+1}$ such that $S_{i+1} \subseteq S'_{i+1}$ and $(M_i,S'_i) \abtrans (M_{i+1},S'_{i+1})$.  Hence  $\gammainit \abtrans^\ast (M_{i+1},S'_{i+1})$ and  thanks to  Lemma \ref{lemma:prop-absconf}.1, $C_{i+1} \in \interp{(M_{i+1},S'_{i+1})}$.
\end{proof}

\subsection{Soundness of the algorithm}\label{subsection:in-pspace:soundness}
We now prove that if we have $\gammainit \abtrans^\ast (M,S)$ then the configuration $M$ can be covered. We first establish that the $S$-part of a reachable abstract configuration stores only states that are reachable in a concrete execution.

\begin{lemma}\label{lemma:in-pspace:soundness-states-in-S-coverable}
	If $\gamma =(M,S)$ is an abstract configuration such that $\gammainit \abtrans^\ast \gamma$, then all states $q \in S$ are coverable.
\end{lemma}
\begin{proof}
We suppose that we have $\gamma_{in}=\gamma_{0} \abtrans \gamma_1\abtrans \ldots \abtrans \gamma_n=(M,S)$ and we prove this lemma by induction on $n$, the length of the abstract execution.

	\textbf{Case $n=0$:} In that case $(M,S) = \gamma_0= (K\cdot\mset{\qinit}, \set{\qinit})$, as $\qinit$ is trivially coverable, the property holds.
	
	\textbf{Case $n>0$:} We assume that the property holds for all $0 \leq m < n$ and consider the abstract execution $\gamma_{0} \abtransup{t_1} \gamma_1\abtransup{t_2} \ldots \abtransup{t_n} \gamma_n$ where $\gamma_0=\gammainit$ and $\gamma_i=(M_i,S_i)$ for all $0 \leq i \leq n$. 
	Let $p\in S_n$. If $p\in S_{n-1}$, then by induction hypothesis, $p$ is coverable. Otherwise, $p\in S_n\setminus S_{n-1}$, and let $t_n=(q,\alpha, q')$ with $\alpha\in \set{!a, !!a\mid a\in \Sigma}$.
	By definition of $\abtrans$,  $q\in S_{n-1}$ and 
	\begin{itemize}
	\item either $q'=p$, and by induction hypothesis, there exists an initialized execution $C_0\trans^\ast C$ with $C(q)>0$ and in that
	case, $C_0\trans^\ast C\transup{t} C'$ for a configuration $C'$ such that $C'(p)>0$ and $p$ is coverable. 
	\item or $\alpha\in\set{!a,!!a}$ for some $a\in \Sigma$ and  $(p',?a,p)\in T$, with $p'\in S_{n-1}$. By induction hypothesis, both $q$ and $p'$ are coverable,
	with $q\in{Q_A}$ and $p\in\waitingset{Q}$. By Lemma \ref{lemma:copycat-action-state}, there exists an execution $C_0 \trans^\ast C$ such that $C(q) \geq 1$ and $C(p') \geq 1$. We then have $C\transup{t_n} C'$ with $C'(p)>0$ (if $\alpha=!a$ then the process on $p'$ can receive the message $a$ and move to $p$, and
	if $\alpha=!!a$ then the process on $p'$ will necessary receive the broadcast and move to $p$), and $p$ is coverable. 
	\end{itemize}
%	
%	Let $q \in S_n$. If $q \in S_{n-1}$ then by induction hypothesis $q$ is coverable. We hence assume that $q \not\in S_{n-1}$. 
%	By definition of the transition relation $\abtrans$, 
%	there exists $p \in S_{n-1}$ and $a \in \Sigma$ such that 
%	$T\cap \set{(p,!a,q),(p,!!a,q),(p,?a,q)} \neq \emptyset$. We have the two following cases:
%\begin{itemize}
%	\item If $T\cap \set{(p,!a,q),(p,!!a,q)} \neq \emptyset$, then apply the induction hypothesis on $p$ in order to get an execution $C_0 \trans^\ast C$ such that $C(p) > 0$. Then, by definition of $\trans$, there exists $C'$ such that $C \transup{t} C'$ and $C'(q) > 0$, which concludes this case. 
%	\item If $T\cap \set{(p,?a,q)} \neq \emptyset$. By construction of $\abtrans$, there exists $t' = (p', \beta, q') \in T$ with $q'\in S$ and $\beta=!a$ or $\beta=!!a$. By induction hypothesis, $p'$ and $p$ are coverable, using \cref{lemma:copycat-action-state}, we get that there exists an execution $C_0 \trans^\ast C$ such that $C(p) \geq 1$ and $C(p') \geq 1$. By definition of $\trans$, there exists $C'$ such that $C \transup{t'} C'$ with $C'(q) >0$: if $\beta = !a$, a rendez-vous occurs between one process on $p'$ and one process on $p$, and if $\beta = !!a$, a broadcast occurs and all processes on states with an outgoing reception of $a$ receives the message (including the process on $p$). Hence, $q$ is coverable.
%\end{itemize}
%\nas{c'est moche tout seul en haut de la page :-)}
\end{proof}

The next lemma establishes soundness of the algorithm. Moreover, it gives an upper bound on the minimal number of processes needed to cover a configuration. 
\begin{lemma}\label{lemma:in-pspace:soundness-final}
	Let $(M,S)$ be an  abstract configuration such that $\gammainit \abtrans \gamma_1\abtrans \ldots \abtrans \gamma_n=(M,S)$. Then, there exist $C_\textit{in} \in \Cinit$, $C \in \CC$ such that $M \preceq C$ and $C_\textit{in} \trans^\ast C$. Moreover, $|C_\textit{in}|=|C| \leq K+2^{|Q|}\times n$.
\end{lemma}

\begin{proof}
  We reason by induction on $n$, the length of the abstract execution.
  
  \textbf{Case $n=0$:}  The property trivially holds for $C = \mset{K \cdot \qinit}$.

  \textbf{Case $n>0$:} We assume that the property holds for all $0 \leq m < n$ and consider the abstract execution $\gamma_{0} \abtransup{t_1} \gamma_1\abtransup{t_2} \ldots \abtransup{t_n} \gamma_n$ where $\gamma_0=\gammainit$ and $\gamma_i=(M_i,S_i)$ for all $0 \leq i \leq n$. By induction hypothesis, we know that there exist $C_\textit{in} \in \Cinit$, $C_{n-1} \in \CC$ such that $M_{n-1} \preceq C_{n-1}$ and $C_\textit{in} \trans^\ast C_{n-1}$ and $|C_\textit{in}|=|C_{n-1}| \leq K+2^{|Q|}\times (n-1)$. If $M_n=M_{n-1}$ then the property holds. Assume now that $M_n \neq M_{n-1}$. We let $t_n=(q, \alpha, q')$ with $\alpha=!a$ or $\alpha=!!a$. By definition of $\abtrans$, we know that $S_{n-1} \subseteq S_n$ and that $q \in S_{n-1}$. Thanks to Lemma \ref{lemma:in-pspace:soundness-states-in-S-coverable}, $q$ is coverable. We now perform a case analysis:
\begin{itemize}
\item Assume $\gamma_{n-1} \abtransStepUpSmall{t_n} \gamma_{n}$. Then $M_{n-1} \transup{t_n} M_{n}$, and since $M_{n-1} \preceq C_{n-1}$, we have $C_{n-1}=M_{n-1}+ M$ for some multiset $M$.
  \begin{itemize}
  \item If $\alpha= !!a$, then $M_{n-1} + M=\mset{q_1, \dots q_{K-1}, q} + \mset{p_1, \dots, p_L}$ and $M_{n}=\mset{q'_1, \dots q'_{K-1}, q'}$ where for all $1 \leq i \leq K-1$, either $(q_i, ?a, q'_i) \in T$ or $a \nin R(q_i)$ and $q_i = q'_i$. For each $1 \leq i \leq L$, define $p'_i$ as $p'_i = p_i$ if $a \nin R(p_i)$ or $p'_i$ is such that $(p_i, ?a, p'_i) \in T$. If we let $M' = \mset{p'_1, \dots, p'_L}$ and $C_{n} = M_n+M'$, we have by definition that $C_{n-1} \transup{t_n} C_{n}$ with
  $M_n\preceq C_n$.
  
   \item If $\alpha=!a$ and $(M_{n-1}-\mset{q})(p)>0$ for some $p\in Q$ such that $(p,?a,p')\in T$ (i.e., a rendez-vous occurred), it holds that $M_{n-1} + M \transup{t} M_{n} + M$ and we choose $C_{n} = M_{n} + M$. 
   \item If $\alpha = !a$ and $(M_{n-1}-\mset{q})(p) = 0$ for all $p \in Q$ such that  $a \in R(p)$ (i.e. it was a non-blocking sending of a message), then either there exists $(p, ?a, p') \in T$ such that $M(p) > 0$, or for all $p\in Q$ such that $M(p) > 0$, $a \nin R(p)$. In the first case, a rendez-vous will occur in the execution of $t_n$ over $C_{n-1}$, and we have $M_{n-1} + M \transup{t_n} M_{n} + M - \mset{p} + \mset{p'}$. We then let $C_{n} = M_{n} + M - \mset{p} + \mset{p'}$. In the latter case, $M_{n-1} + M \transup{t_n} M_{n} + M$ and with $C_{n} = M_{n} + M$. In both cases, we have , 
   $C_{n-1}\transup{t} C_n$ and $M_n\preceq C_n$. 
   \end{itemize}
   In all cases, we have $C_{in} \trans^\ast C_{n-1} \trans C_n$ and $|C_{in}|=|C_{n}|=|C_{n-1}|\leq  K+2^{|Q|}\times (n-1)\leq K+2^{|Q|}\times n$.

\item Assume $\gamma_{n-1} \abtransExtUpSmall{t_n} \gamma_{n}$ or $\gamma_{n-1} \abtransSwitchUpSmall{t_n} \gamma_{n}$. As $q$ is coverable, from Corollary \ref{cor:scover-bound}, there exists an execution $C^q_\textit{in} \trans^\ast C^q$ such that $C^q_\textit{in}\in \Cinit$ and $C^q(q) > 0$ and $|C^q_\textit{in}| \leq 2^{|Q|}$. 
From Lemma \ref{lem:P0}, we have the following execution: $C_\textit{in}+C^q_\textit{in} \trans^\ast C_\textit{in}+ C^q$. 
%Since $C_\textit{in}=\mset{L.\qinit}$ for some $L \in \nat$ and since $\qinit$ is an action state, we deduce that we also have the following execution $C_\textit{in}+C^q_\textit{in} \trans^\ast C_\textit{in}+ C^q$ (it mimics the execution $C^q_{in} \trans^\ast C^q$ and has no effect on the processes in $C_{in}$ which are all in action states)\nas{C'est pas exactement le \cref{lem:P0}?}. 
Next, from Lemma \ref{lem:P1}, by taking $\constantM = 1$ and $\tilde{C}_0 = C_\textit{in} + C^q$, we have an execution $C_{in} + C^q \trans^\ast C_{n-1} + C'^q$ where $C'^q(q)>0$. %\nas{on a besoin de $|C^q|=|C'^q|$ parce que ce n'est pas dans l'énoncé du \cref{lem:P1} (pour le moment)}
%This time, we follow the same step as in the execution $C_{in} \trans^\ast C_{n-1}$ and use the fact that $q$ is an action state (however we do not have necessarily $C^q=C'^q$ as some states in $C^q$ could be waiting states impacted by the execution)\nas{\cref{lem:P1} non??}. 
We deduce that  $C_{n-1}+C'^q=M_{n-1}+\mset{q}+M$ for some multiset $M$. We now proceed with a case analysis:
\begin{itemize}
\item Case $\gamma_{n-1}\abtransExtUpSmall{t_n} \gamma_{n}$ where $M_{n-1} + \mset{q} \transup{t_n} M_{n} + \mset{q'}$. By definition of $\trans$, we also have $M_{n-1} + \mset{q} + M\trans M_{n}+ \mset{q'} + M'$ for some multiset $M'$. Letting $C_{n} = M_{n} + \mset{q'} + M'$ gives us that that $C_{in}+C^q_{in} \trans^\ast M_{n-1}+\mset{q}+M \trans C_n$.
%\item Case $\gamma_{n-1}\abtransExtUp{t_n} \gamma_{n}$ where $M_{n-1} + \mset{q,p} \transup{t} M_{n} + \mset{q',p'}$ for some $(p, ?a, p') \in T$. Oberve that if $\alpha = !a$, then $M_n = M_{n+1}$, and we supposed $M_n \neq M_{n-1}$. Hence, $\alpha = !!a$, and by definition of $\trans$, we have as well that $M_{n-1} + \mset{q} \trans M_{n}+ \mset{q'}$, and we can reuse the proof of the previous case.\lug{new}
\item Case $\gamma_{n-1} \abtransSwitchUpSmall{t_n} \gamma_{n}$: by definition of $\abtransSwitch$, we know that there exists $(p, ?a,p') \in T$ with
$p\in S_{n-1}$ such that  $M_{n-1} = M'+  \mset{p}$ and that $M' + \mset{p} + \mset{q} \transup{t_n} M' + \mset{q'} + \mset{p'}$ and $M_n=M'+\mset{q'}$. Furthermore, by definition of  $\trans$, we have $M' + \mset{p,q} + M \trans M' + \mset{p',q'} + M$. Hence setting $C_{n}= M' + \mset{p',q'} + M=M_n+\mset{p'}+M$ gives us that that $C_{in}+C^q_{in} \trans^\ast M_{n-1}+\mset{q}+M \trans C_n$, with $M_n \preceq C_{n}$.
\end{itemize}

In both cases, we have shown that there exists $C_n$ such that $M_n \preceq C_n$ and  $C_{in}+C^q_{in} \trans^\ast C_n$. Furthermore we have that $|C_n|=|C_{in}+C^q_{in}|\leq K+2^{|Q|}\times (n-1) + 2^{|Q|}\leq  K+2^{|Q|}\times n$.	
\end{itemize}	
\end{proof}

\subsection{Upper Bound}\label{subsec:Ccover:pspace}

Using Lemmas \ref{lemma:in-pspace:completeness-final} and \ref{lemma:in-pspace:soundness-final}, we know that there exists $C\in \mathcal{I}$ and $C' \in \CC$ such that $C \trans^\ast C'$ and $C_f \preceq C'$ iff there exists an abstract execution $\gammainit \abtrans \gamma_1 \abtrans \cdots \abtrans \gamma_n$ with $\gamma_n=(C_f,S)$ for some $S \subseteq Q$, hence the algorithm consisting in deciding reachability of a vertex of the form $(C_f,S)$ from $\gammainit$
in the finite graph $(\Gamma,\abtrans)$ is correct. Note furthermore that the number of abstract configurations $|\Gamma|$ is bounded by $|Q|^{|C_f|}\times 2^{|Q|}$. 
%As a consequence, to solve \CCover, we can seek in  the directed graph $(\Gamma,\abtrans)$ if a vertex of the form $(C_f,S)$ is reachable from $\gammainit$.
 As the reachability of a vertex in a graph is \textsc{NL}-complete, this gives us a \textsc{NPSpace} procedure, which leads to a \textsc{Pspace} procedure thanks to Savitch's theorem.

\begin{theorem}\label{thm:in-pspace:final}
	\CCover\ for Wait-Only protocols is in \pspace.
\end{theorem}

\begin{remark}
Thanks to Lemma \ref{lemma:in-pspace:soundness-final}, we know that  $\gamma_{in} \abtrans \gamma_1 \abtrans \cdots \abtrans \gamma_n$ with $\gamma_n=(C_f,S)$ iff there exists $C_{in} \in \mathcal{I}$, $C \in \CC$ such that $C_f \preceq C$ and $C_{in} \trans^\ast C$ and $|C_{in}|=|C| \leq K+2^{|Q|}\times n$. But due to the number of abstract configurations, we can assume that $n \leq 2^{|Q|} \times |Q|^{|C_f|}$ as it is unnecessary in the abstract execution $\gamma_{in} \abtrans \gamma_1 \abtrans \cdots \abtrans \gamma_n$ to visit twice the same abstract configuration. Hence the configuration $C_f$ is coverable iff there is $C \in \mathcal{I}$ and $C' \in \CC$ such $C \trans^\ast C'$ and $C_f \preceq C'$ and $|C|=|C'|\leq K+ 2^{|Q|} \times 2^{|Q|} \times |Q|^{|C_f|}$. 
\end{remark}

	\subsection{Lower Bound}\label{subsec:Ccover:pspace-hard}
%	\lugtext{faut il ajouter un petit paragraphe qui explique comment la réduction se passe avant les preuves formelles ?}
	To prove \pspace-hardness of the \CCover~problem for Wait-Only protocols, we reduce the intersection non-emptiness problem for deterministic finite automata, which is known
	to be \pspace-complete \cite{Kozen77}. 
	%In fact, the reduction uses 
	%We will in fact prove \pspace-hardness of Wait-Only Br-Networks, which gives \pspace-hardness 
	%	for the whole class of Wait-Only Br+Nb-RDV-Networks.\nas{je ne trouve pas ces noms très jolis... en plus il n'y a pas de macros :-))}
	The \pspace-hardness in fact holds when considering Wait-Only protocols without any (non-blocking) rendez-vous transitions, i.e. transitions of the form $(q, !a, q')$.
	
	Let $\mathcal{A}_1, \dots, \mathcal{A}_n$ be a list of deterministic finite and \emph{complete} automata with $\mathcal{A}_i = (\Sigma, Q_i, q_{i}^0, \set{q_i^f}, \Delta_i)$ for all $1 \leq i \leq n$. Observe that we restrict our reduction to automata with a unique accepting state, which does not change the complexity of the problem. 
	We note $\Sigma^\ast$ the set of words over the finite alphabet $\Sigma$ and $\Delta^\ast_i$ the function extending $\Delta_i$ to $\Sigma^\ast$, i.e, for all $q \in Q_i$, $\Delta_i^\ast(q, \varepsilon) = q$, and for all $w \in \Sigma^\ast$ and $a \in \Sigma$, $\Delta^\ast_i(q, wa) = \Delta_i(\Delta^\ast_i(q,w), a)$. 
	
		\begin{figure}
		\begin{center}
			\input{Figures/prot-pspace-hardness.tex}\caption{Protocol \PP\ for \pspace-hardness of \CCover.}\label{fig:pspace-hard}
		\end{center}
	\end{figure}

	We build the protocol $\PP$ with set of states $Q$, displayed in \cref{fig:pspace-hard}~where $\PP_i$ for $1 \leq i \leq n$ is a protocol mimicking the behaviour of the
	automaton $\mathcal{A}_i$: $\PP_i = (Q_i, \Sigma, q_{i}^0, T_i)$, with $T_i = \set{(q, ?a, q') \mid (q, a ,q') \in \Delta_i}$. 
	%\cup \set{q, ?\textit{end}, q^i)\mid q\in F_i}$. 
	Moreover, from any
	state $q\in \bigcup_{1\leq i\leq n} Q_i$, there is an outgoing transition $(q, ?\textit{go}, q_\textit{fail})$. These transitions are depicted by the outgoing transitions 
	labelled by $?\textit{go}$ from the orange rectangles.
	
	Note that $\PP$ is Wait-Only as all states in $\PP_i$ for all $1 \leq i \leq n$ are waiting states and the only action states are $\qinit$ and $q_s$. We show that 
	$\bigcap_{1\leq i\leq n} L(\mathcal{A}_i)\neq\emptyset$ if and only if there is an initial configuration $C\in \Cinitprot{\PP}$ and a configuration $C'\in\CCprot{\PP}$
	such that $C\trans^\ast C'$ and $C_f\preceq C'$ with $C_f=\mset{q^f_1, \dots, q^f_n}$.
	%\nas{Il faut ajouter un état $q_f$ à $P$, un état $q_f^i$ à chaque $P^i$ et ajouter une transition $(q,!!\textit{end}, q_f)$ et une transition $(q_{f,i},?\textit{end}, q_f^i)$ pour tout $q_{f,i}\in F_i$. Il faut aussi rajouter une transition $(q_f, ?m, q_\textit{fail})$ et des transitions $(q_f^i, ??m, q_\textit{fail})$ pour tout $i$ et pour tout $m$}.

	The idea is to synchronize (at least) $n$ processes into simulating the $n$ automata. To this end, we need an additional (leader) process that will broadcast a message $\textit{go}$, which will be received by the $n$ processes, leading each of them to reach a different automaton initial state. Then, the leader process will broadcast a word letter by letter. Since the automata are all complete, these broadcast will be received by all the processes that simulate the automata, 
	mimicking an execution. If the word belongs to all the automata languages, then each process simulating the automata ends the simulation on the
	unique final state of the automaton. Note that if the leader process broadcasts the message $\textit{go}$ a second time, then all the processes simulating the automata stop their simulation and reach the state $q_\textit{fail}$.
	
We now present the formal proofs.
	Assume that there exists a word $w=a_1\dots a_k\in \bigcap_{1\leq i\leq n} L(\mathcal{A}_i)$, i.e., $q_i^f=\Delta^*(q_i^0, a_1 \linebreak[0]\dots a_k)$ for all $1\leq i\leq n$. 
Then take $C=\mset{(n+1).\qinit}$. There exists an execution $C\trans^\ast C'$ with $C'=\mset{q_s, q_1^f, \dots, q_n^f}\succeq C_f$. 
This execution is $C\transup{(\qinit, !!\tau, q_s)}\tilde C_1\transup{(\qinit,!! \tau, q_1)}\tilde C_2\dots \transup{(\qinit,!! \tau, q_n)}\tilde C_0\transup{(q_s, !!\textit{go}, q_s)}C_0\transup{(q_s,!!a_1,q_s)} C_1\transup{(q_s, !!a_2, q_s)} C_2\dots \linebreak[0] \transup{(q_s, !!a_k,q_s)} C'$. 

One can check that $C_0(q_i ^0)=1$ for all $1\leq i\leq n$, and hence, by definition of $(T_i)_{1\leq i\leq n}$, $C'(q_i^f)=1$ for all $1\leq i\leq n$. Hence 
$C'=\mset{q_s,q_1^f, \dots, q_n^f} \succeq C_f$. 

Reciprocally, assume that there exists an initial configuration $C$ and an execution $C\trans^\ast C'$ with $C'\succeq C_f$. We first make easy observations about the executions of this
protocol. 

\begin{observation} Let $C_1,C_2\in\CCprot{\PP}$. We write $C_1\xrightarrow\tau{}\!\!^* C_2$ if there exists a sequence of $k$ transitions 
	$C_1\transup{t_1}\dots \transup{t_k} C_2$ such that $t_i\in \{(\qinit, !!\tau,q)\mid q\in \{q_s,q_1,\dots, q_n\}\}$ for all $1\leq i\leq k$. Then $C_2(q)\geq C_1(q)$
	for all $q\in \{q_s,q_1,\dots, q_n\}$, $C_2(\qinit)\leq C_1(\qinit)$, and $C_2(q)=C_1(q)$ for all other $q\in Q$. 
\end{observation}

\begin{observation}
	Let $C_1,C_2\in\CCprot{\PP}$ such that $C_1\trans^\ast C_2$ with no transition $(q_s,!!\textit{go}, q_s)$. If there is some $1\leq i\leq n$ with $C_1(q)=0$
	for all $q\in Q_i$, then $C_2(q)=0$ for all $q\in Q_i$. 
\end{observation}

\begin{observation} Let $C_1,C_2\in\CCprot{P}$ such that $C_1\transup{(q_s,!!\textit{go},q_s)} C_2$. Then, $C_2(q_s)=C_1(q_s)>0$, and $C_2(q)=0$ for all $q\in \bigcup_{1\leq i\leq n} (Q_i\setminus{\set{q^0_{i}}})$. 
\end{observation}

We can deduce from these observations that the execution $C\trans^\ast C'$ can be decomposed in $C\trans^\ast \hat{C}\transup{(q_s,!!\textit{go},q_s)} C_0\trans^\ast C'$. Indeed, since 
$C(q)=0$ for all $q\in \bigcup_{1\leq i\leq n} Q_i$, if no transition $(q_s,!!\textit{go}, q_s)$ appears in the execution, Observation 2 allows to conclude that $C'(q^f_i)=0$ for all $1\leq i\leq n$, which contradicts the fact that $C'\succeq C_f$. Assume now that this transition is the last transition where action $\textit{go}$ is sent, i.e., $C_0\trans^\ast C'$ has no transition $(q_s,!!\textit{go}, q_s)$.
By Observation 3, $C_0(q_s)>0$ and $C_0(q)=0$ for all $q\in \bigcup_{1\leq i\leq n} (Q_i\setminus \{q_i^0\})$. By Observation 2, we also deduce that $C_0(q_i^0)>0$ for all $1\leq i\leq n$. Otherwise, if there exists $1\leq j\leq n$ such that $C_0(q_j^0) = 0$, then $C'(q_j^f)=0$ which is a contradiction with $C'\succeq C_f$. 

Now the execution $C'_0\trans^\ast C'$ is of the form $C_0\xrightarrow\tau{}\!\!^*C'_0\transup{(q_s, !!a_1, q_s)} C_1\xrightarrow\tau{}\!\!^*C'_1\transup{(q_s, !!a_2, q_s)} C_2\dots
\transup{(q_s, !!a_k,q_s)}C_k\xrightarrow\tau{}\!\!^*C'$. 
Using Observation 1, we can obtain a new execution $C_0\transup{(q_s, !!a_1, q_s)} \tilde C_1\transup{(q_s, !!a_2, q_s)} \tilde C_2\dots
\transup{(q_s, !!a_k,q_s)}\tilde C_k$ such that for all $1\leq j\leq k$, we let $\tilde C_j(q_s)=C_0(q_s)$, $\tilde C_j(q_i)=C_0(q_i)$ for all $1\leq i \leq n$, $\tilde C_j(\qinit)=C_0(\qinit)$ and $\tilde C_j(q)= C_j(q)$ for all other $q$. Moreover, $\tilde C_k(q_i^f) = C_k(q_f^i)=C'(q_f^i)\succeq C_f$. 

We can show now that the word $a_1\dots a_k$ belongs to $\bigcap_{1\leq i\leq n} L(\mathcal{A}_i)$. 
Let $1\leq i\leq n$. It is easy to see that for all $1\leq j\leq k$, there exists a unique $q_i^j\in Q_i$ such that $\tilde C_j(q_i^j)>0$ and $\tilde C_j(q)=0$ for all $q\in Q_i\setminus\{q_i^j\}$.
Moreover, for all $1\leq j\leq k$, $q_i^j=\Delta_i^*(q_i^0, a_1\dots a_j)$.  As $\tilde C_k(q_i^f)=C'(q_i^f)>0$, we get that $q_i^f=\Delta_i^*(q_i^0, a_1\dots a_k)$, and hence $a_1\dots a_k\in L(\mathcal{A}_i)$.

Together with Theorem \ref{thm:in-pspace:final}, we then get the following theorem.
\begin{theorem}
	\CCover\ for Wait-Only (broadcast) protocols is \pspace-complete.
\end{theorem}

%% file: Figures/example-2.tex
\tikzset{box/.style={draw, minimum width=4em, text width=4.5em, text centered, minimum height=17em}}

\begin{tikzpicture}[->, >=stealth', shorten >=1pt,node distance=2cm,on grid,auto, initial text = {}] 
	\node[state, initial] (q0) {$\qinit$};
	\node[state] (q1) [right = of q0, yshift = 25] {$q_1$};
	\node[state] (q2) [right = of q1] {$q_2$};
	\node[state] (q3) [right = of q2] {$q_3$};
	\node[state] (q4) [right  = of q0, yshift = -25] {$q_4$};
	\node[state] (q5) [right  = of q4] {$q_{5}$};
	\node[state] (q6) [right = of q5] {$q_6$};
	\node[state] (q7) [right = of q6] {$q_7$};

	\path[->] 
	(q0) edge [thick,bend right = 0] node  []{$!!\tau$} (q4)
	edge [thick,bend left = 0] node  [above, xshift =-2]{$!!a$} (q1)
	(q1) edge [thick,bend left = 0] node  [above]{$?c$} (q2)
	(q2) edge [thick,bend left = 20] node  [above]{$!b$} (q3)
	(q3) edge [thick,bend left = 20] node  [below]{$?b$} (q2)
	(q4) edge [thick,bend left = 0] node  []{$?a$} (q5)
	(q5) edge [thick] node  [above]{$!!c$} (q6)
	(q6) edge [thick] node {$?b$} (q7)
	
	;
\end{tikzpicture}

%% file: Figures/prot-pspace-hardness.tex
\tikzset{box/.style={draw, minimum width=4em, text width=4.5em, text centered, minimum height=17em}}

\begin{tikzpicture}[->, >=stealth', shorten >=1pt,node distance=2cm,on grid,auto, initial text = {}] 
	\node[state, initial] (q0) {$\qinit$};
	\node[state] (qgo) [above = 0 of q0, xshift = 65, yshift = 40] {$q_s$};
	\node[state] (q1) [left = 0 of q0, xshift = 65, yshift =10] {$q_1$};
	\node[state] (qn) [right  = 0 of q0, xshift = 65, yshift = -25] {$q_n$};
	
	\node[state] (q01) [below  = 0 of q1, xshift = 60] {$q_{1}^0$};
	\node[state] (q0n) [below  =0 of qn, xshift =60] {$q_{n}^0$};

	\node[state] (qf1) [right  = 1.8 of q01] {$q_{1}^f$};
	\node[state] (qfn) [right  = 1.8 of q0n] {$q_{n}^f$};

	\node[draw, fill = orange, fill opacity = 0.2, text opacity = 1, fit=(q01) (qf1), text height=0.06 \columnwidth, label ={[shift={(0ex,-5ex)}]:$\PP_1$}] (A1) {};
	
	\node[draw, fill = orange, fill opacity = 0.2, text opacity = 1, fit=(q0n) (qfn), text height=0.06 \columnwidth, label ={[shift={(0ex,-5ex)}]:$\PP_n$}] (An) {};
	
%	\node[box, fit=(q1) (qf1), yshift=-5, fill = yellow, fill opacity = 0.2] (A) {};
	
%	\node[box, fit=(qn) (qfn), yshift=-5, fill = yellow, fill opacity = 0.2] (B) {};
	
	\path (q1) -- node[auto=false]{\ldots} (qn);
	
%	\node[state] (q1b) [left = of q0, xshift = 0, yshift = -30] {$q_1$};
%	\node[state] (qnb) [right  = of q0, xshift = 0, yshift = -30] {$q_n$};
	
%	\node[state] (q01b) [below  = 1 of q1] {$q_{0,1}$};
%	\node[state] (q0nb) [below  = 1 of qn] {$q_{0,n}$};
	
%	\node[state] (qf1b) [below  = 3 of q01] {$q_{f,1}$};
%	\node[state] (qfnb) [below  = 3 of q0n] {$q_{f,n}$};
	
	\node[state] (qfails) [below  = 0 of q0, xshift = 240, yshift = -7] {$q_\textit{fail}$};

	\path[->] 
	(q0) edge [] node  [above]{$!!\tau$} (qgo)
	(q0) edge node  [above]{$!!\tau$} (q1)
	(q0) edge  node  [above]{$!!\tau$} (qn)
	
	(qgo) edge [loop left] node {$!!go$} ()
	(qgo) edge [loop right] node {$!!a, a \in \Sigma$} ()
	(q1) edge  node [yshift = 0,xshift =0] {$?go$} (q01)
	(qn) edge node [yshift = 0,xshift =0] {$?go$} (q0n)

	(A1) edge node [above] {$?go$} (qfails)
	
	(An) edge node [below] {$?go$} (qfails)
	;
\end{tikzpicture}

%% file: waitonlyb.tex
\section{\CCover\  for Wait-Only rendez-vous protocols is P-complete}\label{sec:RDV}

In this section, we focus on Wait-Only \emph{rendez-vous} protocols, which do not involve broadcast transitions.

In the last section, we showed that the configuration coverability problem for Wait-Only broadcast protocols (restricted to broadcast transitions) is \pspace-complete.
The algorithm we presented runs in polynomial space in the size of the protocol and logarithmic space in the size of the configuration. This implies that we cannot expect to characterize the set of coverable configurations in an explicit or compact form: if such a characterization were available, the algorithm would run in time independent of the size of the target configuration.
In particular, while it is possible to determine whether a configuration of the form $\mset{j \cdot q}$ is coverable for a given $j \in \mathbb{N}$, there is no known bound on the maximum value of $j$ for which $q$ is still coverable. As a result, computing the maximal number of processes that can simultaneously cover a given waiting state $q$ would require repeatedly querying the coverability procedure for increasing values of $j$, until the answer becomes negative. In the worst case, this upper bound may be very large.
This reflects a limitation of verification in broadcast networks: some receiving states may be coverable only by a specific number of processes, and identifying that number can be computationally expensive.
Then, in the case of Wait-Only rendez-vous protocols, the situation improves significantly. These networks enjoy better verification properties: we will see that it is possible to compute, in polynomial time, a succinct representation of the set of coverable configurations. In particular, we are able to determine the maximum number of processes that can be simultaneously present in each reception state. This leads to an efficient method for reasoning about coverability in RDV-protocols, independent of the size of the configuration we aim to cover.
Observe that the copypaste property (Lemma \ref{lemma:copycat-action-state}) remains valid when considering Wait-Only rendez-vous protocols.

In the sequel, we will often refer to paths in the underlying graph of a protocol. Formally, given a protocol $\PP = (Q, \Sigma, \qinit, T)$, a \emph{path} is either a state $q \in Q$ or a finite sequence of transitions $(q_0, \alpha_0, q_1)(q_1, \alpha_1, q_2)\ldots(q_k, \alpha_k, q_{k+1})$. In the first case, the path starts and ends at $q$; in the second case, it starts at $q_0$ and ends at $q_{k+1}$.

\subsection{Token-Sets of Configurations}
%\lulutex{----------------------------------- from concur23 ----------------------------}

To solve the coverability problem for Wait-Only \rdvprot s in polynomial time, we rely on a sound and complete abstraction of the set of reachable configurations. We represent this set abstractly as a token-set of configurations, consisting of:
\begin{itemize}
	\item a set $S$ of states that can host an unbounded number of processes simultaneously (including action states, i.e., those in $Q_A$), and
	\item a set $\Toks$ of tokens that represent states which can host only a bounded number of processes at a time (specifically, at most one) along with the last sent message that allowed reaching them.
\end{itemize}
This situation arises when a message $m$ can be received from a waiting state $q$, and $m$ must have been sent earlier along the path leading to $q$. Therefore, each token in $\Toks$ is a pair $(q, m)$, where $q \in Q_W$ is a waiting state and $m \in \Sigma$ is a message that was sent before reaching $q$.
Note that since multiple paths can lead to $q$ with different messages being sent along the way, $\Toks$ may contain distinct tokens $(q, m_1)$ and $(q, m_2)$ with $m_1 \neq m_2$.

%\lugtext{missing exemple}
\begin{figure}
\centering\input{Figures/wo-rdv-prot.tex}
\caption{A Wait-Only \rdvprot~$P$}\label{figure:wo-rdv-prot:ex1}
\end{figure}
\begin{example}\label{ex:wo-rdv-prot}
	Consider the Wait-Only rendez-vous protocol $\PP$ of \Cref{figure:wo-rdv-prot:ex1}. Action states ($\qinit, q_2$ and $q_4$) can be populated by any number of processes, in particular for $q_2$:
	\begin{align*}
		\mset{\qinit, \qinit, \qinit , \dots} & \mtransUp{(\qinit, !a, q_1)} \mset{q_1, \qinit, \qinit, \dots} \mtransUp{(\qinit, !a, q_1)} \mset{q_2, q_1, \qinit, \dots} \\
		& \mtransUp{(\qinit, !a, q_1)} \mset{q_2, q_2, q_1, \dots}  \mtransUp{(\qinit, !a, q_1)} \cdots
	\end{align*}
	However $q_1$ (a waiting sate) is coverable by \emph{at most} one process: the message $a$ is necessarily sent to reach $q_1$, and is received from $q_1$. Hence, a second process trying to reach $q_1$ will exit the first one already there.
	This is not necessarily the case for all waiting states: if only $b$ can be received from $q_1$ (and $a$ is no longer received), the above argument no longer holds and $q_1$ can be covered by any number of processes with successive sendings of $a$ from an initial configuration.
%	\lulutex{adapter avec le protocole qui vient après }
\end{example}

%\begin{figure}[h]
%	\begin{center}
%		\input{ch-wo-cover/Figures/wo-rdv-prot-3}
%	\end{center}
%	\caption{Another Wait-Only rendez-vous protocol $\PP'$.}\label{figure:wo-rdv-prot:ex2}
%\end{figure}

In the sequel, we consider a Wait-Only \rdvprot\ $\PP = (Q, \Sigma, \qinit, T)$.

\begin{definition}
	A \emph{token-set of
		configurations} $\gamma$ is a pair $(S,\Toks)$ such that:
%	\vspace*{-0.2cm}
	\begin{itemize}
		\item $S \subseteq Q$ is a subset of states, and,
		\item $\Toks \subseteq Q_W \times \Sigma$ is a subset of pairs (called tokens)
		composed of a waiting state and a message, and,
		\item $q \not\in S$ for all $(q,m) \in \Toks$.
	\end{itemize}
We denote
by $\Gamma$ the set of token-sets of configurations. 
\end{definition}

%for waiting states that can answer requests on a message $m$, message that is necessarily sent for a process to be in this state. 
%Hence, we remember this message along with the state in the set $\Toks$. 
%The intuition for this abstraction is that, in a wait-only
%protocol, there are some states that can contain an unbounded number
%of processes (states in $S$), this is for instance the case of all the
%reachable active states, and other states whose number of processes they can contain
%at any given point is bounded -- 
%%that will be able to receive
%%at any moment a bounded number of processes, 
%these are the states
%appearing in $\Toks$. Furthermore for these states, we shall see that
%the bound is $1$ (this will be a consequence of the correction of our
%abstraction). 

Let $\gamma=(S,\Toks)$ be a token-set of
configurations. Before we go into the configurations represented by
$\gamma$, we need some preliminary definitions. 
\begin{itemize}
	\item We note $\mst(\mathit{{\kern-1pt}\Toks})$ the set $\set{q \in Q_W
		\mid\textrm{there exists } m\in \Sigma\textrm{ such that }(q,m) \in \Toks}$ of control states
	appearing in $\Toks$. 
	\item Given a state $q \in Q$,
	% and a message $a \in \Sigma$, 
	we recall that
%	\begin{itemize}
%		\item 
		$\Rec{q}$ is the set $\set{ m \in \Sigma \mid\textrm{there exists } q'\in Q \textrm{ such that }
			(q,?m,
			q') \in T}$ of messages that can be received in state $q$ (if $q$ is
		not a waiting state, this set is empty);
%		\item $"\staterec{a}"$ is the set $\set{ p \in Q \mid\textrm{there exists } p'\in Q \textrm{ such that }
%			(p,?a,
%			p') \in T}$ of states from which $a$ can be received (it is always the case that $\staterec{a}\subseteq \waitingset{Q}$);\lulu{pour le moment je ne l'ai pas vu utilisé dans ce chapitre, je le rappelle au moment où je l'utilise}
%	\end{itemize}
	
	%\lulu{this is already defined, but we should recall it instead}
	\item Given two different waiting states $q_1$ and $q_2$ in
	$\starg{\Toks}$, we say $q_1$ and $q_2$ are \emph{conflict-free} in
	$\gamma$ if there exist $m_1,m_2 \in \Sigma$ such that $m_1 \neq m_2$, 
	$(q_1,m_1),(q_2,m_2) \in \Toks$ and $m_1 \notin \Rec{q_2}$ and
	$m_2 \notin \Rec{q_1}$. 
\end{itemize}
Intuitively, two states $q$ and $p$ are conflict-free if a process can reach $q$ by sending a message $a$, and another can reach $p$ by sending a message $b$, such that the message $a$ is not received from $p$, and $b$ is not received from $q$. In other words, there exist two tokens $(q, a)$ and $(p, b)$, and the emission of $a$ (resp. $b$) does not trigger a reception from $p$ (resp. $q$). This ensures that both states are jointly reachable, regardless of the order in which $q$ and $p$ are reached.

We can now make the link between configurations in $\mconfs$ and token-sets of configurations.
\begin{definition}
	Let $\gamma=(S,\Toks)$ be a token-set of
	configurations.
	We say that a configuration $\mconf\in\mconfs$ \emph{respects}
	$\gamma$ if and only if for all $q \in Q$ such that $\mconf(q)>0$ one of the following two
	conditions holds:
	\begin{enumerate}
		\item \label{ccover-wo-consistency-1} $q \in S$, or,
		\item \label{ccover-wo-consistency-2}$q \in \starg{\Toks}$, $\mconf(q)=1$ and for all $q' \in \starg{\Toks} \setminus\set{q}$ such that
		$\mconf(q')=1$, we have that $q$ and $q'$ are "conflict-free".
	\end{enumerate}
	Let
	$\Interp{\gamma}$ be the set of configurations respecting $\gamma$. 
\end{definition}
	
Observe that $\Interp{(\set{\qinit}, \emptyset)} = \mconfs_{init}$, hence $(\set{\qinit}, \emptyset)$ will be our \emph{initial} token-set. 
Note that these conditions only speak about states $q$ such that $\mconf(q) > 0$ as we are only interested in characterising the reachable states (and unreachable states should not appear in $S$ or $\starg{\Toks}$).
\color{black}
Note
that in $\Interp{\gamma}$, for $q$ is in $S$ there is no restriction on
the number of processes that can be put in $q$, whereas if $q$ in
$\starg{\Toks}$, it can host at most one process. Two
states from $\starg{\Toks}$  can both host
a process in the same configuration if they are conflict-free.

\begin{example}
	Going back to Example \ref{ex:wo-rdv-prot} and protocol $\PP$ of \Cref{figure:wo-rdv-prot:ex1}, the token-set characterizing all the coverable configurations is $\gamma =(S, \Toks )$, with:
	\begin{itemize}
		\item $S = \set{\qinit,  q_2,  q_4}$, and
		\item $\Toks = \set{(q_1, a), (q_3,b)}$ (last message sent to reach $q_1$ [resp. $q_3$] is $a$ [resp. $b$]).
	\end{itemize}
	Hence, coverable configurations (configurations respecting $\gamma$) can host at most one process in $q_1$ and one process in $q_3$.
	Observe that $q_1$ and $q_3$ are \emph{not} conflict-free: $a \in \Rec{q_3}$ and $b \in \Rec{q_1}$. Hence, any configuration $\mconf$ respecting  $\gamma$ can not host a process in $q_1$ and a process in $q_3$. 
	Indeed, in $\PP$ it is impossible to cover $q_1$ and $q_3$ at the same time: once a process is in $q_1$, if a new process tries to reach $q_3$ it has to send $b$, which has to be received by the process in $q_1$.
	
Consider now the protocol $\PP'$ depicted in \Cref{figure:wo-rdv-prot:ex2}. The only difference with $\PP$ is that state $q_1$ can no longer receive message $a$. Surprisingly, this has a significant impact on the set of coverable configurations: now, both $q_1$ and $q_3$ can be covered simultaneously, as witnessed by the following execution:
\[
\mset{\qinit, \qinit} \mtrans \mset{q_1, \qinit} \mtrans \mset{q_1, q_3}.
\]
This was not possible in $\PP$, since the transition $(\qinit, !b, q_3)$ would have caused the first process (on $q_1$) to leave due to a reception of $b$.
However, note that $q_1$ and $q_3$ are \emph{not} conflict-free in the token-set configuration $\gamma$, since $a \in \recfrom{q_3}$. But unlike in $\PP$, we now have $b \notin \recfrom{q_1}$ in $\PP'$, which is crucial.
In fact, the configuration $\gamma$ should not be considered a ``good'' token-set configuration for $\PP'$. Observe that, thanks to $b \notin \recfrom{q_1}$, any number of processes can be added on $q_1$ without interfering with the behavior of those on $q_3$. For example:
\[
\begin{aligned}
	\mset{\qinit, \qinit, \qinit, \qinit, \qinit} & \mtrans \mset{q_1, \qinit, \qinit, \qinit, \qinit} \mtrans \mset{q_1, q_3, \qinit, \qinit, \qinit} \\
	& \mtrans \mset{q_1, q_4, q_1, \qinit, \qinit} \mtrans \mset{q_1, q_4, q_1, q_2, \qinit} \\
	& \mtrans \mset{q_1, q_4, q_1, q_4, q_1}.
\end{aligned}
\]
Hence, in $\PP'$, $q_3$ should be in $S$. 
A ``good'' token-set configuration should satisfy the following condition: for \emph{every} pair of tokens $(q, m)$ and $(q', m')$ with $q \neq q'$, either:
\begin{itemize}
	\item $m \notin \recfrom{q'}$ and $m' \notin \recfrom{q}$ (i.e., $q$ and $q'$ are conflict-free, this time we ask that each pair of tokens with $q$ and $q'$ realizes the conflict-freeness), or
	\item $m \in \recfrom{q'}$ and $m' \in \recfrom{q}$ (i.e., $q$ and $q'$ cannot be reached together).
\end{itemize}

We formalize this notion below and call it \emph{consistency}.

\end{example}

\begin{figure}[h]
	\begin{center}
		\input{Figures/wo-rdv-prot-2}
	\end{center}
	\caption{Another Wait-Only rendez-vous protocol $\PP'$.}\label{figure:wo-rdv-prot:ex2}
\end{figure}

%We need a last notion to characterise the manipulated sets of
%configurations. We restrict indeed our reasoning to 

%, which means that the token
%$(q,m) \in \Toks$ should really come from a 'feasible path' in the
%protocol starting by the request of a rendez-vous $!m$ and followed by
%reception of messages that can effectively be emitted. We as well add
%a property to ensure that if $q_1$ and $q_2$ are conflict-free thanks
%to $(q_1,m_1)$ and $(q_2,m_2)$ and
%$q_1$ and $q_3$ are conflict-free thanks to $(q_1m'_1)$ and $(q_3,m_3)$ then $q_1$ and $q_3$ are conflict-
%free thanks to $(q_1,m_1)$ and $(q_3,m_3)$ too. 

%\begin{figure}[t]
%	\begin{center}
%		\input{ch-wo-cover/Figures/rdv-prot-ccover-scover}
%	\end{center}
%	\caption{The rendez-vous protocol $\PP'$ obtained from $\PP$ (\Cref{figure:wo-rdv-prot:ex1}) and the multiset $\mset{q_1, q_2}$.}\label{figure:rdv-prot:from-prot:ccover-scover}
%\end{figure}

%\begin{example}
%	\lulutex{todo}
%\end{example}

%\lulutex{Avant def: exemple de deux etats tel que $m \in \recfrom{q'}$ et $m' \nin \recfrom{q}$ et montrer que on peut mettre autant qu'on veut dedans + Expliquer que sans cond (ii) on pourrait avoir $q$ et $q'$ pas couvrables en meme temps car dans Toks.}
%Finally, we will only consider token-sets of configurations that
%are \emph{consistent}. This property aims to ensure that concrete configurations
%that respect it are indeed reachable from states of $S$.
\begin{definition}
	A token-set of
	configurations $\gamma=(S,\Toks)$ is \emph{""consistent""} if 
	\begin{enumerate}[(i)]
		\item for all $(q,m) \in
		\Toks$, there exists a path
		$(q_0,!m,q_1)(q_1,?m_1,q_2)\ldots(q_k,?m_k,q)$ in $\PP$
		such that $q_0
		\in S$, 
		 and
		there exists $(q'_i,!m_i,q''_i) \in T$ with $q'_i \in S$ for all $1\leq i \leq k$;
		\item for two tokens $(q,m), (q',m') \in \Toks$ either $m\in\Rec{q'}$ and $m'\in\Rec{q}$, or, $m\notin\Rec{q'}$ and $m'\notin\Rec{q}$.
	\end{enumerate}
\end{definition}
%\lulutex{reformuler (i) je pense}
Condition (i) ensures that processes in $S$ can indeed lead to a process in the states from $\starg{\Toks}$. Condition (ii) ensures that if in a configuration $\mconf$,
 some states in $\starg{\Toks}$ are pairwise conflict-free, then they can all host a process together no matter the path chosen to reach them (hence conflict freeness is realized for each pair of tokens with $q$ and $q'$). If they are not pairwise conflict-free, then, there is no way to reach them together.
% \lulu{Arnaud (et Tali): pourquoi? ce serait bien d'avoir une petite intuition pourquoi cela nous sert ? mais je ne la trouve pas}

We assume w.l.o.g. that each state and each message appears in at least one transition, i.e.\ $|Q| \leq |T|$ and $|\Sigma| \leq |T|$.
We now prove that checking that a configuration respects a token-set can be done in polynomial time and that $\Interp{\gamma}$ is \emph{downward-closed}, i.e. \ for any $\mconf \leq \mconf'$ such that $\mconf' \in \Interp{\gamma}$, it holds that $\mconf \in \Interp{\gamma}$.

\begin{lemma}\label{lem:interp-cover-check}
Given  $\gamma\in \Gamma$ and a configuration $\mconf$, there exists $\mconf' \in
\Interp{\gamma}$ such  that $\mconf' \geq \mconf$ if and only if $\mconf \in
\Interp{\gamma}$. Checking that $\mconf\in\Interp{\gamma}$ can be done in polynomial time.
\end{lemma}
\begin{proof}
	Let $\mconf' \in \Interp{\gamma}$ such that $\mconf' \geq \mconf$. Let $q \in Q$
	such that $\mconf(q)>0$. Then we have $\mconf'(q)>0$. If $q \notin S$, then $q
	\in \starg{\Toks}$ and   $\mconf'(q)=1$ and $\mconf(q)=1$ too.  Furthermore for all $q' \in \starg{\Toks} \setminus\set{q}$ such
	$\mconf(q')=1$, we have that $\mconf'(q')=1$ and $q$ and $q'$ are
	conflict-free. This allows us to conclude that $\mconf \in
	\Interp{\gamma}$. Checking whether $\mconf$ belongs to $\Interp{\gamma}$ can be done in
	polynomial time applying the definition of $\Interp{\cdot}$.
\end{proof}

%\begin{example}
%	Going back to \Cref{ex:wo-rdv-prot} and protocol $\PP$ of \Cref{figure:wo-rdv-prot:ex1}, the token-set configuration characterizing all the coverable configurations is $(S, \Toks )$, with:
%	\begin{itemize}
%		\item $S = \set{\qinit,  q_2,  q_4}$, and
%		\item $\Toks = \set{(q_1, a), (q_2,b)}$.
%	\end{itemize}
%\end{example}

\subsection{Computing Token-Sets of Configurations}

Our polynomial time algorithm is based on the computation of a
sequence of consistent token-sets of
configurations. This sequence, of a polynomial length, leads to a final
token-set  representing a correct abstraction for the set of
coverable configurations. This will be achieved by a function $F:\Gamma \to \Gamma$,
that inductively computes  this final token-set
starting from $\gamma_0=(\set{\qinit}, \emptyset)$. %For this matter, we rely on a function
%$F:\Gamma \mapsto \Gamma$ which allows to increase in a certain sense our token-set of
%configurations. Our sequence will then start with the token-set of
%configurations $(\set{q_{in}},\emptyset)$ and will be built by
%applying the function $F$ successively until saturation.
Formal definition of the function $F$ relies on intermediate sets
$S''\subseteq Q$ and $\Toks''\subseteq Q_W \times\Sigma$, which are the smallest sets satisfying the conditions described in \autoref{table:F}. 

\begin{table}[h]
\begin{center}
\label{tab:S''}
%\makebox[\textwidth]{%
%\scalebox{1}{
\begin{tabular}{ p{15cm}}
\toprule
\textbf{Construction of intermediate states $S''$ and $\Toks''$}\\
\midrule
\vspace{-0.4cm}
\begin{enumerate}[]
	\item $S\subseteq S''$ and $\Toks\subseteq \Toks''$
%	\item \label{ccover-wo-F-cond-internal}for all $(p,\tau,p') \in T$ with $p \in S$, we have $p' \in S''$
	\item for all $(p,!a,p') \in T$ with $p \in S$, we have: 
	\begin{enumerate}[]
		\item $p' \in S''$ if $a \notin \Rec{p'}$ or if there exists
		$(q,?a,q') \in T$ with $q \in S$;\label{ccover-wo-F-cond-send-S}
		\item $(p',a) \in \Toks''$ otherwise (i.e. when $a \in \Rec{p'}$ and for all $(q,?a,q') \in T$, $q \notin S$);\label{ccover-wo-F-cond-newtok}
	\end{enumerate}
	\item for all $(q,?a,q') \in T$ with $q \in S$ or $(q,a) \in \Toks$, we have $q' \in
	S''$ if there exists $(p,!a,p') \in T$ with $p \in S$;\label{ccover-wo-F-cond-reception-S}
	\item for all $(q,?a,q') \in T$ with $(q,m) \in \Toks$ with $m
	\neq a$, if there exists $(p, !a, p') \in T$ with $p \in S$, we have:\label{ccover-wo-F-cond-tok}
	\begin{enumerate}[itemsep=0cm,itemindent=-0.2cm]
		\item $q' \in S''$ if $m \notin \Rec{q'}$;\label{ccover-wo-F-cond-tok-end}
		\item $(q',m) \in \Toks''$ if $m \in \Rec{q'}$.\label{ccover-wo-F-cond-tok-step}
	\end{enumerate}
\end{enumerate}
\\
 \toprule
%\end{enumerate}
\end{tabular}
%}}
%}
%}
\caption{{Definition of $S'', \Toks''$ for $\gamma=(S,\Toks)$.}}\label{table:F}
\end{center}
\end{table}

\color{black}
From $S$ and $\Toks$, rules described in \autoref{table:F} add states and tokens to $S''$ and $\Toks''$ from the outgoing transitions from states in $S$ and $\starg{\Toks}$.
It must be that every state added to $S''$ can host an unbounded number of processes, and every state added to $\Toks''$ can host at most one process. Furthermore, two conflict-free states in $\Toks''$ should be able to host at most one process at the same time.  
We now give an example of computation of $S''$ and $ \Toks''$.

\begin{figure}[t]
%  \begin{minipage}[c]{.49\columnwidth}
%	\resizebox*{!}{2.3cm}{
	\centering
	\input{Figures/example-wo}
%  }
  \caption{Wait-only protocol $\PP_1$.}\label{fig-example-wo}
%\end{minipage}
%\begin{minipage}[c]{.49\columnwidth}
%  \resizebox*{!}{2.3cm}{
%\end{minipage}
%\vspace*{-0.3cm}
\end{figure}

\begin{example}
%We provide now some intuition on how we defined $F$ using different
%examples. 
Consider the Wait-Only protocol $\PP_1$ depicted on Figure
\ref{fig-example-wo}. From $(\set{q_{in}},\emptyset)$, rules described in \autoref{table:F}~construct the following pair:
\[
(S_1'', \Toks_1'') = (\set{q_{in},q_4},\set{(q_1,a),\linebreak[0](q_1,b),(q_5,c)}).
\] 
%We have
%$F((\set{q_{in}},\emptyset))=(\set{q_{in},q_4},\set{(q_1,a),\linebreak[0](q_1,b),(q_5,c)})$. 
In
$\PP_1$, it is indeed possible to reach a configuration with as
many processes as one wishes in the state $q_4$
by repeating the transition $(q_{in},!d,q_4)$ (rule \ref{ccover-wo-F-cond-send-S}). On the other hand, it
is possible to put \emph{at most} one process in the waiting state $q_1$
(rule \ref{ccover-wo-F-cond-newtok}), because any other attempt from a process in $\qinit$ will yield a reception
of the message $a$ (resp. $b$) by the process already in $q_1$. % For
%instance, to add one token to $q_1$, one needs to use transitions
%$(q_{in},!a,q_1)$ or $(q_{in},!b,q_1)$, by repeating one of the transitions, a process will be added to  $q_1$ but
%another one will be removed from $q_1$ because of the transitions
%$(q_1,?a,q_2)$ and $(q_1,?b,q_2)$. 
Similarly, we can put at most
one process in $q_5$. Note that in
$\Toks_1''$, the states $q_1$
and $q_5$ are conflict-free and it is hence possible to have
simultaneously one process in both of them.

If we apply rules of \autoref{table:F} one more time to $(S''_1, \Toks''_1)$, we get 
\begin{align*}
	&S_2''=\set{\qinit, \textcolor{blue}{q_2}, {q_4}, \textcolor{blue}{q_6},\textcolor{blue}{q_7}} \mbox{ and }\\
	&\Toks_2''=\set{{(q_1,a)}, {(q_1,b)} ,\textcolor{blue}{(q_3,a)},\textcolor{blue}{(q_3,b)},{(q_5,c)}}.
\end{align*}

We can put at most one process in $q_3$: to add one, a process will take the transition
$(q_1,?c,q_3)$. Since $(q_1,a)$, $(q_1,b)\in\Toks''_1$, there can be at most one process in
state $q_1$, and this process arrived by a path in which
the last message sent was $!a$ or $!b$. % (this is witnessed by
%the tokens  $(q_1,a)$ and $(q_1,b)$), since these two rendez-vous can
%be accepted from state $q_3$ we cannot put a great number of processes
%in it. 
Since $\{a,b\}\subseteq\Rec{q_3}$, by rule~\ref{ccover-wo-F-cond-tok-step}, $(q_3,a),(q_3,b)$ are added. On the other hand one can put as many processes as one wants in the
state $q_7$ (rule \ref{ccover-wo-F-cond-tok-end}): from a configuration with one process on state $q_5$, successive sendings of letter $c$, and 
rendez-vous
on letter $d$ will allow to increase the number of processes in state $q_7$.

\color{black}

However, one can observe that $q_5$ can in fact host an unbounded number of processes:
 once two processes have been put on states $q_1$ and $q_5$ respectively (remember that $q_1$ and $q_5$ are conflict-free in $(S''_1, \Toks''_1)$), iterating rendez-vous on letter $c$ (with transition $(q_1, ?c, q_3)$) and rendez-vous on letter $a$ can increase unboundedly the number of processes on $q_5$.

As a consequence we need to apply another transformation to $(S_2'',\Toks_2'')$ to obtain $F(S_1'',\Toks_1'')$. We shall see that this second step has no impact when computing $F((\set{\qinit}, \emptyset))$ hence we have that $F((\set{\qinit}, \emptyset)) = (S''_1, \Toks''_1)$.

%This is why we need another transformation from $S_2'', \Toks_2''$ to $F(S''_1, \Toks''_1)$.
%As we shall see, this transformation does not have any impact on $S''_1$ and $\Toks''_1$ and so it holds that $F((\set{\qinit}, \emptyset)) = (S''_1, \Toks''_1)$.
\end{example}
\color{black}
%Note $F(\gamma) = (S', \Toks')$, \autoref{table2:F} describes the construction of $S'$ from $(S'', \Toks'')$, while $\Toks' = \Toks'' \setminus (S \times \Sigma)$, i.e.\ all states added to $S'$ are removed from $\Toks'$ so a state belongs either to $S'$ or to $\starg{\Toks'}$. 

We shall finally set $F(\gamma)$ equal to $(S', \Toks')$, where the construction of $S'$ from $(S'', \Toks'')$ is given by \autoref{table2:F} and $\Toks' = \Toks'' \setminus (S' \times \Sigma)$, i.e.\ all states added to $S'$ are removed from $\Toks'$ so a state belongs either to $S'$ or to $\starg{\Toks'}$.

\begin{table}[h]
	\begin{center}
		\label{tab2:S''}
		\makebox[\textwidth]{%
			\scalebox{1}{
				\begin{tabular}{ p{13.5cm}}
					\toprule
					\textbf{Construction of state $S'$, the smallest set including $S''$ and such
						that:
					}\\
					\midrule
					\vspace{-0.4cm}
					\begin{enumerate}[]\addtocounter{enumi}{5}
						\item for all $(q_1, m_1), (q_2, m_2) \in \Toks''$ such that
						\begin{itemize}
							\item $m_1
							\ne m_2$ and $m_2 \notin \Rec{q_1}$ and $m_1 \in \Rec{q_2}$, 
						\end{itemize}    
					we	have $q_1 \in S'$;\label{ccover-wo-F-cond-2toks-1}
						\item for all $(q_1, m_1), (q_2, m_2), (q_3,m_2) \in \Toks''$ such that
						\begin{itemize}
							\item $m_1 \ne m_2$, and $(q_2, ?m_1, q_3) \in T$,
						\end{itemize}
						we have $q_1 \in S'$;\label{ccover-wo-F-cond-3toks-1}
						\item for all $(q_1, m_1), (q_2, m_2), (q_3, m_3) \in \Toks''$ such
						that
						\begin{itemize}
							\item $m_1 \ne m_2$, $m_1\ne m_3$ and $m_2 \ne m_3$ and 
							\item $m_1 \notin \Rec{q_2}$ and $m_1 \in \Rec{q_3}$ and 
							\item $m_2\notin \Rec{q_1}$ and $m_2 \in \Rec{q_3}$, and 
							\item $m_3 \in \Rec{q_2}$ and $m_3 \in \Rec{q_1}$,
						\end{itemize} 
						we have $q_1 \in S'$.\label{ccover-wo-F-cond-3toks-2}
					\end{enumerate}
					%  6. for all $(q_1, m_1), (q_2, m_2) \in \Toks''$ such that $m_1
					%    \ne m_2$ and $m_2 \nin \Rec{q_1}$ and  $m_1 \in \Rec{q_2}$, we
					%    have $q_1 \in S'$;\\
					%    7. for all $(q_1, m_1), (q_2, m_2), (q_3,m_2) \in \Toks''$ s.t $m_1 \ne m_2$
					%   and $(q_2, ?m_1, q_3) \in T$, we have $q_1 \in S'$;\\
					% 8. for all $(q_1, m_1), (q_2, m_2), (q_3, m_3) \in \Toks''$ such
					%  that $m_1 \ne m_2$ and  $m_1\ne m_3$ and $m_2 \ne m_3$ and 
					%  $m_1 \nin \Rec{q_2}$,\\
					%   $m_1 \in \Rec{q_3}$ and  $m_2\nin \Rec{q_1}$, $m_2 \in \Rec{q_3}$, and $m_3 \in \Rec{q_2}$ and $m_3 \in \Rec{q_1}$,
					%   we have $q_1 \in S'$.\\
					%    \midrule
					%     \textbf{Construction of state $\Toks'$}\\
					%     \midrule
					%    $\Toks'=\set{(q,m) \in \Toks'' \mid q \not\in
						%  S'}$.
					\\
					%\midrule
					%\textbf{Construction of state $\Toks'$}\\
					%\midrule
					%$\Toks'=\set{(q,m) \in \Toks'' \mid q \not\in
						%	S'}$.
					%\\
					\toprule
					%\end{enumerate}
				\end{tabular}
				%}}
	}
}
\caption{{Definition of $S'$ where $F(\gamma)=(S',\Toks')$ for $(S'', \Toks'')$.}}\label{table2:F}
\end{center}
\end{table}

%Now, observe that the tokens $(q_5,c)$, $(q_1,a)$, $(q_3,a)$ allow for application of rule~\ref{ccover-wo-F-cond-3toks-1}, since $(q_1,?c,q_3)\in T$, and yields $q_5$ in $S'$. Once two processes have been put on states $q_1$ and $q_5$ respectively (remember that $q_1$ and $q_5$ are conflict-free in $F(\gamma)$), iterating rendez-vous on letter $c$ (with transition $(q_1, ?c, q_3)$) and rendez-vous on letter $a$ put as many processes as one wants on state $q_5$.
%Finally, $F({F(\set{q_{in}},\emptyset)})=(\set{{q_{in}}, q_2,{q_4}, q_5, q_6,q_7},\linebreak[0]\set{{(q_1,a)}, {(q_1,b)} ,(q_3,a),(q_3,b)})$. Since $q_1$ and $q_3$ are not
%conflict-free, they won't be reachable together in a configuration. % from a configuration with a
%great number of processes on state $\qinit$ and one process on state
%$q_5$, then, one can put a great number of processes on state $q_7$ by
%doing successively rendez-vous with letter $d$ and a non-blocking
%request on letter $c$.
%\end{example}

%Observe that it might be that a state is both added to $S''$ and $\Toks''$; in that case, it will be removed from $\Toks'$ by application of the last rule of $F$. Hence, 
%a state belongs either to $S'$ or to $\starg{\Toks'}$. 

\begin{figure}
		\centering
	\input{Figures/example-wo-2}
	%	}
\caption{Wait-only protocol $\PP_2$.}\label{fig-example-wo-2}
\end{figure}

\color{black}
\begin{example}
	\color{black}
Now the case of state $q_5$ evoked in the previous example leads to the application of  rule~\ref{ccover-wo-F-cond-3toks-1}, since $(q_5,c)$, $(q_1,a), (q_3,a) \in \Toks''_2$, and $(q_1,?c,q_3)\in T$.
Finally, we get that 
\[
F(S''_1, \Toks''_1) = F({F(\set{q_{in}},\emptyset)})=(\set{{q_{in}}, q_2,{q_4}, q_5, q_6,q_7},\linebreak[0]\set{{(q_1,a)}, {(q_1,b)} ,(q_3,a),(q_3,b)}).
\]
Since $q_1$ and $q_3$ are not
conflict-free, they won't be reachable together in a configuration. 
\color{black}

We consider now the wait-only protocol $\PP_2$ depicted on Figure
\ref{fig-example-wo-2}. In that case, to compute
$F((\set{q_{in}},\emptyset))$ we will first have 
\begin{align*}
	&S''=\set{q_{in}} \mbox{ and }\\
	 &\Toks''=\set{(q_1,a),(q_2,b),(p_1,m_1),(p_2,m_2),\linebreak[0](p_3,m_3)}
\end{align*}
(using rule \ref{ccover-wo-F-cond-newtok}), to finally get 
\[
F((\set{q_{in}},\emptyset))=(\set{q_{in},q_1,p_1},\set{(q_2,b),(p_2,m_2),\linebreak[0](p_3,m_3)})).
\]
Applying rule \ref{ccover-wo-F-cond-2toks-1}~to tokens $(q_1,
a)$ and $(q_2, b)$ from $\Toks''$, we obtain that $q_1\in S'$: whenever one manages
to obtain one process in state $q_2$, this process can receive message $a$ instead of processes in state $q_1$, allowing one
to obtain as many processes as desired in state $q_1$.  
% appear on the tokens set, then applying $F$ should
%add $q_1$ to the set of unbounded states. 
%Indeed, take the reachable
%configuration with one process on state $q_1$, many processes on state
%$\qinit$ and no process on state $q_2$. With successive non-blocking
%requests on letter $b$, and rendez-vous on letter $a$ with transitions
%$(\qinit, !a, q_1)$ and $(q_2, ?a, q_3)$, we can reach a configuration
%with many processes on state $q_1$.
%
Now, since $(p_1,m_1)$, $(p_2, m_2)$ and $(p_3, m_3)$ are in $\Toks''$
and respect the conditions of rule \ref{ccover-wo-F-cond-3toks-2}, $p_1$ is added to the set $S'$ of unbounded states.
% (we have indeed
%$m_1 \nin \Rec{p_2}$, $m_1 \in \Rec{p_3}$ and  $m_2\nin \Rec{p_1}$, $m_2 \in \Rec{p_3}$, and $m_3 \in \Rec{p_2}$ and $m_3 \in \Rec{p_1}$).
This case is a generalization of the previous one, with 3 processes. Once one process has
been put on state $p_2$ from $\qinit$, iterating the following actions: rendez-vous over $m_3$, rendez-vous over $m_1$, sending of $m_2$,
will ensure as many processes as one wants on state $p_1$.
%
%To justify this, take the reachable configuration with one process on
%state $p_1$, one on state $p_2$ and many processes on state
%$\qinit$. By doing successively: rendez-vous on letter $m_3$ with
%transitions $(\qinit, !m_3, p_3)$ and $(p_2, ?m_3, p_4)$, rendez-vous
%on letter $m_1$ with transitions $(\qinit, !m_1, p_1)$ and $(p_3,
%?m_1, p_4)$ and non-blocking request on message $m_2$ (note that $m_2
%\nin \Rec{p_1}$), we can reach a configuration with many processes on
%state $p_1$. 
Finally applying successively $F$, we get in this case
the token-set $(\set{q_{in},q_1,q_3,p_1,p_2,p_3,p_4},\set{(q_2,b)})$.
\end{example}

 We show that $F$ satisfies several properties: the following lemma ensures that when repeatedly applying $F$ to $(\set{q_{in}},\emptyset)$, the computation eventually reaches a consistent token-set $\gamma_f$ such that $\gamma_f = F(\gamma_f)$. Moreover, each step runs in polynomial time.

\begin{lemma}\label{lem:F-consistent}
	For all consistent $\gamma \in \Gamma$, $F(\gamma)$ is also consistent and can be computed in polynomial time in the size of $\PP$.
\end{lemma}

\begin{proof}
	The fact that $F(\gamma)$ can be computed in polynomial time is a direct consequence of the definition of $F$ (see \autoref{table:F} and \ref{table2:F}).
	
	% TBD (direct from the definition of the function $F$, rules 3.b and 5.b).
	Assume $\gamma = (S,\Toks) \in \Gamma$ to be "consistent". Note $(S'', \Toks'')$ the intermediate sets computed during the computation of $F(\gamma)$, and note $F(\gamma) = (S', \Toks')$.
	
	To prove that $F(\gamma)$ is "consistent", we need to argue that 
	\begin{enumerate}[(1)]
		\item for all $(q,m) \in
		\Toks'' \setminus \Toks$, there exists a path
		$(q_0,\alpha_0,q_1)(q_1,\alpha_1,q_2)\ldots(q_k,\alpha_k,q)$ in $\PP$
		such that $q_0
		\in S$, $\alpha_0=\ !m$, and $\alpha_i=\ ?m_i$ 
		and
		there exists $(q'_i,!m_i,q''_i) \in T$ with $q'_i \in S$ for all $1\leq i \leq k$;
%		\item for all $(q, m) \in \Toks'' \setminus \Toks$, there exists a finite sequence of transitions $(q_0, \alpha_0, q_1) \dots (q_k, \alpha_k, q)$ such that $q_0 \in S$, and $\alpha_0 = !m$ and for all $1 \leq i\leq k$, we have that $\alpha_i = ?m_i$ and that there exists $(q'_i, !m_i, q'_{i+1}) \in T$ with $q'_i \in S$, and 
		\item for all $(q,m), (q',m') \in \Toks'$ either $m\in\Rec{q'}$ and $m'\in\Rec{q}$ or $m\notin\Rec{q'}$ and $m'\notin\Rec{q}$. 
	\end{enumerate}

	\emph{We start by proving property (1).}
	If $(q, m)$ has been added to $\Toks''$ with rule \ref{ccover-wo-F-cond-newtok}, then by construction, there exists $p \in S$ such that $(p, !a, p') \in T$, and $(q, m) = (p', a)$. The sequence of transitions, is the single transition $(p, !m, q)$. 
	%As $S\subseteq S'$, it concludes this case.
	
	If $(q, m)$ has been added to $\Toks''$ with rule \ref{ccover-wo-F-cond-tok-step}, then there exists $(q',m) \in \Toks$, and $(q', ?a, q)$ with $m \ne a$. Furthermore, $m \in \Rec{q}$ and there exists $(p, !a,p') \in T$ with $p \in S$. By hypothesis, $\gamma$ is consistent, hence there exists a finite sequence of transitions $(q_0, \alpha_0, q_1) \dots (q_k, \alpha_k, q')$ such that $q_0 \in S$, and $\alpha_0 = !m$ and for all $1 \leq i\leq k$, we have that $\alpha_i = ?m_i$ and that there exists $(q'_i, !m_i, q'_{i+1}) \in T$ with $q'_i \in S$. By completing this sequence with transition $(q', ?a, q)$ we get an appropriate finite sequence of transitions. \\
	%As $S\subseteq S'$, it concludes the proof.
	
	\emph{It remains to prove property (2).}
	Assume there exist $(q, m), (q',m') \in \Toks'$ such that $m \in \Rec{q'}$ and $m' \notin \Rec{q}$, then as $\Toks' \subseteq \Toks''$, $(q, m), (q',m') \in \Toks''$. By condition \ref{ccover-wo-F-cond-2toks-1}, $q \in S'$, therefore, as $\Toks' = \{(p, a) \in \Toks'' \mid p \notin S'\}$, we have that $(q, m) \notin \Toks'$, and we reached a contradiction.
\end{proof}

The next lemma provides the key ingredient to guarantee that the algorithm terminates in a polynomial number of steps.
It states that by successively applying $F$, we always make ``progress'': either new states are added to $S$, or the set of tokens grows (or remains equal). When both $S$ and $\Toks$ remain unchanged after applying $F$, we have reached a fixpoint.

\begin{lemma}\label{lem:F-increase}
	If $(S',\Toks')=F(S,\Toks)$ then  $S \subseteq S'$ or
	$\Toks \subseteq \Toks'$.
\end{lemma}

\begin{proof}
	From the construction of $F$ (see \autoref{table:F} and \ref{table2:F}), we have $S \subseteq S'' \subseteq S'$.
	
	Assume now that $S=S'$. First note that $\Toks \subseteq \Toks''$ (see Table \ref{table:F}) and that $\starg{\Toks} \cap S=\emptyset$. But $\Toks'=\set{(q,m) \in \Toks'' \mid q \not\in S'}=\set{(q,m) \in \Toks'' \mid q \not\in S}$. Hence the elements that are removed from $\Toks''$ to obtain $\Toks'$ are not elements of $\Toks$. Consequently  $\Toks \subseteq \Toks'$.
\end{proof}

\subsection{Completeness}

The next lemma proves the completeness of the computed abstraction.  
We use the following notation: for a transition $t$ between two configurations $\mconf$ and $\mconf'$, we write $\mconf \mtransse{t} \mconf'$ when the sent message is not received because no other process can receive the message (sending step), and $\mconf \mtransrdv{t} \mconf'$ when the message is received by one other process (rendez-vous step).

We consider a configuration \(\mconf\) that respects a consistent token-set \(\gamma\) and show that a single step from \(\mconf\) yields a configuration \(\mconf'\) that respects \(F(\gamma)\).  

Specifically, we show that:  
(1) each state \(q\) populated in \(\mconf'\) appears in \(F(\gamma)\);  
(2) states that appear in the token-part of \(F(\gamma)\) are populated by at most one process. Moreover, if two such states are populated, they must be conflict-free.

%\lulutex{rappel $\mtransse{}$ et $\mtransrdv{}$.}

\begin{lemma}\label{lem:abstract-completeness}
	For all consistent $\gamma \in \Gamma$, if $\mconf \in \Interp{\gamma}$ and $\mconf \mtrans \mconf'$ then $\mconf' \in \Interp{F(\gamma)}$.
\end{lemma}

\begin{proof}
%	\lulutex{relire}
	\input{proof-abstract-completeness}
\end{proof}

\subsection{Soundness}
The goal of this subsection is to prove the soundness of our construction, formalized by the following lemma. 

\begin{lemma}\label{lem:abstract-soundness}
	For all consistent $\gamma \in \Gamma$, if $\mconf' \in
	\Interp{F(\gamma)}$, then there exists $\mconf'' \in \mconfs$ and $\mconf \in
	\Interp{\gamma}$ such that $\mconf \mtrans^\ast \mconf''$ and $\mconf'' \geq \mconf'$.
\end{lemma}

\ifappendix
In order to prove it, we will use the following steps. Let $\gamma$ be a consistent token-set of configurations and $\mconf'\in \Interp{F(\gamma)}$. We write $\gamma=(S,\Toks)$ and $F(\gamma)=\gamma'=(S',\Toks')$. 
We first show that, starting from a configuration in $\Interp{\gamma}$, we can cover all states in $S'$ simultaneously with any number of processes (Claim~\ref{claim:aux:soundness:wo:rdv}). The proof of this claim, given in the appendix, relies on Lemma~\ref{lem:consistent-reach}, which states that in the Rdv networks, as soon as two states can be covered \emph{separately} by an unbounded number of processes, they can be covered \emph{together} by an unbounded number of processes.
%%Using the previous lemma, we show that if, for every $n \in \nat$ and every state in a given subset, there exists an execution putting at least $n$ processes in that state, then we can construct an execution putting at least $n$ processes in all states of the subset \emph{simultaneously, within the same configuration}.
%Note that this result is \emph{stronger} than the copypaste property (Lemma \ref{lemma:copycat-action-state}), as the subset of states may include both action and waiting states.
%Thus, we can ensure that at least $n$ processes reach each state in the subset simultaneously, including waiting states.

The proof of Lemma~\ref{lem:consistent-reach} itself relies on an easy lemma (Lemma~\ref{lem:monotonicity}) that states that  if there is an execution from $\mconf$ to $\mconf'$, then the same sequence of steps can be realized from any larger configuration $D$. This property follows from the fact that the network’s transition system is well-structured. 
Furthermore, in a rendez-vous network, each state loses at most $2\ell$ processes, where $\ell$ is the length of the run between $\mconf$ and $\mconf'$.
Indeed, at each step, at most two processes change states: one if it is a sending step, and two if it is a rendez-vous. Thus, a state loses at most two processes per step.

Finally, it remains to show that it is also possible to cover states in $\mst(\Toks')$ that are populated in $\mconf'$. This will be done in the proof of Lemma~\ref{lem:abstract-soundness}. 

\else
In order to prove it, we first present two technical lemmas. The first lemma states that if there is an execution from $\mconf$ to $\mconf'$, then the same sequence of steps can be realized from any larger configuration $D$. This property follows from the fact that the network’s transition system is well-structured. 
Furthermore, in a rendez-vous network, each state loses at most $2\ell$ processes, where $\ell$ is the length of the run between $\mconf$ and $\mconf'$.
Indeed, at each step, at most two processes change states: one if it is a sending step, and two if it is a rendez-vous. Thus, a state loses at most two processes per step, which leads to the following lemma.\fi

\begin{lemma}\label{lem:monotonicity}
	Let $\mconf,\mconf' \in \mconfs$ such that $\mconf=\mconf_0 \mtrans \mconf_1 \cdots \mtrans
	\mconf_\ell=\mconf'$. Then, the following two properties hold.
	\vspace{-0.5em}
	\begin{enumerate}
		\item For all $q \in Q$ such that
		$\mconf(q)=2\cdot \ell+x$ for some $x \in \nat$, we have $\mconf'(q)\geq x$.\label{it:lem-1}
		\item For all $D_0 \in \mconfs$ such that $D_0 \geq \mconf_0$, there exist $D_1,\ldots,D_\ell$ such that $D_0 \mtrans D_1 \cdots \mtrans	D_\ell$ and $D_i \geq \mconf_i$ for all $1 \leq i \leq \ell$.\label{it:lem-2}
	\end{enumerate}
\end{lemma}

\begin{proof}
	By the semantics of rendez-vous protocols, each step in the execution from $\mconf$ to $\mconf'$ removes at most two processes per state $q$, which proves the first item.
	The second item is a direct consequence from the fact that $(\mconfs, \mtrans, \leq)$ is a Well Structured Transition System (WSTS). 	Recall that a \emph{well-structured transition system} is a triple
	$(\SS, \trans, \preceq)$ where $\preceq$ is a well-quasi-order on the set $\SS$
	that is \emph{compatible} with the transition relation: whenever $s_1 \preceq
	s_2$ and $s_1 \trans s_1'$, there exists $s_2' \succeq s_1'$ such that $s_2
	\trans^\ast s_2'$. Here $\leq$ is the pointwise order on
	configurations, which is a well-quasi-order on $\nat^Q$ by Dickson's lemma, and
	compatibility is exactly the first item of the present lemma: the extra
	processes of the larger configuration can be left where they are while the same
	sequence of transitions is replayed.
%	, which we proved in \Cref{lemma:network-wsts} (\Cref{chap:model}, page \pageref{lemma:network-wsts}).
\end{proof}
We have seen that in broadcast networks in general, states can be in conflict when one aims at covering them.
\ifappendix As explained hereabove, in Rdv networks, this is not true anymore. The next lemma formalizes the fact that as soon as two states can be 
populated by an unbounded number of processes in two separate runs, they can be covered by an unbounded number of processes in the same run. \else We show here that in the Rdv networks, as soon as two states can be covered separately by an unbounded number of processes, they can be covered together by an unbounded number of processes.\fi 
%Using the previous lemma, we show that if, for every $n \in \nat$ and every state in a given subset, there exists an execution putting at least $n$ processes in that state, then we can construct an execution putting at least $n$ processes in all states of the subset \emph{simultaneously, within the same configuration}.
Note that this result is \emph{stronger} than the copypaste property (Lemma \ref{lemma:copycat-action-state}), as the subset of states may include both action and waiting states.
Thus, we can ensure that at least $n$ processes reach each state in the subset simultaneously, including waiting states.

\begin{lemma}\label{lem:consistent-reach}
	Let $\gamma$ be a consistent token-set of configurations. Given a
	subset of states $U \subseteq Q$, if for all $n \in \nat$ and for all
	$q \in U$ there exists $\mconf_q \in \Interp{\gamma}$ and $\mconf'_q \in \mconfs$ such
	that $\mconf_q \arrowP{}^\ast \mconf'_q$ and $\mconf'_q(q)\geq n$, then for all $n
	\in \nat$, there exists $\mconf \in \Interp{\gamma}$ and $\mconf' \in \mconfs $ such that $\mconf \arrowP{}^\ast
	\mconf'$  and $\mconf'(q) \geq n$ for all $q \in U$.
\end{lemma}

\begin{proof}
	We suppose $\gamma=(S,\Toks)$ a consistent token-set of configurations, and a set of states $U \subseteq Q$, and reason by induction on the number of
	elements in $U\setminus S$. \\

	\textbf{Base case.} The base case is obvious.  Indeed assume $U
	\setminus S=\emptyset$ and let $n\in \nat$. We define the configuration $\mconf$
	such that $\mconf(q)=n$ for all $q \in S$ and $\mconf(q)=0$ for all $q \in
	Q\setminus S$. It is clear that $\mconf \in \Interp{\gamma}$ and that $\mconf(q)
	\geq n$ for all $q \in U$ (since $U
	\setminus S=\emptyset$, we have in fact $U \subseteq S$).\\
	
	\textbf{Induction case.}
	We now assume that the property holds for a set $U$ and we shall see
	it holds for $U \cup \set{p}$, $p\notin S$. We assume hence that for all $n \in \nat$ and for all
	$q \in U \cup \set{p}$ there exists $\mconf_q \in \Interp{\gamma}$ and $\mconf'_q \in \mconfs$ such
	that $\mconf_q \mtrans^\ast \mconf'_q$ and $\mconf'_q(q)\geq n$. Let $n \in \nat$. By induction
	hypothesis, 
	\[
	\mbox{there exists } \mconf_U \in  \Interp{\gamma} \mbox{, } \mconf'_U \in \mconfs \mbox{ s.t. } \mconf_U \mtrans^\ast \mconf'_U
	\mbox{ and } \mconf_U'(q) \geq n \mbox{ for all }q
	\in U.
	\]
%	$ and $
%	$ such that $$  and $$ for all $
	
	We denote by $\ell_U$ the minimal number of steps in a run from
	$\mconf_U$ to $\mconf'_U$. We will see that we can build a configuration $\mconf
	\in \Interp{\gamma}$ such that:
	\[
	\mconf \mtrans^\ast \mconf''_U \mbox{ with } \mconf''_U
	\geq \mconf_U \mbox{ and } \mconf''_U(p) \geq n+2\cdot\ell_U.
	\]
	Using Lemma \ref{lem:monotonicity}, we will
	then have that:
	\[
	\mconf''_U \mtrans^\ast \mconf'\mbox{ with } \mconf' \geq \mconf'_U \mbox{ and } \mconf'(p)
	\geq n.
	\]
	 This will allow us to conclude.

	Having $\mconf_U \in \Interp{\gamma}$, we name $(q_1, m_1) \dots (q_k,
	m_k)$ the
	tokens in $\Toks$ such that $\mconf_U(q_j) = 1$ for all $1 \leq j \leq
	k$, (and for all $q \in \starg{\Toks} \setminus \{q_j\}_{1 \leq j \leq
		k}$, $\mconf_U(q) =0$). Since $\gamma$ is consistent, for each $(q_j,
	m_j)$ there exists a path 
	$(q_{0,j},!m_j,q_{1,j})(q_{1,j},?m_{1,j},q_{2,j})\ldots(q_{\ell_j,j},?m_{\ell_j,j},q_j)$
	in $\PP$ such that $q_{0,j}
	\in S$ and  such that there
	exists $(q'_{i,j},!m_{i,j},q''_{i,j}) \in T$ with $q'_{i,j} \in S$ for
	all $1 \leq i \leq \ell_j$. We denote by $\ell = \max_{1 \leq j\leq
		k}(\ell_j)+1$.
	We know also that there exist $\mconf_p \in \Interp{\gamma}$ and $\mconf'_p \in \mconfs$ such
	that $\mconf_p \mtrans^\ast \mconf'_p$ and $\mconf'_p(p)\geq
	n+2\cdot\ell_U+(k\cdot\ell)$. We denote by $\ell_p$ the minimum number of steps in a run from
	$\mconf_p$ to $\mconf'_p$.  We build the configuration $\mconf$ as follows:
	\begin{align*}
		&\mconf(q)=\mconf_U(q)+2\cdot\ell_p+(k\cdot\ell)+\mconf_p(q) && \mbox{ for all } q \in S\\
		&\mconf(q)=\mconf_p(q) &&\mbox{ for all } q \in \starg{\Toks}
	\end{align*}
%
%	and  we have $\mconf(q)=\mconf_p(q)$ for all $q \in \starg{\Toks}$. 
	note that
	since $\mconf_p \in \Interp{\gamma}$, we have that $\mconf \in
	\Interp{\gamma}$. Furthermore, we have $\mconf \geq \mconf_p$, hence using
	again Lemma \ref{lem:monotonicity}, we know that there exists a
	configuration $\mconf''_p$ such that $\mconf \mtrans^\ast \mconf''_p$ and  $\mconf''_p
	\geq \mconf'_p$, i.e. 
	\begin{align*}
		&\mconf''_p(p) \geq n+2\cdot\ell_U+(k\cdot\ell)\\
		&\mconf''_p(q) \geq \mconf_U(q)+(k\cdot\ell) + \mconf_p(q) \mbox{ for all } q\in S.
	\end{align*}
%	$\mconf''_p(p) \geq n+2\cdot\ell_U+(k\cdot\ell)$ and
%	$\mconf''_p(q) \geq \mconf_U(q)+(k\cdot\ell) + \mconf_p(q)$ for all $q \in S$ by~\Cref{lem:monotonicity}, \autoref{it:lem-1}.
	%\nas{Je crois que c'est $C''_p(q) \geq C_U(q)+(k*\ell)+C_p(q)$, non?}.

	We now prove the following claim.
	\begin{claim}
		For all 
		$1\leq i< j\leq k$, for all $(q_i,m_i), (q_j,m_j)\in\Toks$, $m_i\notin\Rec{q_j}$ and $m_j\notin\Rec{q_i}$. 
	\end{claim}
\begin{proof}
Assume there exists $1\leq i< j\leq k$ such that $(q_i,m_i),(q_j,m_j)\in\Toks$ and $\mconf_U(q_i)=\mconf_U(q_j)=1$, and $m_i\in\Rec{q_j}$ and $m_j\in \Rec{q_i}$, because $\gamma$ is consistent. 
Since $\mconf_U$ respects $\Interp{\gamma}$,  $q_i$ and $q_j$ are conflict-free: there exist $(q_i,m), (q_j,m')\in\Toks$ such that $m\notin\Rec{q_j}$ and 
$m'\notin\Rec{q_i}$. Hence, $(q_i,m_i), (q_i, m), (q_j,m_j), (q_j,m')\in\Toks$, and $m\notin\Rec{q_j}$ and $m_j\in\Rec{q_i}$. Therefore, we have 
$(q_i,m), (q_j,m_j)\in\Toks$ and $m\notin\Rec{q_j}$ and $m_j\in\Rec{q_i}$, which is in contradiction with the fact that $\gamma$ is consistent. Hence, for all 
$1\leq i< j\leq k$, for all $(q_i,m_i), (q_j,m_j)\in\Toks$, $m_i\notin\Rec{q_j}$ and $m_j\notin\Rec{q_i}$. 
\end{proof}

	%Since $C_U \in \Interp{\gamma}$ is consistent, for all $q_i, q_j$ with $i \ne j$, there exists $(q_i, m)$ and $(q_j, m')$ such that $m \nin \Rec{q_j}$ and $m' \nin \Rec{q_i}$ (there are conflict-free). As $\gamma$ is consistent, for all tokens $(q_i, m_i)$, $(q_j, m_j) \in \Toks$, either $m_i \nin \Rec{q_j}$ and $m_j \nin \Rec{q_i}$ or $m_i \in \Rec{q_j}$ and $m_j \in \Rec{q_i}$. Assume we are in the latter case, then the pair of tokens $(q_i,m), (q_j,m_j)$ is such that $m  \in \Rec{q_j}$ and $m_j \nin \Rec{q_i}$ which is not coherent with $\gamma$'s consistency, as a consequence, if $q_i$ and $q_j$ are conflict-free, for all tokens $(q_i, m_i), (q_j, m_j) \in \Toks$, it holds that $m_i\nin\Rec{q_j}$ and $m_j \nin \Rec{q_i}$. As a consequence, for all $i$, for all $j \ne i$, $m_j \nin \Rec{q_i}$ and $m_i \nin \Rec{q_j}$.
	%
	
	We shall now explain how from $\mconf''_p$ we reach $\mconf''_U$ in $k\cdot\ell$
	steps, i.e. how we put (at least) one process in
	each state $q_j$ such that $q_j \in \starg{\Toks}$ and $\mconf_U(q_j)=1$
	in order to obtain a configuration $\mconf''_U \geq \mconf_U$.  We begin by
	$q_1$. Let a process on $q_{0,1}$ send the message $m_1$ (remember
	that $q_{0,1}$ belongs to $S$) and let $\ell_{1}$ other processes on
	states of $S$ send the messages needed for the process to reach
	$q_1$ following the path
	$(q_{0,1},!m_1,q_{1,1})(q_{1,1},?m_{1,1},q_{2,1})\ldots(q_{\ell_1,1},?m_{\ell_1,1},q_1)$. At this stage, we have that the number of processes in each state $q$ in $S$
	is bigger than $\mconf_U(q)+((k-1)\cdot\ell) + \mconf_p(q)$ and we have (at least) one process in
	$q_1$. We proceed similarly to put a process in $q_2$, note that the
	message $m_2$ sent at the beginning of the path cannot be received by the
	process in $q_1$ since, as explained above, $m_2 \notin \Rec{q_1}$.
	
	We proceed again to put a process in the
	states $q_1$ to $q_K$ and at the end we obtain the configuration
	$\mconf''_U$ with the desired properties.
	%
	% \lugtext{ici il y a un problème : $q_1$ and $q_2$ conflict-free -> il existe deux tokens tels que ...
		% 
		%les tokens "conflict-free" pour $q_1$ et $q_2$ et pour $q_2$ et $q_3$ ne sont pas forcément les mêmes, c'est à dire que quand on veut mettre quelqu'un sur $q_3$, si on a "choisit" le mauvais tokens, ce n'est pas (directement) vrai que le message ne sera reçu ni par un processus sur $q_1$ ni sur $q_2$.
		%
		%
		%En fait ça l'est mais ça l'est par construction de l'algorithme. on ne peut donc pas le prouver tout de suite. ou alors il faut changer la définition de conflict-free.. sinon il faut changer l'énoncé de ce lemme et ne parler que des états dans $S$, ce lemme est utilisé qu'avec la partie qui parle des états dans $S$ et je reprouve la partie sur les états tokens plus tard grâce à l'algo }
\end{proof}

\ifappendix The following claim is an important step towards the proof of Lemma~\ref{lem:abstract-soundness}. Recall that $\gamma$ is a consistent token-set of configurations and $\mconf'\in \Interp{F(\gamma)}$. We write $\gamma=(S,\Toks)$ and $F(\gamma)=\gamma'=(S',\Toks')$. 
%We first show that, starting from a configuration in $\Interp{\gamma}$, we can cover all states in $S'$ simultaneously with any number of processes. 
To establish the next claim, we prove that each state in $S'$ is individually coverable by any number of processes, and then we apply Lemma \ref{lem:consistent-reach}.

\begin{claim}\label{claim:aux:soundness:wo:rdv}
	For any $n \in \nat$, there exist $\mconf_n \in \Interp{\gamma}$ and $\mconf'_n \in \mconfs$ such that $\mconf_n \mtrans^\ast \mconf'_n$ and $\mconf'_n(q) \geq n$ for all $q \in S'$.
\end{claim}

The formal proof of this claim can be found in the appendix~\ref{sec:app-rdv}.
\fi

%\begin{proof}
	\input{proof-abstract-soundness}

\subsection{Polynomial Time Algorithm}

We now present a polynomial-time algorithm to solve \CCover~for Wait-Only rendez-vous protocols. 
We define the sequence \((\gamma_n)_{n \in \nat}\) as follows:
\[
\gamma_0=(\set{\qinit},\emptyset) \quad \text{and} \quad \gamma_{i+1}=F(\gamma_i) \quad \text{for all } i \in \nat.
\]
First, note that \(\gamma_0\) is consistent and that \(\Interp{\gamma_0} = \mconfs_{init}\) is the set of initial configurations. 
By Lemma~\ref{lem:F-consistent}, we deduce that \(\gamma_i\) remains consistent for all \(i \in \nat\). 

Each application of \(F\) to a token-set of configurations \((S, \Toks)\) either increases \(S\) or \(\Toks\), or leaves \((S, \Toks)\) unchanged (Lemma \ref{lem:F-increase}). 
Since \(S\) increases at most $|Q|$ times and \(\Toks\) \(|Q|\cdot|\Sigma|\) times, the sequence \((\gamma_n)\) stabilizes within at most \(|Q|^2\cdot|\Sigma|\) iterations. 
Let \(\gamma_f = \gamma_{|Q|^2\cdot|\Sigma|}\). We are now ready to prove the following Lemma:

\begin{lemma}\label{lem:correct-abstraction}
	For any \(\mconf \in \mconfs\), there exist \(\mconf_0 \in \mconfs_{init}\) and \(\mconf' \geq \mconf\) such that \(\mconf_0 \mtrans^\ast \mconf'\) if and only if there exists \(\mconf'' \in \Interp{\gamma_f}\) such that \(\mconf'' \geq \mconf\).
\end{lemma}
\begin{proof}
	Let $\mconf \in \mconfs$. We start by proving completeness: assume there exist \(\mconf_0 \in \mconfsInit\) and \(\mconf' \geq \mconf\) such that \(\mconf_0 \mtrans^\ast \mconf'\). Let $k$ be the length of the run between $\mconf_0$ and $\mconf'$. As $\mconfInit \in \Interp{\gamma_0}$ and thanks to Lemma \ref{lem:abstract-completeness}, $\mconf' \in \Interp{F^k(\gamma_0)}$. From the definition of $\gamma_f$, either $F^k(\gamma_0) = \gamma_f$ or there exists $j >k $ such that $F^j (\gamma_0) = \gamma_f$. In the first case we immediately have that $\mconf' \in \Interp{\gamma_f}$. Otherwise, thanks to Lemma \ref{lem:F-increase}, we get that $\Interp{F^k(\gamma_0)} \subseteq \Interp{F^j(\gamma_0)}$, hence $\mconf' \in \Interp{\gamma_f}$.
	
	We now prove soundness: assume there exists $\mconf'' \in \Interp{\gamma_f}$ such that $\mconf'' \geq \mconf$.
	Let $k$ be the first integer such that $F^k(\gamma_0) = F^{k+1}(\gamma_0) = \gamma_f$. We prove by induction on $i \in [0, k]$ that there exist $\mconf_k \geq \mconf''$ and $\mconf_{k-i} \in \Interp{F^{k-i}(\gamma_0)}$ such that $\mconf_{k-i} \mtrans^\ast \mconf_k$.

	For $i = 0$, $\mconf_{k} = \mconf''$ satisfies the claim. Let now $i \in [0, k-1]$ and assume that there exist $\mconf_k \geq \mconf''$ and $\mconf_{k-i} \in \Interp{F^{k-i}(\gamma_0)}$ such that $\mconf_{k-i} \mtrans^\ast \mconf_k$. As $i < k$, $F^{k-i-1}(\gamma_0)$ is defined (with the convention that $F^0(\gamma_0)  = \gamma_0$). From Lemma \ref{lem:abstract-soundness}, there exists $\mconf'_{k-i} \geq \mconf_{k-i}$ and $\mconf_{k-i-1} \in \Interp{F^{k-i-1}(\gamma_0)}$ such that $\mconf_{k-i-1} \mtrans^\ast \mconf'_{k-i}$. Thanks to Item 2 of Lemma \ref{lem:monotonicity} there exists $\mconf'_k$ such that $$\mconf_{k-i-1} \mtrans^\ast \mconf'_{k-i} \mtrans^\ast \mconf'_k \text{ with } \mconf'_k \geq \mconf_k \geq \mconf''.$$
	Hence, the claim holds for $i+1$ and so for all $i \in [0, k]$. As a consequence, there exists $\mconf_0 \in \Interp{\gamma_0} = \mconfsinit$ and $\mconf_k \geq \mconf'' \geq \mconf$ such that $\mconf_0 \mtrans^\ast \mconf_k$, concluding the proof.
\end{proof}

To compute \(\gamma_f\), we iterate \(F\) at most \(|Q|^2\cdot|\Sigma|\) times. 
Remember that we have assumed that $|Q| \leq |T|$, and $|\Sigma |\leq |T|$.
Each application of \(F\) runs in polynomial time (Lemma \ref{lem:F-consistent}). 
Furthermore, checking whether there exists \(\mconf'' \in \Interp{\gamma_f}\) such that \(\mconf'' \geq \mconf\) can be done in polynomial time by Lemma~\ref{lem:interp-cover-check}. 

Since \statecovernb\ is easier than \confcovernb\ and by Theorem \ref{thm:scover-wo-rdv-phard}, we obtain the desired result.

\begin{theorem}
	\confcovernb~restricted to Wait-Only rendez-vous protocols is P-complete.
\end{theorem}

The drop in complexity from PSPACE-complete to P-complete when moving from Wait-Only broadcast protocols to Wait-Only rendez-vous protocols comes 
from a loss in global synchronization: while a single broadcast can empy several receiving states, in a rendez-vous one sender interacts with at most one
receiver at a time (and do not empty the place). This has consequences on the type of monotonicity ensured: in a broadcast network, population growth 
implies population growth but maintains dependency between states (all the processes of the receiving state being flushed into another one). In a 
rendez-vous network, this dependency is not so strong, hence the ability to cover a configuration can be decoupled into the ability to cover individual
states separately. 

%\lulutex{Arnaud dit : ecrire la preuve}

%
%We get the immediate corollary.
%\begin{corollary}
%	\Cover~restricted to wait-only protocols is in \Ptime.
%\end{corollary}

%% file: Figures/wo-rdv-prot.tex
\tikzset{box/.style={draw, minimum width=4em, text width=4.5em, text centered, minimum height=17em}}

\begin{tikzpicture}[->, >=stealth', shorten >=1pt,node distance=2cm,on grid,auto, initial text = {}] 
	\node[state, initial] (q0) {$\qinit$};
	\node[state] (q1) [right = of q0, yshift = 25] {$q_1$};
	\node[state] (q3) [right  = of q0, yshift = -25] {$q_3$};
	\node[state] (q2) [right  = of q1] {$q_{2}$};
	\node[state] (q4) [right = of q3] {$q_4$};

	\path[->] 
	(q0) edge [thick,bend right = 0] node  []{$!b$} (q3)
	edge [thick,bend left = 0] node  [above, xshift =-2]{$!a$} (q1)
	(q1) edge [thick,bend left = 0] node  [above]{$?a, ?b$} (q2)
	(q3) edge [thick,bend left = 0] node  [below]{$?a, ?b$} (q4)
%	(q4) edge [loop above] node {$?c$} ()
	;
\end{tikzpicture}

%% file: Figures/wo-rdv-prot-2.tex
\tikzset{box/.style={draw, minimum width=4em, text width=4.5em, text centered, minimum height=17em}}

\begin{tikzpicture}[->, >=stealth', shorten >=1pt,node distance=2cm,on grid,auto, initial text = {}] 
	\node[state, initial] (q0) {$\qinit$};
	\node[state] (q1) [right = of q0, yshift = 25] {$q_1$};
	\node[state] (q3) [right  = of q0, yshift = -25] {$q_3$};
	\node[state] (q2) [right  = of q1] {$q_{2}$};
	\node[state] (q4) [right = of q3] {$q_4$};

	\path[->] 
	(q0) edge [thick,bend right = 0] node  []{$!b$} (q3)
	edge [thick,bend left = 0] node  [above, xshift =-2]{$!a$} (q1)
	(q1) edge [thick,bend left = 0] node  [above]{$?a$} (q2)
	(q3) edge [thick,bend left = 0] node  [below]{$?a, ?b$} (q4)
	%	(q4) edge [loop above] node {$?c$} ()
	;
\end{tikzpicture}

%% file: Figures/example-wo.tex
\begin{tikzpicture}[->, >=stealth', shorten >=1pt,node distance=2cm,on grid,auto, initial text = {}] 
	\node[state, initial above] (q0) {$\qinit$};
	\node[state] (q1) [ left = of q0] {$q_1$};
	\node[state] (q2) [ left = of q1, yshift = -0.75cm] {$q_2$};
	\node[state] (q3) [below = 1.5 of q1] {$q_3$};
	\node[state] (q4) [ below = 1.5 of q0] {$q_4$};
	\node[state] (q5) [ right = of q0] {$q_5$};
	\node[state] (q6) [ below = 1.5 of q5] {$q_6$};
	\node[state] (q7) [ right = of q5] {$q_7$};

%	\node[state] (q6) [ left = of q5] {$q_3$};
%	\node[state] (q6) [ right = of q5] {$q_6$}; 
%	\node[state] (q7) [ right = of q6] {$q_7$}; 	
%	\node[state] (q8) [ right =of q5, xshift = 1cm, yshift = -1cm] {$q_8$};
%	\node[state] (q9) [ right = of q8 ] {$q_6$};
%	\node[state] (q10) [  right= of q8, yshift = -30] {$q_6$};
%	
%	\node[state] (q2) [right = of q1] {$q_2$};
	
	\path[->] 
	(q0) edge [bend right = 15] node [above] {$!a$} (q1)
        	edge [bend left = 15] node {$!b$} (q1)
        	edge node {$!d$} (q4)
        	edge node {$!c$} (q5)
	(q1) edge [bend right = 15] node [above] {$?a,?b$} (q2)
			edge node {$?c$} (q3)
	(q3) edge [bend left = 15] node {$?a,?b$} (q2)

	(q5) edge node {$?c$} (q6)
	(q5) edge node {$?d$} (q7)
%	edge   node [above] {$?b$} (q4)
%	edge   node  {$!b$} (q6)
%	(q6) edge [bend right] node {$?c$} (q2)
%	%					 edge  [bend right] node [below left] {$?d$} (q6)
	;

\end{tikzpicture}

%% file: Figures/example-wo-2.tex
\begin{tikzpicture}[->, >=stealth', shorten >=1pt,node distance=2cm,on grid,auto, initial text = {}] 
	\node[state, initial above] (q0) {$\qinit$};
	\node[state] (q1) [ above = 1 of q0, xshift = -1.5cm] {$q_1$};
	\node[state] (q2) [ below = 1 of q0, xshift = -1.5cm] {$q_2$};
	\node[state] (q3) [left = 3  of q0] {$q_3$};
	\node[state] (q5) [ right = 2.5 of q0] {$p_2$};
	\node[state] (q4) [ above  = 1 of q5] {$p_1$};
	\node[state] (q6) [ below = 1 of q5] {$p_3$};
	\node[state] (q7) [ right = 5 of q0] {$p_4$};
	
	%	\node[state] (q6) [ left = of q5] {$q_3$};
	%	\node[state] (q6) [ right = of q5] {$q_6$}; 
	%	\node[state] (q7) [ right = of q6] {$q_7$}; 	
	%	\node[state] (q8) [ right =of q5, xshift = 1cm, yshift = -1cm] {$q_8$};
	%	\node[state] (q9) [ right = of q8 ] {$q_6$};
	%	\node[state] (q10) [  right= of q8, yshift = -30] {$q_6$};
	%	
	%	\node[state] (q2) [right = of q1] {$q_2$};
	
	\path[->] 
	(q0) edge [bend right = 15] node [above, xshift = 5] {$!a$} (q1)
	edge [bend left = 15] node {$!b$} (q2)
	edge [bend left = 15]  node {$!m_1$} (q4)
	edge node {$!m_2$} (q5)
	edge [bend right = 15]  node [below] {$!m_3$} (q6)
	(q1) edge [bend right = 15] node [above, xshift = - 5] {$?a$} (q3)
	(q2) edge [bend left = 15] node {$?a, ?b$} (q3)
	
	(q4) edge [bend left = 15] node [xshift= -5] {$?m_1, ?m_3$} (q7)
	(q5) edge node {$?m_2, ?m_3$} (q7)
	(q6) edge [bend right = 15] node [below ,xshift = 12] {$?m_1, ?m_2, ?m_3$} (q7)
	%	edge   node [above] {$?b$} (q4)
	%	edge   node  {$!b$} (q6)
	%	(q6) edge [bend right] node {$?c$} (q2)
	%	%					 edge  [bend right] node [below left] {$?d$} (q6)
	;

\end{tikzpicture}

%% file: proof-abstract-completeness.tex
Let $\gamma = (S,\Toks)\in\Gamma$ be a consistent abstract set of configurations, and $\mconf \in \mconfs$ such that $\mconf \in  \Interp{\gamma}$ and $\mconf \mtrans \mconf'$. We denote by $(S', \Toks')$ the token-set $F(\gamma) $ and by $\gamma'' = (S'', \Toks'')$ the intermediate sets used to compute $F(\gamma)$.
We will first prove that:
\begin{enumerate}[(1)]
	\item for all state $q$ such that $\mconf'(q) > 0$, $q \in S'$ or $q \in \mst(\Toks')$,
\end{enumerate}
  and then we will prove that 
  \begin{enumerate}[(2)]
  \item for all states $q$ such that $q \in \mst(\Toks')$ and $\mconf'(q)>0$, (2a) $\mconf'(q) = 1$ and (2b) for all other states $p\in \mst(\Toks')$ such that $\mconf'(p) >0$, $p$ and $q$ are "conflict-free". 
\end{enumerate}

Observe that $S \subseteq S'' \subseteq S'$, $\Toks \subseteq \Toks'' $, and $\mst(\Toks'') \subseteq \mst(\Toks') \cup S'$. 
%We shall prove that for all $q \in Q$ such that $\mconf'(q) > 0$, either (1) $q \in S'$ or (2) $q\in \mst(\Toks')$ and $\mconf'(q) = 1$ and for all $q' \in \mst(\Toks') \setminus \{q\}$ such that $\mconf'(q') = 1$, we have that $q$ and $q'$ are conflict-free. 

%\begin{proof-far}[(1)]
\textit{Proof of (1).}
	First, let us prove that for every state $q$ such that $\mconf'(q)>0$, it holds that $q \in S' \cup \mst(\Toks')$.
	Note that for all $q$ such that $\mconf(q) > 0$, because $\mconf$ respects $\gamma$, $q \in \mst(\Toks) \cup S$. As $\mst(\Toks) \cup S \subseteq \mst(\Toks') \cup S'$, the property holds for $q$.
	Hence, we only need to consider states $q$ such that $\mconf(q) = 0$ and $\mconf'(q) > 0$.
	
%	 If $\mconf \mtransup{\tau} \mconf'$ then $q$ is such that there exists $(q', \tau, q) \in T$, $q'$ is therefore an active state and so $q' \in S$, (recall that $\Toks \subseteq Q_W \times \Sigma$). Hence, $q$ should be added to $\mst(\Toks'') \cup S''$ by condition \ref{ccover-wo-F-cond-internal}. 
%	 As $\mst(\Toks'') \cup S'' \subseteq \mst(\Toks') \cup S'$, it concludes this case.
	 
	  If $\mconf \mtransup{(q', !a, q)} \mconf'$ then $q'$ is an action state. 
	  Then $q' \in S$, (recall that $\Toks \subseteq Q_W \times \Sigma$). Hence, $q$ should be added to $\mst(\Toks'') \cup S''$ by condition \ref{ccover-wo-F-cond-send-S}~or \ref{ccover-wo-F-cond-newtok}. 
	  
	  If $\mconf \mtransup{(p, !a, p')} \mconf'$ and $p' \neq q$, then $(q', ?a, q) \in T$ with $q' \in S \cup \mst(\Toks)$, and it should be added to $\mst(\Toks'')\cup S''$ by condition \ref{ccover-wo-F-cond-reception-S}, \ref{ccover-wo-F-cond-tok-end}, or \ref{ccover-wo-F-cond-tok-step}. Therefore, we proved that for all state $q$ such that $\mconf'(q) >0$, it holds that $q \in \mst(\Toks') \cup S'$. 
	
%\end{proof-far}
 
\emph{The rest of the proof is devoted to prove (2).}

%, \ie that 
%if $q \in \mst(\Toks)$, then $\mconf'(q) = 1$ and for all $q' \in \mst(\Toks') \setminus \{q\}$ such that $\mconf'(q') = 1$, we have that $q$ and $q'$ are conflict-free. 
Let $q$ such that $\mconf'(q) >0$.
Note that if $q \in \mst(\Toks)$ and $\mconf(q) = \mconf'(q) = 1$, then for every state $p$ such that $p \in \mst(\Toks)$ and $\mconf(p) = \mconf'(p) = 1$, it holds that $q$ and $p$ are conflict-free.

%Observe that if $\mconf \mtransup{\tau} \mconf'$, then note $q$ the state such that $(q', \tau ,q)$, it holds that $\{p \mid p \in \mst(\Toks') \textrm{ and } \mconf'(p) > 0\} \subseteq \{p \mid p \in \mst(\Toks) \textrm{ and } \mconf(p) = 1\}$: $q'$ is an active state, $q$ might be in $\mst(\Toks)$ but it is added to $S'' \subseteq S'$ with rule \ref{ccover-wo-F-cond-internal}, and for all other states, $\mconf'(p) = \mconf(p)$. If $p \in \mst(\Toks')$ and $\mconf(p) > 0$, it implies that $\mconf'(p)= \mconf(p) = 1$ and $p\in \mst(\Toks)$ (otherwise $p$ is in $S \subseteq S'$). Hence, there is nothing to do as $\mconf$ respects $\gamma$.

Let $q \in \mst(\Toks') \setminus \mst(\Toks)$ with $\mconf'(q) > 0$, we shall prove that $\mconf'(q) =1$ and for all $p \in \mst(\Toks')$ and $\mconf'(p) > 0$, $q$ and $p$ are conflict-free. If $q \in \mst(\Toks') \setminus \mst(\Toks)$, it implies that $\mconf(q) = 0$ because $\mconf$ respects $\gamma$. Hence, we fall in one of the following cases:
\begin{itemize}[]
	\item \textbf{Case A.} $\mconf \mtransse{(q', !a, q)} \mconf'$ for a transition $(q', !a, q) \in T$ and $\mconf' = \mconf - \mset{q'} + \mset{q}$.
	%	\lug{j'ai introduit cette notation avec une footnote...} with $(q', !a, q) \in T$;
	\item \textbf{Case B.} $\mconf \mtransrdv{(q_1, !a, q_1')} \mconf'$ for a transition $(q_1, !a, q'_1) \in T$ and there exists $(q_2, ?a, q'_2) \in T$ with $\mconf' = \mconf - \mset{q_1, q_2} + \mset{q_1', q_2'}$. In that case, $q = q'_1$ or $q=q'_2$.Then we should be careful as we need to prove that $q'_2 \ne q'_1$, otherwise, $\mconf'(q) = 2$. 
\end{itemize}

\textbf{Case A:} Note that as only one process moves between $\mconf$ and $\mconf'$ and $\mconf(q)= 0$, it is trivial that $\mconf'(q) = 1$ (2a). In this case, as it is a sending of $a$ between $\mconf$ and $\mconf'$ (and not a rendez-vous), it holds that: for all $p \in \mst(\Toks)$ such that $\mconf(p) = 1$, $a \notin \Rec{p}$ (as otherwise, the process in $p$ needs to receive $a$). 
Take $p  \in\mst(\Toks')$, such that $p\ne q$ and $\mconf'(p) = 1$. Observe that $(q, a) \in \Toks'$ by construction. Then $\mconf'(p) = \mconf(p) = 1$ and so $p \in \mst(\Toks)$, and $a \notin \Rec{p}$. Suppose $(p, m) \in \Toks'$ such that $m \in \Rec{q}$, then we found two tokens in $\Toks'$ such that $m \in \Rec{q}$ and $a \notin \Rec{p}$ which contradicts $F(\gamma)$'s consistency (Lemma \ref{lem:F-consistent}). Hence, $p$ and $q$ are conflict-free (2b).\\
%Suppose $p$ and $q$ are not conflict-free in $F(\gamma)$, then for all $(p, m), (q, m') \in \Toks'$, $m\in \Rec{q}$ or $m'\in \Rec{p}$. As $F(\gamma)$ is consistent by Lemma \ref{lem:F-consistent}, for all $(p, m), (q, m') \in \Toks'$, $m\in \Rec{q}$ \emph{and} $m'\in \Rec{p}$. 

%As $q \in \mst(\Toks')$ and $q' \in S$, it should be from condition \ref{ccover-wo-F-cond-newtok}~that $(q, a) \in \Toks''$. Note that $\Toks' \subseteq \Toks''$ and if $(q,a) \notin \Toks'$, it should be that $q \in S'$. Hence, $(q,a)\in \Toks'$. Furthermore, for all $p \in \Toks$ such that $\mconf(p) > 1$ and $p \notin q$, it holds that $a \notin \Rec{p}$ as the rendez-vous is not answered. By construction $\mconf'(p) = \mconf(p)$, hence, we found two tokens $(p,m), (q,a)$ such that $a \notin \Rec{p}$ and $m \in \Rec{q}$, which is absurd given $F(\gamma)$'s consistency. Hence for all $p  \mst(\Toks')$, such that $p\ne q$, it holds that $p$ and $q$ are conflict-free.

\textbf{Case B:} 
We start by proving some small facts.
\begin{enumerate}[(i)]
	\item \emph{If $q'_2 \in \mst(\Toks')$, then $q_2 \in \mst(\Toks)$.}\\
	\emph{Proof}. Otherwise, $q'_2$ should be in $S'$ by condition \ref{ccover-wo-F-cond-reception-S}. 
	\item \emph{If $q'_1 \in \mst(\Toks')$, then $a \in \Rec{q'_1}$.}\\
	\emph{Proof.} Otherwise, $q'_1$ should be in $S'$ by condition \ref{ccover-wo-F-cond-send-S}. 
	\item \emph{If $q'_1 \in \mst(\Toks')$, then $(q'_1 ,a) \in \Toks'$.}\\
	\emph{Proof.}  From (ii) and condition \ref{ccover-wo-F-cond-newtok}. 
	\item \emph{If $q'_1 \in \mst(\Toks')$, then $q_2 \in \mst(\Toks)$}.\\
	\emph{Proof.} Otherwise $q'_1$ should be added to $S'$ by condition \ref{ccover-wo-F-cond-send-S}. 
\end{enumerate}

%and $(q'_1 ,a) \in \Toks'$ by condition \ref{ccover-wo-F-cond-newtok}. `
%Furthermore, if $q'_1 \in \mst(\Toks')$, $q_2 \in \mst(\Toks)$ as well as otherwise $q'_1$ should be added to $S'$ by condition \ref{ccover-wo-F-cond-send-S}.

\emph{We now prove that either $q'_1 \in S'$, or $q'_2 \in S'$.}\\
 For the sake of contradiction, assume this is not the case.
Then, thanks to (i), let $(q_2, m) \in \Toks$, with $(q'_2, m) \in \Toks'$ (by condition \ref{ccover-wo-F-cond-tok-step} thre must exist two such tokens with the same message $m$).  
Then, together with (iii), there are three tokens $(q'_1, a), (q_2, m), (q'_2, m) \in \Toks' \subseteq \Toks''$, such that $(q_2, ?a, q'_2) \in T$. If $m \neq a$, from condition \ref{ccover-wo-F-cond-3toks-1}, $q'_1$ should be added to $S'$ and so $(q'_1, a) \notin \Toks'$. It is not possible that $m = a$, as otherwise, $q'_2 \in S'$ from condition \ref{ccover-wo-F-cond-reception-S}.
Hence, $q'_1 \in S'$, or $q'_2 \in S'$.\\

Note that, as a consequence $q'_1 \ne q'_2$ or $q'_1 = q'_2 \in S'$. Take $q \in \mst(\Toks') \setminus \mst(\Toks)$ such that $\mconf'(q) >0$, if such a $q$ exists, then $q = q'_1$ or $q = q'_2$ and $q'_1 \ne q'_2$. As a consequence, $\mconf'(q) = 1$ (note that if $q'_1 = q_2$, $\mconf(q_2) = 1$) which proves (2a). 

It is left to prove (2b).
Take $p \in \mst(\Toks') \setminus \{q\}$ such that $\mconf'(p) > 0$, we now prove that $q$ and $p$ are conflict-free, i.e.\ (2b). 
We first justify that $p \nin \set{q_1, q_1', q_2, q_2'}$. 
As $p \neq q$ and $q \in \set{q_1', q_2'}$ and ($q_1' \in S'$ or $q_2' \in S'$), we deduce that $p\nin \set{q_1', q_2'}$.
If $p = q_1$, then $\mconf'(q_1) = 1$. As $\mconf' = \mconf - \mset{q_1, q_2} + \mset{q_1', q_2'}$, it implies that
\[
\mconf(q_1) = 2 \mbox{ or } p = q_1 = q_1' \mbox{ or } p = q_1 = q_2'.
\]
In the first case, it implies that $p = q_1 \in S$ because $\mconf$ "respects" $\gamma$, which is absurd given that $p \in \mst(\Toks')$.
Hence $\mconf(p)= \mconf'(p)$ and $p \in \mst(\Toks)$.

The two latter cases are absurd as we justified that $p\nin \set{q_1', q_2'}$.
%If $p \ne q$ and $p \in \mst(\Toks')$, then $\mconf'(p) = \mconf(p)$ (because $q'_1 \in S'$ or $q'_2 \in S'$). Hence, $p \in \mst(\Toks)$ and $\mconf'(p) = 1$.
We now distinguish between cases.
\begin{itemize}
	\item \textbf{(Case $q =q'_1$)} Assume $q = q'_1$ and assume $q$ and $p$ are not "conflict-free". 
%	Remember that we justified that $q_2 \in \mst(\Toks)$, and therefore, $\mconf(q_2) = 1$. 
%	Hence, either $\mconf'(q_2) = 0$, or $q_2 = q'_2$ (and in that case $q_2,q_2' \in S'$) or $q_2 = q_1'$ (and then $q_2=q$). In any case, $p \ne q_2$.\\
As $q_1' \in \mst(\Toks' )$, it holds that $q_2 \in \mst(\Toks )$ by (iv).
	As $\mconf$ "respects" $\gamma$, there exists $(p, m_p)$ and $(q_2, m) \in \Toks$ such that $m_p \notin \Rec{q_2}$ and $m \notin \Rec{p}$ ($q_2$ and $p$ are conflict-free). As $p \in\mst(\Toks')$, $(p,m_p) \in \Toks'$ and so $m_p\in \Rec{q}$ or $a \in \Rec{p}$ ($q$ and $p$ are not conflict-free).
	As $F(\gamma)$ is consistent, $m_p\in \Rec{q}$ and $a \in \Rec{p}$.
	Note that $a \ne m_p$ because $a \in \Rec{q_2}$, $a \ne m$ because $m \notin \Rec{p}$, and obviously $m \ne m_p$. Note also that if $m \notin \Rec{q}$, then we found two tokens $(q,a)$ and $(q_2,m)$ in $\Toks'$ such that $a \in \Rec{q_2}$ and $m \notin \Rec{q}$, which contradicts the fact that $F(\gamma)$ is consistent (Lemma \ref{lem:F-consistent}). Hence, $m\in \Rec{q}$.
	Note that even if $q_2$ is added to $S''$, it still is in $\Toks''$.  As $\Toks' \subseteq \Toks''$ we found three tokens $(p, m_p), (q_2,m)$, $(q, a)$ in $\Toks''$, satisfying condition \ref{ccover-wo-F-cond-3toks-2}, and so $p$ should be added to $S'$, which is absurd as $p \in \mst(\Toks')$. We reach a contradiction and so $q$ and $p$ should be conflict-free.
	
	\item \textbf{(Case $q =q_2'$)} Finally assume $q = q_2'$. 
	As $q_2' \in \mst(\Toks' )$, it holds that $q_2 \in \mst(\Toks )$ by (i).
	If $q = q_2$, then, because $\mconf$ respects $\gamma$, $q$ and $p$ are conflict-free. Otherwise, as $q_2$ is conflict-free with $p$, there exists $(q_2, m )$ and $(p, m_p)$ in $\Toks$ such that $m \notin \Rec{p}$ and $m_p \notin \Rec{q_2}$. Note that $(q,m) \in \Toks''$ from condition \ref{ccover-wo-F-cond-tok-step}~(otherwise, $q \in S''$ which is absurd). Hence, $(q, m) \in \Toks'$ and, as $p \in \mst(\Toks')$, $(p,m_p)$ is conserved from $\Toks$ to $\Toks'$. It remains to show that $m_p \notin \Rec{q}$. Assume this is not the case, then there exists $(p,m_p)$ and $(q,m) \in \Toks'$ such that $m\notin \Rec{p}$ and $m_p\in \Rec{q}$ which is absurd given $F(\gamma)$'s consistency.
	As a consequence, $q$ and $p$ are conflict-free. 
\end{itemize}

% If $m_p \notin \Rec{q}$ or $a \notin \Rec{p}$, then the two tokens $(p,m_p)$ and $(q, a)$ satisifies condition \ref{ccover-wo-F-cond-2toks-1}, and so either $p$ or $q$ should be added to $S'$, which can not be as both states are in $\mst(\Toks')$. 

%

%and from conditon \ref{ccover-wo-F-cond-2toks-1}, $q$ should be added to $S'$, which is absurd as $q \in \mst(\Toks')$. 
%

We managed to prove that for all $q$ such that $\mconf'(q) >0$, $q \in S' \cup \mst(\Toks')$, and if $q \in \mst(\Toks')$, then $\mconf'(q) = 1$ and for all others $p\in \mst(\Toks')$ such that $\mconf'(p) = 1$, $p$ and $q$ are conflict-free.

%% file: proof-abstract-soundness.tex

\ifappendix\else The above claim shows that, from a configuration in $\Interp{\gamma}$, we can cover $S'$ with any number of processes simultaneously. \fi
To prove the soundness, it remains to show that we can also cover $\Toks'$ (simultaneously with $S'$). 
However, we do not need to (and, in fact, cannot) cover all states in $\mst(\Toks')$ simultaneously: we only need to cover the states occupied by $\mconf'$. 
%Since $\mconf' \in \Interp{F(\gamma)}$ and $F(\gamma)$ is consistent (\Cref{lem:F-consistent}), the states in $\Toks'$ that are \emph{conflict-free} are sufficient.\lulu{Arnaud dit: pourquoi ? cette phrase ne sert pas en fait }
In fact, this follows from the definition of conflict-freeness, as explained in the proof of Lemma \ref{lem:abstract-soundness} below. 

\begin{proofof}{Lemma~\ref{lem:abstract-soundness}}
	We have $\gamma$ a consistent token-set of configurations and we let $\mconf' \in \Interp{F(\gamma)}$.
	
	Let us recall what we want to prove: there exists $\mconf \in \Interp{\gamma}$ and $\mconf'' \in \mconfs$ such that $\mconf \mtrans^\ast \mconf''$ and $\mconf'' \geq \mconf'$.
	
	Thanks to Claim \ref{claim:aux:soundness:wo:rdv}, we know how to build for any $n \in \nat$, a configuration $\mconf'_n$ such that $\mconf'_n(q) \geq n$ for all states $q \in S'$ and a configuration $\mconf_n \in \Interp{\gamma}$, such that $\mconf_n \mtransup{}^\ast \mconf'_n$. In particular for $n$ bigger than the maximal value $\mconf'(q)$ for $q \in S'$, $\mconf'_n$ is greater than $\mconf'$ on all the states in $S'$.  
	
	We need to prove that from a configuration $\mconf'_{n}$ for a particular $n \in \nat$, we can reach a configuration $\mconf''$ such that $\mconf''(q) \geq \mconf'(q)$ for $q \in S' \cup \mst(\Toks')$. As $\mconf'$ respects $F(\gamma)$, remember that for all $q \in \mst(\Toks')$, $\mconf'(q) \leq 1$. The run is actually built in the same fashion as the one we built at the end of the proof of Lemma \ref{lem:consistent-reach}.

	We enumerate states $q_1, \dots, q_m$ in $ \mst(\Toks')$ such that $\mconf'(q_i) = 1$. 
	%As $\mconf'$ respects $F(\gamma)$, for $i \ne j$, $q_i$ and $q_j$ are conflict free.
	From Lemma \ref{lem:F-consistent}, $F(\gamma)$ is consistent, and so we write $(p^j_0, !m^j, p^j_1) $ $(p^j_1, ?m^j_1, p^j_2)$ $ \dots $ $(p^j_{k_j}, ?m^j_{k_j}, p^j_{k_j+1})$ the sequence of transitions associated to state $q_j$ such that: $p^j_{k_j+1} = q_j$, and $(q_j, m^j) \in \Toks$ and for all $m^j_i$, there exists $(q_{m^j_i}, !m_i^j, q'_{m^j_i})$ with $q_{m^j_i}\in S'$.
	Note that for all $i \ne j$, $q_i$ and $q_j$ are conflict-free (because $\mconf'$ respects $F(\gamma)$) and so there exists $(q_i, m), (q_j,m') \in \Toks'$ such that $m \notin \Rec{q_j}$ and $m' \notin \Rec{q_i}$. As $F(\gamma)$ is consistent, it should be the case for all pairs of tokens $(q_i, a), (q_j, a')$. Hence $m^j \notin \Rec{q_i}$ and $m^i \notin \Rec{q_j}$. 
	
	Let $\ell_j = k_j + 1$ and $n_{\max}$ be the maximum value for any $\mconf'(q)$, \ie\ $n_{\max} = \max_{q\in Q} \mconf'(q)$. 
	For $n =  n_{\max} + \sum_{1\leq j \leq m} \ell_j$, there exists a configuration $\mconf'_{n}$ such that there exists $\mconf_{n} \in \Interp{\gamma}$, $\mconf_{n}\mtransup{}^*\mconf'_{n}$, and $\mconf'_{n}(q) \geq n$ for all $q \in S'$ (Claim \ref{claim:aux:soundness:wo:rdv}). In particular, for all $q \in S'$, $\mconf'_{n}(q) \geq \mconf'(q)  + \sum_{1\leq j \leq m} \ell_j$.

	Then, we still have to build a run leading to a configuration $\mconf''$ such that for all $q \in \mst(\Toks')$, $\mconf''(q) \geq \mconf'(q)$. We then use the sequences of transitions associated to each state $q_j$. With $\ell_1$ processes we can reach a configuration $\mconf'_{n + \ell_1}$ such that $\mconf'_{\ell_1}(q_1) \geq 1$: 
	\[
	\mconf'_{n}  \mtransup{(p_0^1, !{m^1}, p_1^1)} \mconf'_{n+1} \mtransrdv{(q_{m_1^1}, !m_1^1, q'_{m_1^1})} \dots \mtransrdv{(q_{m_{k_1}^1}, !m_{k_1}^1, q'_{m_{k_1}^1})} \mconf'_{n+\ell_1 }
	\]
	If there is no process on $\staterec{m^1}$ then $\mconf'_{n}  \mtransse{(p_0^1, !{m^1}, p_1^1)} \mconf'_{n+1}$, and $\mconf'_{n}  \mtransrdv{(p_0^1, !{m^1}, p_1^1)} \mconf'_{n+1}$ otherwise.
	Each rendez-vous on $m_i^1$ is made with transitions $(p_i^1, ?m_i^1, p_{i+1}^1)$ and $(q_{m_i^1}, ! m_i^1, q'_{m_i^1})$. As a result, for all $q \in S'$, $\mconf'_{n +\ell_1}(q) \geq \mconf'(q) +\sum_{2\leq j \leq m} \ell_j$ and $\mconf'_{n +\ell_1}(q_1) \geq 1$. We then do the following execution from $\mconf'_{n +\ell_1}$: 
	\[
	\mconf'_{ n + \ell_1 } \mtransup{(p_0^2, !{m^2}, p_1^2)} \mconf'_{n+\ell_1+1} \mtransrdv{(q_{m_1^2}, !m_1^2, q'_{m_1^2})} \dots  \mtransrdv{(q_{m_{k_2}^2}, !m_{k_2}^2, q'_{m_{k_2}^2})} \mconf'_{n +\ell_1+ \ell_2 }
	\]
	 If there is no process on $\staterec{m^2}$, then $\mconf'_{ n + \ell_1 } \mtransse{(p_0^2, !{m^2}, p_1^2)} \mconf'_{n+\ell_1+1}$, otherwise $\mconf'_{ n + \ell_1 } \mtransse{(p_0^2, !{m^2}, p_1^2)} \mconf'_{n+\ell_1+1}$.
	  Remember that we argued that $m^2 \notin \Rec{q_1}$, and therefore $\mconf'_{n + \ell_1+ \ell_2}(q_1) \geq \mconf'_{n +\ell_1 }(q_1) \geq 1$.
	
%	\lulutex{is it enough from there}
	Each rendez-vous on $m_i^2$ is made with transitions $(p_i^2, ?m_i^2, p_{i+1}^2)$ and $(q_{m_i^2}, ! m_i^2, q'{m_i^2})$. As a result, $\mconf'_{n + \ell_1+\ell_2 }(q) \geq \mconf'(q) +\sum_{3\leq j \leq m} \ell_j$ for all $q \in S'$ and $\mconf'_{n + \ell_1+ \ell_2 } \geq \mset{q_1, q_2}$. 

	We can then repeat the reasoning for each state $q_i$ and so reach a configuration $\mconf''$ such that $\mconf''(q) \geq \mconf'(q)$ for all $q \in S'$ and, $\mconf'' \geq \mset{q_1, q_2, \dots, q_m}$.

	We built the following execution: $\mconf_{n} \mtrans^\ast \mconf'_{n} \mtrans^\ast \mconf''$, such that $\mconf''\geq \mconf'$, and $\mconf'_{n} \in \Interp{\gamma}$ which concludes the proof.
	
\end{proofof}

%% file: conclusion.tex
\section{Conclusion}

We have provided a complete analysis of the coverability problems (both of a state and a configuration) of Wait-Only networks where processes communicate via broadcasts, non-blocking rendez-vous, or both. All the upper bounds that we presented (P for state coverability, \pspace~for configuration coverability, and P for configuration coverability with networks restricted to non-blocking rendez-vous) rely on the \emph{copypaste property} whose crucial point is that if an \emph{action} state and another state (action or waiting) are both coverable, then they are coverable \emph{together}. From this property, we were able to derive a simple saturation algorithm solving \SCover~in networks where processes communicate via broadcasts, non-blocking rendez-vous, or both. This property is also a crucial stone of the \pspace~algorithm solving \CCover, as it allowed us to define an abstract semantics on our networks, deciding \CCover~with a bounded number of processes. We also showed that when processes communicate only via non-blocking rendez-vous, \CCover~becomes easier (P-complete), and we can even compute the exact set of coverable configurations. 

Future work includes exploring liveness properties. We have proved in~\cite{GSS25} that repeated coverability is in \textsc{Expspace} and \textsc{Pspace}-hard when
one considers Wait Only protocols that use only broadcast (and no non-blocking rendez-vous). This is in contrast with broadcast protocols that are not Wait-Only,
where this problem is undecidable~\cite{esparza-verif-lics99,EmersonK03}. In addition to closing this complexity gap, the
decidability status of the same problem for general Wait Only non-blocking broadcast protocols remains to be determined. 

%\lugtext{to change}
%
%We have proved that when extending the model presented in \cite{guillou-safety-concur23}\ with broadcasts, \SCover\ for Wait-Only protocols remains in P, and is even P-complete. We also explained how to retrieve a P lower bound for the model of \cite{guillou-safety-concur23} for both problems restricted to Wait-Only protocols. However \CCover\ for Wait-Only protocols is now \pspace-complete. In the future, we wish to study not only coverability problems but extend the analysis of this model to liveness properties. We also wish to expand this model with dynamic creations of messages and processes in order to take a step closer to the modelling of Java Threads programming, where Threads can dynamically create new objects in which they can synchronize with \textsf{notify} and \textsf{notifyAll} messages.

%% file: annex.tex
\section{Appendix}

\subsection{Proof of Copypaste Lemma}\label{sec:app-copypaste}
The first auxiliary lemma needed for the proof is formalized as follows.

\begin{lemma}\label{lem:P0}
	Given an execution $\mconf_0 \mtransup{t_1}\mconf_1\mtransup{t_2}\dots \mtransup{t_k} \mconf_k$
	and another initial configuration $\mconfInit$, we can build an execution $\widehat{\mconf_0}\mtransup{t_1}\widehat{\mconf_1}\mtransup{t_2}\dots \mtransup{t_k} 
	\widehat{\mconf_k}$ 
	with $\widehat{\mconf_i}=\mconf_i+\mconfInit$ for all $0\leq i\leq k$.
\end{lemma}
\begin{proof}
	Formally, we can build the execution described in Lemma \ref{lem:P0} by induction on $0\leq i\leq k$. 
	For $i=0$, let $\widehat{\mconf_0}=\mconf_0+\mconfInit$. Hence as $\mconf_0 \in \mconfsInit$ and $\mconfInit\in \mconfsInit$, it follows that $\widehat{\mconf_0}\in \mconfsInit$.
	Now assume that $\widehat{\mconf_0}\mtransup{}^\ast\widehat{\mconf_i}$ with $\widehat{\mconf_i}=\mconf_i+\mconfInit$ for some $i \in [0, k-1]$.
	We let $\constantM =\mconfInit(\qinit)$.
	\begin{itemize}
		
		\item If $t_{i+1}=(p_1,!!m,p_2)$, let $\mconf_i=\mset{p_1, q_1, \dots , q_{N'-1}}$ and $\mconf_{i+1}=\mset{p_2, q'_1,\dots, q'_{N'-1}}$.
		For all $1\leq i\leq N'-1$, either $m\notin \recfrom{q_i}$,  and $q'_i=q_i$, or $(q_i,?m, q'_i)\in T$. By induction hypothesis, $\widehat{\mconf_0}\mtransup{}^\ast\widehat{\mconf_i}$ with 
		$\widehat{\mconf}_i=\mset{p_1, q_1, \dots , q_{N'-1}, \constantM \cdot \qinit}$. Since $\qinit$ is an action state, we know that $m\notin \recfrom{\qinit}$, hence, with
		$\widehat{\mconf}_{i+1}=\mset{p_2, p'_1,\dots, p'_{N'-1}, \constantM \cdot \qinit}$, it holds that $\widehat{\mconf}_i \mtransup{t_{i+1}}\widehat{\mconf}_{i+1}$
		and $\widehat{\mconf}_{i+1} =\mconf_{i+1} + \mconfInit$.
		
		\item If $t_{i+1}=(p_1,!m, p_2)$ and there exist $q,q'\in Q$ such that $(q,?m,q')\in T$ and $(\mconf_i - \mset{p_1})(q)>0$ (observe that in Wait-Only protocols, as ${Q}_A \cap {Q}_W = \emptyset$, $p_1 \neq q$ and so $(\mconf_i - \mset{p_1})(q)>0$ if and only if $\mconf_i (q)>0$), then $\mconf_{i+1}= \mconf_i - \mset{p_1, q} + \mset{p_2, q'}$. 
		By induction hypothesis, $\widehat{\mconf_0}\mtransup{}^\ast\widehat{\mconf_i}$ with 
		$\widehat{\mconf}_i=\mconf_i + \mconfInit$.
		Hence, $\widehat{\mconf}_i (p_1) >0$; $\widehat{\mconf}_i (q) >0$ and so
		$\widehat{\mconf_0}\mtransup{}^\ast\widehat{\mconf}_i  \mtransUp{t_{i+1}} \widehat{\mconf}_i - \mset{p_1, q} + \mset{p_2, q'}$
		and $ \widehat{\mconf}_i - \mset{p_1, q} + \mset{p_2, q'} = \mconf_i + \mconfInit - \mset{p_1, q} + \mset{p_2, q'} = \mconf_{i+1} + \mconfInit$.
		
		\item If $t_{i+1}=(p_1,!m, p_2)$ and for all $q,q'\in Q$ such that $(q,?m,q')\in T$, $\mconf_i(q)=0$, then $\mconf_{i+1} = \mconf_i - \mset{p_1} + \mset{p_2}$. By induction hypothesis, $\widehat{\mconf_0}\mtransup{}^\ast\widehat{\mconf_i}$ with 
		$\widehat{\mconf}_i=\mconf_i + \mconfInit$.
		Observe that for all $q,q'\in Q$ such that $(q,?m,q')\in T$, 
		$\widehat{\mconf}_i(q)=\mconf_i(q)=0$. Indeed, since $\qinit$ is an action state, $\qinit\neq q$. Then $\widehat{\mconf}_i \mtransUp{t_{i+1}} \widehat{\mconf}_i - \mset{p_1} + \mset{p_2}$ and $\widehat{\mconf}_{i+1} = \widehat{\mconf}_i - \mset{p_1} + \mset{p_2} = \mconf_i + \mconfInit - \mset{p_1} + \mset{p_2} = \mconf_{i+1} + \mconfInit$.
	\end{itemize}
\end{proof}

Recall that the second auxiliary lemma formalizes the fact that from an execution $\mconf  \mtrans^\ast \mconf'$, if one considers a configuration $\mconf''$, then from the configuration $\mconf + \mconf''$, one can execute the same sequence of events happening in the execution $\mconf  \mtrans^\ast \mconf'$ and reach a configuration $\mconf'''$ such that:
\begin{itemize}
	\item $\mconf'''(q) \geq \mconf'(q)$ for all $q \in \waitingset{Q}$;
	\item $\mconf'''(q) \geq \mconf'(q) + \mconf''(q)$ for all $q \in {Q}_A$.
\end{itemize}

We generalize this statement to any configuration $\mconf + N \cdot \mconf''$ in the lemma below.

\begin{lemma}\label{lem:P1}
	Given an execution $\mconf_0 \mtransup{t_1}\mconf_1\mtransup{t_2}\dots \mtransup{t_k} \mconf_k$, and an integer $\constantM \geq 1$,
	for all configurations (not necessarily initial) $\widetilde \mconf_0$ such that $\widetilde{\mconf}_0(\qinit)\geq \constantM \cdot \mconf_0(\qinit)$, we can build a run
	$\widetilde \mconf_0\mtransup{t_1}\dots\mtransup{t_k}\widetilde{\mconf}_k$ in which:
	\begin{itemize}
		\item $\widetilde \mconf_i(q)\geq \widetilde \mconf_0(q)+\mconf_i(q)$ for all $q \in {Q}_A \setminus \set{\qinit}$,
		\item $\widetilde \mconf_i(q)\geq \mconf_i(q)$ for all $q\in\waitingset{Q}$, and 
		\item $\widetilde \mconf_i(\qinit)\geq (\constantM-1) \cdot \mconf_0(\qinit)+\mconf_i(\qinit)$
	\end{itemize}
	for all $0\leq i \leq k$.
	
\end{lemma}
\begin{proof}
	Again, we can build the run by induction on $0\leq i\leq k$. 
	%%%
	For $i=0$, it is obvious since $\mconf_0(q)=0$ for all $q\in Q\setminus\set{\qinit}$. Let now $0\leq i<k$ and assume that we have built a run $\widetilde \mconf_0\mtrans^\ast \widetilde{\mconf}_i$ such that:
	\begin{align*}
		&\widetilde \mconf_i(q)\geq \widetilde \mconf_0(q)+\mconf_i(q) \textrm{   for all }q\in{Q}_A\setminus\set{\qinit}\\
		&\widetilde \mconf_i(q)\geq \mconf_i(q)  \textrm{   for all }q\in\waitingset{Q}\\
		&\widetilde \mconf_i(\qinit)\geq (\constantM-1) \cdot \mconf_0(\qinit)+\mconf_i(\qinit) 
	\end{align*}
	
	Let $N_1=\mconf_0(\qinit)$. 	
	\begin{itemize}

		\item If $t_{i+1}=(p_1,!!a,p_2)$, then by induction hypothesis (and because $p_1 \in Q_A$), 
		\begin{itemize}
			\item $\widetilde{\mconf_i}(p_1)\geq \widetilde{\mconf_0}(p_1)+\mconf_i(p_1)$ if $p_1\neq\qinit$, and 
			\item $\widetilde{\mconf_i}(p_1)\geq (\constantM  -1)\cdot N_1+\mconf_i(p_1)$ if $p_1=\qinit$.
		\end{itemize}
		Moreover, $\mconf_i(p_1)>0$, hence the transition $t_{i+1}$ can be taken from $\widetilde{\mconf_i}$. 
		
		Let $\mconf_i=\mset{p_1,q_1,\dots, q_{N_1-1}}$, then, from the "broadcast step" definition,
		$\mconf_{i+1}=\mset{p_2,q'_1,\dots, q'_{N_1-1}}$ such that for all $1\leq i\leq N_1$, either $a\notin \recfrom{q_i}$ and $q'_i=q_i$, or $(q_i,?a, q'_i)\in T$. 
		Also, $\widetilde{\mconf}_i=\mset{p_1, q_1, \dots, q_{N_1-1}, q''_1,\dots, q''_K}$
		and we now define 
		$\widetilde{\mconf}_{i+1}=\mset{p_2, q'_1, \dots, q'_{N_1-1}, p''_1,\dots, p''_K}$,
		with, for all $1\leq i\leq K$: if $a\notin \recfrom{q''_i}$ then $p''_i=q''_i$, and otherwise, let $p''_i \in Q$ such that $(q''_i,?a, p''_i)\in T$. By definition of a "broadcast step", we get that $\widetilde{\mconf}_i \mtransup{t_{i+1}} \widetilde{\mconf}_{i+1}$. 
		
		It naturally follows that for all $q\in Q$, 
		$\widetilde{\mconf}_{i+1}(q)\geq
		\mconf_{i+1}(q)$. 
		Let now $q\in {Q}_A$. Either $\widetilde{\mconf}_{i+1}(q)=\widetilde{\mconf}_i(q)+\ell$ for some $\ell\geq 0$, or 
		$\widetilde{\mconf}_{i+1}(q)=\widetilde{\mconf}_i(q)-1$ 
		(and $q=p_1$), \ie\  $\widetilde{\mconf}_{i+1}(q)=\widetilde{\mconf}_i(q)+\ell$ for some $\ell\geq -1$
		%		\begin{enumerate}
			%			\item
			%			 In the first 
			%			case, 
			Hence, 
			$\mconf_{i+1}(q)= \mconf_i(q)+\ell'$, with $\ell\geq \ell'\geq -1$. Indeed, let $\set{q^1,\dots,q^r}$ be the set of states such that $\widetilde{\mconf}_i(q^j)>0$ and 
			$(q^j, ?a, q) \in T$ with $q$ being the chosen state of the corresponding component when defining $\widetilde{M}_{i+1}$, for all $1\leq j\leq r$. Then $\mconf_i(q^j)\leq \widetilde{\mconf}_i(q^j)$ by induction
			hypothesis, so $\mconf_{i+1}(q)=\mconf_i(q)+\ell'$ and  $\widetilde{\mconf}_{i+1}(q)=\widetilde{\mconf}_i(q)+\ell$ with $\ell\geq \ell'\geq -1$.
			%(observe that if $\ell'=-1$, it means that $q=p_1$
			%and that no process has gone to $p_1$ when receiving $a$ between $\mconf_i$ and $\mconf_{i+1}$).
			Hence, if $q \neq\qinit$, by induction hypothesis, 
			\begin{align*}
				\widetilde{\mconf}_{i+1}(q)& \geq \widetilde{\mconf}_0(q)+\mconf_i(q)+\ell \\ & =\widetilde{\mconf}_0(q)+\mconf_{i+1}(q)-\ell'+\ell \\ & \geq \widetilde{\mconf}_0(q)+\mconf_{i+1}(q).
			\end{align*}
			If otherwise $q=\qinit$, by induction hypothesis, \begin{align*} \widetilde{\mconf}_{i+1}(\qinit) & =\widetilde{\mconf}_i(\qinit)+\ell \\ & \geq (\constantM-1) \cdot N_1+\mconf_i(\qinit)+\ell \\& = (\constantM-1)\cdot N_1+\mconf_{i+1}(\qinit)-\ell'+\ell \\ & \geq
				(\constantM-1)\cdot N_1+\mconf_{i+1}(\qinit).\end{align*}
			%		\end{enumerate}

		%		In the second case, we have that $\widetilde{\mconf}_{i+1}(p_1) = \widetilde{\mconf}_i(p_1)-1$, so no process receiving message $a$ goes to $p_1$. We prove that it implies that $\mconf_{i+1}(p_1)=\mconf_i(p_1)-1$:
		%%		because for all $q\in \waitingset{Q}$ such that $(q,?a,p_1)\in T$, $\_{m,i}(q)=0$,
		%%		and $\mconf''_{m,i}(q)\geq \mconf_i(q)$, so no process in $\mconf_i$ can receive message $a$ and go to $p_1$ neither. 
		%		otherwise there exists $q \in Q_W$ such that $(q,?a,p_1)\in T$ and $\mconf_{i} = \mset{p_1, q_1, \dots, q}$ and $\mconf_{i+1} = \mset{p_2, q'_1, \dots, p_1}$, and so, by definition of $\widetilde{\mconf}_{i+1}$: $\widetilde{\mconf}_{i+1} = \mset{p_2, q'_1, \dots, p_1, p''_1, \dots, p''_k}$. Hence, $\widetilde{\mconf}_{i+1}(p_1) = 1 + (\mconf_{i+1} - \mset{p_1})(p_1) + \mset{p''_1, \dots, p''_k}(p_1)$. As $p_1 \in Q_A$, we get, $\widetilde{\mconf}_{i+1}(p_1) \geq \mconf_{i}(p_1)  + \mset{q''_1, \dots, q''_k}(p_1) \geq  \widetilde{\mconf}_i(p_1)$. Hence, $\mconf_{i+1}(p_1)=\mconf_i(p_1)-1$.
		%		
		%		So, if $p_1\neq \qinit$, by induction hypothesis, 
		%		$\widetilde{\mconf}_{i+1}(p_1)\geq \widetilde{\mconf}_0(p_1)+\mconf_i(p_1)-1=\widetilde{\mconf}_0+\mconf_{i+1}(p_1)$, and if $p_1=\qinit$, $\widetilde{\mconf}_{i+1}(\qinit)\geq (\constantM-1)\cdot N_1+\mconf_i(\qinit)-1=(\constantM-1)\cdot N_1+\mconf_{i+1}(\qinit)$ and we are done.
		
		\item Let $t_{i+1}=(p_1, !a, p_2)$ and there exist $p,p'\in Q$ such that $(p,?a,p')\in T$ and $\mconf_i(p)>0$.
		Then, $\mconf_i(p_1)>0$ and
		$\mconf_{i+1}=\mconf_i-\mset{p_1,p}+\mset{p_2,p'}$. 
		By induction hypothesis, 
		\begin{itemize}
			\item $\widetilde{\mconf}_i(p_1)\geq \widetilde{\mconf}_0(p_1)+\mconf_i(p_1)  \text{ if }p_1 \neq \qinit$
			\item $\widetilde{\mconf}_i(p_1)\geq (\constantM -1) \cdot N_1+\mconf_i(p_1)  \text{ if }p_1=\qinit$
			\item $\widetilde{\mconf}_i(p)\geq \mconf_i(p)$.
		\end{itemize}
		Hence 
		$\widetilde{\mconf}_i\mtransup{t_{i+1}} \widetilde{\mconf}_{i+1}$ where $\widetilde{\mconf}_{i+1}=\widetilde{\mconf}_i-\mset{p_1,p}+\mset{p_2,p'}$. 
		Observe that for all $q \in Q$, $\widetilde{\mconf}_{i+1}(q) - \widetilde{\mconf}_{i}(q) = {\mconf}_{i+1}(q) - {\mconf}_{i}(q)$.
		Hence, if we let $q \in Q$, by induction hypothesis we get that:
		\begin{align*}
			\widetilde{\mconf}_{i+1}(q)  &= \widetilde{\mconf}_i(q) - \mconf_i(q) + \mconf_{i+1}(q) \\
			&\geq \widetilde{\mconf}_0(q) + \mconf_{i+1}(q) &&\textrm{ if }q \in {Q}_A \setminus \set{\qinit}\\
			&\geq (\constantM - 1) \cdot N_1 + \mconf_{i+1}(q) &&\textrm{ if }q ={\qinit}\\
			&\geq \mconf_{i+1}(q) &&\textrm{ if }q\in \waitingset{Q}.
		\end{align*}

		\item If $t_{i+1}=(p_1, !a, p_2)$ and for all $p,p'\in Q$ such that $(p,?a,p')\in T$, we have that $\mconf'_i(p)=0$, then, either $\widetilde{\mconf}_i(p)=0$ for all $p,p'\in Q$
		such that $(p,?a,p')\in T$, or there exist some $p,p'\in Q$ such that $(p,?m,p')\in T$ and $\widetilde{\mconf}_i(p)>0$. In the first case, since 
		$\widetilde{\mconf}_i(p_1)\geq 0$ by induction hypothesis, $\widetilde{\mconf}_{i+1} =\widetilde{\mconf}_i -\mset{p_1}+\mset{p_2}$, and $\mconf_{i+1}=\mconf_i-\mset{p_1}+\mset{p_2}$. Then, as in the previous case 
		$\widetilde{\mconf}_{i+1} - \widetilde{\mconf}_{i} = {\mconf}_{i+1} - {\mconf}_{i}$, which allows us to conclude.

		In the second case, 
		$\widetilde{\mconf}_{i+1}=\widetilde{\mconf}_i-\mset{p_1,p}+\mset{p_2,p'}$. Then the only states $q$ for which 
		$\widetilde{\mconf}_{i+1}(q) - \widetilde{\mconf}_{i}(q) \neq {\mconf}_{i+1}(q) - {\mconf}_{i}(q)$ are states $p'$ and $p$. Hence, we only focus on those two states as for other states we can conclude using the reasoning of the previous case. Observe that the only interesting case is $p' \neq p$ as otherwise, we can again use the previous reasoning. Hence, consider $p \neq p'$. By construction, $p \in \waitingset{Q}$, $\widetilde{\mconf}_{i}(p) >0$ and $\mconf_i(p) = 0$. 
		If $p_2 = p$, we get that $\mconf_{i+1}(p) = 1$ and $\widetilde{\mconf}_{i+1}(p) = \widetilde{\mconf}_{i}(p) + 1 - 1 = \widetilde{\mconf}_{i}(p)  \geq 1 = \mconf_{i+1}(p)$ which concludes this case.
		Otherwise, $\mconf_{i+1}(p) = 0$ and $\widetilde{\mconf}_{i+1}(p) = \widetilde{\mconf}_{i}(p)  - 1 \geq 0 = \mconf_{i+1}(p)$ which concludes this case.
		
		Consider now $p'$. Observe that $\widetilde{\mconf}_{i+1}(p')  - \widetilde{\mconf}_{i}(p') > {\mconf}_{i+1}(p')  - {\mconf}_{i}(p')$, hence $\widetilde{\mconf}_{i+1}(p') > \widetilde{\mconf}_{i}(p') + {\mconf}_{i+1}(p')  - {\mconf}_{i}(p')$, we conclude with the reasoning of the previous item.

	\end{itemize}
\end{proof}

We are now ready to prove Lemma \ref{lemma:copycat-action-state}.
%\lulutex{Arnaud dit pourquoi ne pas mettre la preuve apres le lemme ?}

\begin{proof}	
	Let $\PP=(Q, \Sigma, \qinit, T)$ be a Wait-Only protocol, $A = \set{q_1, \dots, q_n}\subseteq {Q}_A$ a subset of coverable \emph{action} states and $p\in {Q_W}$ a coverable \emph{waiting} state. 
	Let $N \in \nat$.
	Using Lemmas \ref{lem:P0} and \ref{lem:P1}, we can now prove the lemma. 
	
	We start by proving that there exists an execution $\mconfInit \mtrans^\ast \mconf$ such that for all $q\in A$, $\mconf(q) \geq N$ and $|\mconfInit| = N\cdot \sum_{i=1}^{n} \mathsf{min}_{q_i}$.
	
	We prove it by induction on the size of $A$. If $A = \emptyset$, the property is trivially true.
	
	Let $n\in\nat$, and assume the property to hold for all subsets $A \subseteq Q_A$ of size $n$. Take $A = \set{q_1,q_2, \dots, q_{n+1}}\subseteq {Q_A}$ of size $n+1$ such that all states $q \in A$ are coverable and let $A'=A\setminus\set{q_1}$. 
	If $\qinit \in A$, w.l.o.g. we assume that $q_1 = \qinit$, and hence $\qinit \nin A'$.
	Consider the execution
	\[
	\mconf_0 \mtransup{t_1}\mconf_1\mtransup{t_2}\dots... \mtransup{t_k} \mconf_k \mbox{ \ with  } \mconf_k(q_1)>0 \mbox{  and  } |\mconf_0|=\mathsf{min}_{q_1}
	\]
	
	and the execution
	\[		\mconf'_0 \mtransup{t'_1}\mconf'_1\mtransup{t'_2}\dots \mtransup{t'_{\ell}} \mconf'_\ell \mbox{\  s.t.  } \mconf'_\ell(q') \geq N \mbox{ for all } q'\in A', \mbox{ and } |\mconf'_0| = N \cdot \sum_{i = 2}^{n+1} \mathsf{min}_{q_i}
	\] 
	(it exists by induction hypothesis). 
	
	%If $N_1=||\mconf_0||$, % and $N'=||\mconf'_0||$, 
	We let $\mconf_0^N=\mset{(N \cdot \mathsf{min}_{q_1})\cdot\qinit}$ and $\mconf''_0=\mconf'_0+\mconf_0^N$. Thanks to Lemma \ref{lem:P0}, we can build an execution
	\[
	\mconf''_0\mtransup{t'_1}\mconf''_1\mtransup{t'_2}\dots\mtransup{t'_\ell}\mconf''_\ell \mbox{ \  with } \mconf''_\ell =\mconf'_\ell + \mconf_0^N
	\] 
	
	In particular,
	for all $q'\in A'$, 
	\[
	\mconf''_\ell(q')=\mconf'_\ell(q')\geq N \mbox{ \  and } |\mconf''_0| = |\mconf'_0| + |\mconf_0^N| = N\cdot \sum_{i=2}^{n+1}\mathsf{min}_{q_i} + N\cdot \mathsf{min}_{q_1}.
	\]

	Now that we have shown how to build an execution that leads to a configuration with more than $N$ processes on all states in $A'$ and enough
	processes in the initial state, we show that mimicking $N$ times the execution allowing to cover $q_1$ allows to obtain the desired result. 
	Observe that if $q_1 = \qinit$, there is nothing to do, as $\mathsf{min}_{q_1} = 1$. We assume now that $q_1 \neq \qinit$. 
	Let $\mconf_{0,1}=\mconf''_\ell$. We know that for all $q'\in A'$, $\mconf_{0,1}(q')\geq N$, and $\mconf_{0,1}(\qinit)\geq N\cdot \mathsf{min}_{q_1}$. 
	Since $|\mconf_0|=\mathsf{min}_{q_1}$, using
	Lemma \ref{lem:P1}, we can build the execution 
	\begin{align*}
		\mconf_{0,1}\mtransup{t_1}\dots\mtransup{t_k} \mconf_{k,1} \mbox{ \  with } &\mconf_{k,1}(\qinit)\geq (N-1)\cdot \mathsf{min}_{q_1}, \\
		&\mconf_{k,1}(q')\geq \mconf_{0,k}(q')+\mconf_k(q')\geq N \mbox{ \  for all  } q'\in A' \mbox{ and }\\
		&\mconf_{k,1}(q_1)\geq \mconf_{0,k}(q_1)+\mconf_k(q_1)\geq 1
	\end{align*}
	
	Iterating this construction and
	applying each time Lemma \ref{lem:P1}, we obtain that there is an execution 
	\begin{align*}
		\mconf_{0,1}\mtransup{t_1}\dots\mtransup{t_k}  \mconf_{k,1}\mtransup{t_1}\dots\mtransup{t_k}
		&\mconf_{k,2}\dots\mtransup{t_1}\dots\mtransup{t_k} \mconf_{k,N-1}\mtransup{t_1}\dots\mtransup{t_k} \mconf_{k,N} \\
		\mbox{with for all } 1 \leq i \leq N\mbox{: \ }& \mconf_{k,i}(\qinit)\geq (N-i)\cdot \mathsf{min}_{q_1}\\
		& \mconf_{k,i}(q')\geq N \mbox{ for all } q'\in A' \mbox{, and }\\
		& \mconf_{k,i}(q_1)\geq \mconf_{k,i-1}(q_1)+1\geq i 
	\end{align*} 
	
	Observe that to obtain that $\mconf_{k,i}(q_1)\geq i$ from Lemma \ref{lem:P1}, we use the fact that 
	$q_1\in {Q_A}$. Hence, $\mconf_{k,N}(q_1)\geq N$ and $\mconf_{k,N}(q')\geq N$ for all $q'\in A'$ and we have built an execution where
	$\mconf_{k,N}(q)\geq N$ for all $q\in A$ and $|\mconf_{k,N}| = |\mconf''_0| = N \cdot \sum_{i=1}^{|A|} \mathsf{min}_{q_i}$, as expected. \\
	
	At last, 
	consider $p\in\waitingset{Q}$ the coverable state and $\mconf'_0\mtrans^\ast \mconf'_k$ such that $\mconf'_k(p)\geq 1$ and $|\mconf'_0|=\mathsf{min}_p$. 
	Let $\mconf_0\mtrans^\ast \mconf_m$ be an execution
	such that $\mconf_m(q)\geq N$ for all $q\in A$ and $|\mconf_0| = N\cdot \sum_{i=1}^{n} \mathsf{min}_{q_i}$, as we have built before. 
	%We rename it for readability's sake.
	By Lemma \ref{lem:P0}, we let $\widehat{\mconf}_0=\mconf_0+\mconf'_0$ and we have an execution $\widehat{\mconf}_0\mtrans^\ast \widehat{\mconf}_m$ with 
	$\widehat{\mconf}_m = \mconf_m +\mconf'_0$.
	Hence, $\widehat{\mconf}_m(q)\geq N$
	for all $q\in A$ and $\widehat{\mconf}_m(\qinit)\geq \mconf'_0(\qinit)$, and note that $|\widehat{\mconf}_m| = |\widehat{\mconf}_0| = |\mconf_0| + |\mconf'_0| = N\cdot \sum_{i=1}^n \mathsf{min}_{q_i} + \mathsf{min}_p$. Then, with $\widetilde \mconf_0=\widehat{\mconf}_m$, by Lemma \ref{lem:P1}, we have a run
	$\widetilde \mconf_0\mtrans^\ast\widetilde{\mconf}_k$ with $\widetilde{\mconf}_k(q)\geq \widetilde \mconf_0(q)+\mconf'_k(q)\geq \widetilde \mconf_0(q)\geq N$ for all $q\in A$,
	and $\widetilde{\mconf}_k(p)\geq \mconf'_k(p)\geq 1$, and $|\widetilde{\mconf}_0| = |\widehat{\mconf}_0| =  N\cdot \sum_{i=1}^n \mathsf{min}_{q_i} + \mathsf{min}_p$.
	
\end{proof}

\subsection{Proof of Claim~\ref{claim:aux:soundness:wo:rdv}}~\label{sec:app-rdv}

\begin{proof}
We will first show that for all $n \in \nat$, for all $q \in S'$ there exists a configuration $\mconf_q \in \Interp{\gamma}$ and a configuration $\mconf_q' \in \mconfs$ such that $\mconf_q \mtrans^\ast \mconf_q'$ and $\mconf'_q(q) \geq n$. This will allow us to rely then on Lemma \ref{lem:consistent-reach} to conclude. 

Take $n \in \nat$ and $q \in S'$, if $q \in S$, then take $\mconf_q \in \Interp{\gamma}$ to be $\mset{n \cdot q}$. Clearly $\mconf_q \in \Interp{F(\gamma)}$, $\mconf_q(q) \geq n$ and $\mconf_q \mtrans^\ast \mconf_q$. Now let $q \in S' \setminus S$. Note $(\Toks'', S'')$ the intermediate sets of $F(\gamma$)'s computation.\\

\textbf{Case 1:} $q \in S''$.  As a consequence, $q$ was added to $S''$ by one of the conditions  \ref{ccover-wo-F-cond-send-S}, \ref{ccover-wo-F-cond-reception-S}~or \ref{ccover-wo-F-cond-tok-end}. 

\begin{itemize}
	\item Case \ref{ccover-wo-F-cond-send-S}~and $a \notin \Rec{q}$. Denote $q'$ the state such that $(q', !a, q)$, and consider the configuration $\mconf_q = \mset{n \cdot q'}$. By doing $n$ sendings, we reach $\mconf'_q= \mset{n \cdot q}$. Note that messages are not received as $q' \in Q_A$ and $a \notin \Rec{q}$. It holds that $\mconf'_q \in \Interp{F(\gamma)}$.
	
	\item Case \ref{ccover-wo-F-cond-send-S} and $a\in \Rec{q}$~or case \ref{ccover-wo-F-cond-reception-S}. Note $(q_1, !a, q_1')$ and $(q_2, ?a, q_2')$ the two transitions realizing the conditions. As a consequence $q_1, q_2 \in S$. Take the configuration $\mconf_q =\mset{n \cdot q_1, n \cdot q_2}$. $\mconf_q \in \Interp{\gamma}$ and by doing $n$ successive rendez-vous on the letter $a$, we reach configuration $\mconf'_q = \mset{n\cdot q'_1, n \cdot q'_2}$. Hence, $\mconf'_q \in \Interp{F(\gamma)}$, and as $q \in \{q'_1, q'_2\}$, $\mconf'_q(q) \geq n$.
	
	\item In case \ref{ccover-wo-F-cond-tok-end}, there exists $(q', m) \in \Toks$ such that $(q', ?a, q) \in T$,  $m \notin \Rec{q}$, and there exists $p \in S$ such that $(p, !a,p') \in T$. Remember that $\gamma$ is consistent, and so there exists a finite sequence of transitions $(q_0, !m, q_1) (q_1, ?m_1, q_2) \dots (q_k, ?m_k, q')$ such that $q_0 \in S$ and there exists $(q'_i , !m_i, q''_i) \in T$ with $q'_i \in S$ for all $1 \leq i \leq k$.
	Consider 
	\[
	\mconf_q = \mset{(n-1) \cdot q_0, (n-1) \cdot q'_1 ,  \dots , (n-1) \cdot q'_k, n \cdot p, q'}.
	\] 
	Clearly $\mconf_q \in \Interp{\gamma}$ as all states except $q'$ are in $S$ and $q' \in \mst(\Toks)$ with $\mconf_q(q') = 1$. We shall show how to put 2 processes on $q$ from $\mconf_q$ and then explain how to repeat the steps in order to put $n$. Consider the following run: 
	\[
	\mconf_q \mtransrdv{(p, !a, p')} \mconf_1 \mtransup{(q_0, !m, q_1)} \mconf_2 \mtransrdv{(q'_1, !m_1, q''_1)} \dots \mtransrdv{(q'_k, !m_k, q''_k)} \mconf_{k+2} \mtransrdv{(p, !a, p')} \mconf_{k+3}.
	\]
	The first rendez-vous on $a$ is made with transitions $(p, !a, p')$ and $(q', ?a, q)$. 
	Then either $m \notin \Rec{p'}$ and $\mconf_1 \mtransse{(q_0, !m, q_1)} \mconf_2$, otherwise $\mconf_1 \mtransrdv{(q_0, !m, q_1)} \mconf_2$. In any case, the rendez-vous or sending is made with transition $(q_0, !m, q_1)$ and  the message is not received by the process on $q$ (because $m \notin \Rec{q}$) and so $\mconf_2 \geq \mset{q, q_1}$. Then, each rendez-vous on $m_i$ is made with transitions $(q'_i, !m_i,q''_i)$ and $(q_i, ?m_i, q_{i+1})$ ($q_{k+1} = q'$), and the last rendez-vous with transition $(q', ?a, q)$.
	Hence 
	\[
	\mconf_{k+3} \geq \mset{(n-2)\cdot q_0, (n-2) \cdot q'_1 , \dots , (n-2) \cdot q'_k, (n-2) \cdot p, 2 \cdot q}.
	\]
	We can reiterate this run (without the first rendez-vous on $a$) $n-2$ times to reach a configuration $\mconf'_q$ such that $\mconf'_q \geq \mset{n \cdot q}$.
\end{itemize}

\textbf{Case 2:} $q \notin S''$. Hence, $q$ should be added to $S'$ by one of the conditions \ref{ccover-wo-F-cond-2toks-1}, \ref{ccover-wo-F-cond-3toks-1}, and \ref{ccover-wo-F-cond-3toks-2}.
\begin{itemize}
	\item If it was added with condition \ref{ccover-wo-F-cond-2toks-1}, let $(q_1, m_1), (q_2, m_2) \in \Toks''$ such that $q =q_1$, $m_1 \ne m_2$, $m_2 \notin \Rec{q_1}$ and $m_1 \in \Rec{q_2}$.
	From the proof of Lemma \ref{lem:F-consistent}, one can actually observe that all tokens in $\Toks''$ correspond to ``feasible'' paths regarding states in $S$, i.e there exists a finite sequence of transitions $(p_0, !m_1, p_1) (p_1, ?b_1, p_2) \dots (p_k, ?b_k, q_1)$ such that $p_0 \in S$ and there exists $(p'_i , !b_i, p''_i) \in T$ with $p'_i \in S$ for all $1 \leq i \leq k$. The same such sequence exists for the token $(q_2, m_2)$, we note the sequence $(s_0, !m_2, s_1) (s_1, ?c_1, s_2) \dots (s_\ell, ?c_\ell, q_2)$ such that $s_0 \in S$ so there exists $(s'_i , !c_i, s''_i) \in T$ with $s'_i \in S$ for all $1 \leq i \leq \ell$. 
	Consider
	\[
	\mconf_q = \mset{n \cdot p_0, n \cdot s_0, n \cdot  p'_1, \dots ,n \cdot p'_k, n \cdot s'_1 ,\dots , n \cdot s'_\ell}.
	\] 
	 Clearly, $\mconf_q \in \Interp{\gamma}$, as all states are in $S$. Consider the following run: 
	 \[
	 \mconf_q \mtransse{(p_0, !{m_1}, p_1)} \mconf_1 \mtransrdv{(p'_1, !b_1, p''_1)} \dots \mtransrdv{(p'_k, !b_k, p''_k)} \mconf_{k+1}
	 \]
	 Each rendez-vous on letter $b_i$ is made with transitions $(p'_i, !b_i, p_i'')$ and $(p_i, ?b_i, p_{i+1})$ ($p_{k+1} = q_1$). Hence, $\mconf_{k+1}$ is such that 
	 \[
	 \mconf_{k+1} \geq \mset{(n-1) \cdot p_0, n \cdot s_0, (n-1) \cdot  p'_1, \dots ,(n-1) \cdot p'_k, n \cdot s'_1 ,\dots , n \cdot s'_\ell, q_1}
	 \]
	 
	  From $\mconf_{k+1}$, consider the following run: 
	  \[
	  \mconf_{k+1} \mtransup{(s_0, !m_2, s_1)} \mconf_{k+2} \mtransrdv{(s'_1, !c_1, s_1'')} \dots \mtransup{(s'_\ell, !c_\ell, s_\ell'')} \mconf_{k+\ell +2} \mtransup{(p_0, !m_1, p_1)}\mconf_{k+\ell +3}
	  \]
	  If no process is on a state in $\staterec{m_2}$ (we use $\staterec{m_2} := \set{q \mid \exists q' \in Q \text{ s.t. } (q, ?m_2, q') \in T}$) then $\mconf_{k+1} \mtransse{(s_0, !m_2, s_1)} \mconf_{k+2}$, otherwise $\mconf_{k+1} \mtransrdv{(s_0, !m_2, s_1)} \mconf_{k+2}$.
	  In any case, as $m_2 \notin \Rec{q_1}$, hence $\mconf_{k+2} \geq \mset{q_1}$.
	  	Moreover each rendez-vous on letter $c_i$ is made with transitions $(s'_i, !c_i, s_i'')$ and $(s_i, ?c_i, s_{i+1})$ ($s_{k+1} = q_2$), the last rendez-vous on $m_1$ is made with transitions $(p_0, !m_1, p_1)$ and $(q_2, ?m_1, q_2')$ (such a $q_2'$ exists as $m_1 \in \Rec{q_2}$). Hence, $\mconf_{k+\ell +3} \geq \mset{p_1, q_1}$.
	  \begin{align*}
	  \mconf_{k+2} \geq  \mset{ & (n-1) \cdot p_0, (n-1) \cdot s_0, (n-1) \cdot  p'_1, \dots ,(n-1) \cdot p'_k, \\
	  	& (n-1) \cdot s'_1 ,\dots , (n-1) \cdot s'_\ell, q_1, p_1}
	  \end{align*}
	 By repeating the two sequences of steps (without the first sending of $m_1$) $n-1$ times (except for the last time where we don't need to repeat the second run), we reach a configuration $\mconf'_q$ such that $\mconf'_q\geq \mset{n \cdot q_1}$.
	  
	  \item If it was added with condition \ref{ccover-wo-F-cond-3toks-1}, then let $(q_1, m_1), (q_2,m_2), (q_3,m_2) \in \Toks''$ such that $m_1 \ne m_2$ and $(q_2, ?m_1, q_3) \in T$ with $q =q_1$. From the proof of Lemma \ref{lem:F-consistent}, $\Toks''$ is made of ``feasible'' paths regarding $S$ and so there exists a finite sequence of transitions $(p_0, !m_2, p_1) (p_1, ?b_1, p_2) \dots$ $(p_k, ?b_k, q_2)$ such that $p_0 \in S$ and there exists $(p'_i , !b_i, p''_i) \in T$ with $p'_i \in S$  for all $1 \leq i \leq k$.
	  The same sequence exists for the token $(q_1, m_1)$, we write the sequence $(s_0, !m_1, s_1) (s_1, ?c_1, s_2)\dots$ $(s_\ell, ?c_\ell, q_1)$ such that $s_0 \in S$ and there exists $(s'_i , !c_i, s''_i) \in T$ with $s'_i \in S$ for all $1 \leq i \leq \ell$. 
	  Consider
	  \[
	  \mconf_q = \mset{n \cdot p_0, n \cdot s_0, n \cdot p'_1 , \dots , n \cdot p'_k, n \cdot s'_1 , \dots , n \cdot s'_\ell}.
	  \] 
	  Clearly, $\mconf_q \in \Interp{\gamma}$, as all states are in $S$. We do the same run from $\mconf_q$ to $\mconf_{k+1}$ as in the previous case: 
	  \[
	  \mconf_q \mtransse{(p_0, !{m_2}, p_1)} \mconf_1 \mtransrdv{(p'_1, !b_1, p''_1)} \dots \mtransrdv{(p'_k, !b_k, p''_k)} \mconf_{k+1}.
	  \]
	  Here $\mconf_{k+1}$ is then such that:
	  \[
	   \mconf_{k+1} \geq \mset{(n-2) \cdot p_0, n \cdot s_0, (n-1) \cdot  p'_1, \dots ,(n-1) \cdot p'_k, n \cdot s'_1 ,\dots , n \cdot s'_\ell, q_2}
	  \]
	   Then, from $\mconf_{k+1}$ we do the following:
	  \[
	  \mconf_{k+1} \mtransup{(s_0, !m_1, s_1)} \mconf_{k+2} \mtransup{(s'_1, !c_1, s_1'')} \dots \mtransup{(s'_\ell, !c_\ell, s_\ell'')} \mconf_{k+\ell+2} \mtransup{(p_0, !m_2, p_1)} \mconf_{k+\ell+3}
	  \] 
	  The rendez-vous on letter $m_1$ is made with transitons $(s_0, !m_1, s_1)$ and $(q_2, ?m_1, q_3)$. Then, each rendez-vous on letter $c_i$ is made with transitions $(s'_i, !c_i, s_i'')$ and $(s_i, ?c_i, s_{i+1})$ ($s_{k+1} = q_1$), and the last rendez-vous on letter $m_2$ is made with transitions $(p_0, !m_2, p_1)$ and $(q_3, ?m_2,q_3')$ (such a state $q_3'$ exists as $(q_3, m_2) \in \Toks''$ and so $m_2\in \Rec{q_3}$). Hence, $\mconf_{k+\ell+3}$ is such that:
	  \begin{align*}
	  	\mconf_{k+\ell +3} \geq  \mset{ & (n-2) \cdot p_0, (n-1) \cdot s_0, (n-1) \cdot  p'_1, \dots ,(n-1) \cdot p'_k, \\
	  		& (n-1) \cdot s'_1 ,\dots , (n-1) \cdot s'_\ell, q_1, p_1}
	  \end{align*}
	  We can repeat the steps from $\mconf_q$ (except the first sending of $m_2$ from $p_0$), $n-1$ times (except for the last time where we don't need to repeat the second part of the run), to reach a configuration $\mconf'_q$ such that $\mconf'_q\geq \mset{n \cdot q_1}$.
	  
	\item 
	If it was added with condition \ref{ccover-wo-F-cond-3toks-2}, then let $(q_1, m_1), (q_2, m_2), (q_3, m_3) \in \Toks''$, such that $m_1\ne m_2$, $m_2\ne m_3$, $m_1 \ne m_3$, and $m_1 \notin \Rec{q_2}$, $m_1 \in \Rec{q_3}$, and $m_2 \notin \Rec{q_1}$, $m_2 \in \Rec{q_3}$ and $m_3 \in \Rec{q_2}$ and $m_3 \in \Rec{q_1}$, and $q_1 = q$. Then there exists three finite sequences of transitions:
	\begin{align*}
		& (p_0, !m_1, p_1) (p_1, ?b_1, p_2) \dots (p_k, ?b_k, p_{k+1}) & \mbox{ with }  p_{k+1} = q_1 \\
		&(s_0, !m_2, s_1) (s_1, ?c_1, s_2) \dots (s_\ell, ?c_\ell, s_{\ell +1}) & \mbox{ with }  s_{\ell +1} = q_2 \\
		& (r_0, !m_3, r_1) (r_1, ?d_1, r_2) \dots (r_j, ?d_j, r_{j+1}) & \mbox{ with } r_{j+1} = q_3
	\end{align*}
%	 such that $$, $$ and $$, 
We denote the \emph{multiset} of messages $ \mset{ b_{i_1}, c_{i_2}, d_{i_3} \mid {1 \leq i_1 \leq k, 1 \leq i_2 \leq \ell, 1 \leq i_3 \leq j} }$ by ${Mess}$.
For all messages $a \in Mess$, there exists $q_{a} \in S$ such that $(q_a, !a, q'_a)$. Consider
\[
\mconf_q = \mset{n \cdot p_0, n \cdot s_0, n \cdot r_0} + \sum_{a \in Mess}\mset{n \cdot q_{a}}.
\] 
From $\mconf_q$, consider the following run: 
\[
\mconf_q \mtransse{(p_0, !m_1, p_1)} \mconf_1 \mtransrdv{(q_{b_1}, !b_1, q'_{b_1})} \dots \mtransrdv{(q_{b_k}, !b_k, q'_{b_k})} \mconf_{k +1}.
\]
Each rendez-vous with letter $b_i$ is made with transitions $(q_{b_i}, !b_i, q'_{b_i})$ and $(p_i, ?b_i, p_{i+1})$. Hence,
\begin{align*}
\mconf_{k+1} \geq &  \mset{ q_1, (n-1) \cdot p_0, n \cdot s_0, n \cdot r_0} + \sum_{a \in Mess - \mset{b_1 \dots b_k}}\mset{n \cdot q_{a}} \\
& + \sum_{a \in \mset{b_1 \dots b_k}}\mset{(n-1) \cdot q_{a}} 
\end{align*}
	Then, we continue the run in the following way:
	\[
	\mconf_{k+1} \mtransup{(s_0, !m_2, s_1)} \mconf_{k+2} \mtransrdv{(q_{c_1}, !c_i, q'_{c_1})} \dots \mtransup{(q_{c_\ell}, !c_\ell, q'_{c_\ell})} \mconf_{k+ \ell +2} 
	\]
	If there is no process on a state in $"\staterec{m_2}"$ then $\mconf_{k+1} \mtransse{(s_0, !m_2, s_1)} \mconf_{k+2}$, and  otherwise $\mconf_{k+1} \mtransrdv{(s_0, !m_2, s_1)} \mconf_{k+2}$. In any case, the rendez-vous is not answered by a process on state $q_1$ because $m_2 \notin \Rec{q_1}$.
	Furthermore, each rendez-vous with letter $c_i$ is made with transitions $(q_{c_i}, !c_i, q'_{c_i})$ and $(s_i, ?c_i, s_{i+1})$. Hence, 
	\begin{align*}
		\mconf_{k+\ell+2} \geq &  \mset{ q_2, q_1, (n-1) \cdot p_0, (n-1) \cdot s_0, n \cdot r_0} + \sum_{a \in \mset{d_1 \dots d_k}}\mset{n \cdot q_{a}} \\
		& + \sum_{a \in Mess - \mset{d_1 \dots d_k}}\mset{(n-1) \cdot q_{a}} 
	\end{align*}
	From $\mconf_{k+\ell +2}$ let the following run be: 
	\[
	\mconf_{k+\ell +2} \mtransrdv{(r_0, !m_3, r_1)} \mconf_{k+\ell +3} \mtransrdv{(q_{d_1}, !d_1, q'_{d_1})} \dots \mtransrdv{(q_{d_j}, !d_j, q'_{d_j})} \mconf_{k +\ell + j +3}
	\]
	where the rendez-vous on letter $m_3$ is made with transitions $(r_0, !m_3, r_1)$ and $(q_2, ?m_3, q_2')$ (this transition exists as $m_3 \in \Rec{q_2}$). Each rendez-vous on $d_i$ is made with transitions $(q_{d_i}, !d_i, q'_{d_i})$ and $(r_i, ?d_i, r_{i+1})$. Hence, the configuration $\mconf_{k+ \ell +j+3}$ is such that:
	\begin{align*}
		\mconf_{k+\ell+2} \geq &  \mset{ q_3, q_1, (n-1) \cdot p_0, (n-1) \cdot s_0, (n-1) \cdot r_0} + \sum_{a \in Mess }\mset{(n-1) \cdot q_{a}} 
	\end{align*}
	 Then from $\mconf_{k+\ell +j +3}$: $\mconf_{k+\ell + j +3} \mtransrdv{(p_0, !m_1, p_1)} \mconf_{k+\ell + j +4}$ where the rendez-vous is made with transitions $(p_0, !m_1, p_1)$ and $(q_3, ?m_1, q'_3)$ (this transition exists as $m_1 \in \Rec{q_3}$). By repeating $n-1$ times the run from configuration $\mconf_q$ (without the first sending of $m_1$) from  $\mconf_{k+\ell + j +4}$, we reach a configuration $\mconf'_q$ such that $\mconf'_q(q_1) \geq n$.
\end{itemize}

Hence, for all $n \in \mathbb{N}$, for all $q \in S'$, there exists $\mconf_q \in \Interp{\gamma}$, such that $\mconf_q\mtransup{}\mconf'_q$ and $\mconf'_q(q) \geq n$. From Lemma \ref{lem:consistent-reach}, there exists $\mconf'_n$ and $\mconf_n \in \Interp{\gamma}$ such that $\mconf_n \mtrans^\ast \mconf'_n$ and for all $q \in S'$, $\mconf_n(q) \geq n$.
\end{proof}